\documentclass[11pt]{article}
\usepackage{fullpage}
\usepackage{amsmath,amsfonts,amsthm}
\usepackage{epic,eepic}
\usepackage{hyperref}

\newtheorem{theorem}{Theorem}

\newtheorem{lemma}[theorem]{Lemma}
\newtheorem{corollary}[theorem]{Corollary}
\newtheorem{definition}[theorem]{Definition}
\newtheorem{protocol}[theorem]{Protocol}

\newcommand{\lp}{\left}
\newcommand{\rp}{\right}
\newcommand{\be}{\begin{eqnarray}}
\newcommand{\ee}{\end{eqnarray}}
\newcommand{\beq}{\begin{equation}}
\newcommand{\eeq}{\end{equation}}
\newcommand{\ba}{\begin{array}}
\newcommand{\ea}{\end{array}}

\newcommand{\braket}[2]{{\langle{#1}|{#2}\rangle}}
\newcommand{\ket}[1]{{|{#1}\rangle}}
\newcommand{\bra}[1]{{\langle{#1}|}}

\newcommand{\mypmatrix}[1]{\begin{pmatrix}#1\end{pmatrix}}

\DeclareMathOperator{\Tr}{Tr}
\DeclareMathOperator{\Rot}{Rot}
\DeclareMathOperator{\vspan}{span}
\DeclareMathOperator{\Pos}{Pos}
\DeclareMathOperator{\dist}{dist}
\DeclareMathOperator{\eig}{eig}
\DeclareMathOperator{\Prob}{Prob}
\DeclareMathOperator{\closure}{closure}

\newcommand{\HA}{{\mathcal{A}}}
\newcommand{\HB}{{\mathcal{B}}}
\newcommand{\HAP}{{\mathcal{A'}}}

\newcommand{\HBP}{{\mathcal{B'}}}
\newcommand{\HAB}{{\mathcal{\bar A}}}

\newcommand{\HM}{{\mathcal{M}}}
\newcommand{\HH}{{\mathcal{H}}}
\newcommand{\HHP}{{\mathcal{H'}}}
\newcommand{\HHB}{{\mathcal{\bar{H}}}}

\newcommand{\R}{\mathbb{R}}
\newcommand{\C}{\mathbb{C}}
\newcommand{\Z}{\mathbb{Z}}

\begin{document}


\title{Quantum weak coin flipping with arbitrarily small bias}

\author{Carlos Mochon\thanks{Perimeter Institute for Theoretical Physics,
cmochon@perimeterinstitute.ca}}

\date{November 26, 2007}

\maketitle

\begin{abstract}
``God does not play dice. He flips coins instead.'' And though for some
reason He has denied us quantum bit commitment. And though for some reason
he has even denied us strong coin flipping. He has, in His infinite mercy,
granted us quantum weak coin flipping so that we too may flip coins.

Instructions for the flipping of coins are contained herein. But be warned!
Only those who have mastered Kitaev's formalism relating coin flipping and
operator monotone functions may succeed. For those foolhardy enough to even
try, a complete tutorial is included.
\end{abstract}

\setcounter{tocdepth}{2}
{\small \tableofcontents}

\newpage
\section{Introduction}

It is time again for a sacrifice to the gods. Alice and Bob are highly
pious and would both like the honor of being the victim. Flipping a coin to
choose among them allows the gods to pick the worthier candidate for this
life changing experience. To keep the unworthy candidate from desecrating
the winner, the coin flip must be carried out at a distance and in a manner
that prevents cheating. Very roughly speaking, this is the problem known as
\textit{coin flipping by telephone} \cite{Blum}. Quantum coin flipping is a
variant of the problem where the participants are allowed to communicate
using quantum information.

Why is quantum coin flipping interesting/important/useful? First of all, it
is conceivably possible that someday somewhere someone will want to
determine something by flipping a coin with a faraway partner, and that for
some reason both will have access to quantum computers. The location of QIP
2050 could very well be determined in this fashion.

Secondly, coin flipping belongs to a class of cryptographic protocols known
as secure two-party computations. These arise naturally when two people
wish to collaborate but don't completely trust one another. Sadly, the
impossibility of quantum bit commitment \cite{May96,Lo:1998pn} shows that
quantum information is incapable of solving many of the problems in this
area. On the other hand, the possibility of quantum weak coin flipping
shows that quantum information may yet have untapped potential. Among the
most promising open areas is secure computation with cheat detection
\cite{Aharonov00,Hardy99} which may be better explored with the techniques
in this paper. At a minimum, the standard implementation of bit commitment
with cheat detection uses quantum weak coin flipping as a subroutine, so
improvements in the latter offer (modest) improvements in the former.

Finally, coin flipping is interesting because it appears to be hard, at
least relative to other cryptographic tasks such as key distribution
\cite{Wie83,BB84}. Of course, it is not our intent to belittle 
the discovery of key distribution, whose authors had to invent many of the
foundations of quantum information along the way. But a savvy student
today, familiar with the field of quantum information, would likely have no
trouble in constructing a key distribution protocol. Most reasonable
protocols appear to work. Not so with coin flipping, where most obvious
protocols appear to fail.

A cynical reader may argue that hardness is relative and may simply be a
consequence of having formulated the problem in the wrong language. But
this is exactly our third point: that the difficulty of coin flipping is
really an opportunity to develop a formalism in which such problems are
(relatively) easily solvable.

This new formalism, which is the cornerstone of the protocols in this
paper, was developed by Kitaev \cite{Kit04} and can be used to relate
coin flipping (and many other quantum games) to the theory of convex cones
and operator monotone functions. From this perspective, the value of the
present coin flipping result is that it provides the first demonstration
of the power of Kitaev's formalism.

We note that the formalism relating coin-flipping and operator monotone
functions is an extension of Kitaev's original formalism which was used in
proving a lower bound on strong coin flipping \cite{Kitaev}. When we need
to distinguish them, we shall refer to them respectively as Kitaev's second
and first coin flipping formalisms.

We will delay the formal definition of coin flipping to
Section~\ref{sec:def} and the history and prior work on the problem to
Section~\ref{sec:hist}. Instead, we shall give below an informal
description of the new formalism. We shall focus more on what the finished
formalism looks like rather than on how to relate it to the more
traditional notions of quantum states and unitaries (a topic which will be
covered at length later in the paper).

In its simplest form Kitaev's formalism can be described as a sequence of
configurations, each of which consists of a few marked points on the
plane. The points are restricted to the closure of the first quadrant
(i.e., have non-negative coordinates) and each point carries a positive
weight, which we call a probability.

Two successive configurations can only differ by points on a single
vertical or horizontal line. The rule is that the total probability on the
line must be conserved (though the total number of points can change) and
that for every $\lambda\in(0,\infty)$ we must satisfy
\be
\sum_z \frac{\lambda z}{\lambda + z}p_z\leq \sum_{z'} 
\frac{\lambda z'}{\lambda + z'}p_{z'},
\label{eq:constraint}
\ee
\noindent
where the left hand side is a sum over points before the transition and the
right hand side is a sum over points after the transition. The variable $z$
is respectively the $x$ coordinate for transitions occurring on a
horizontal line or the $y$ coordinate for transitions occurring on a
vertical line.  The numbers $p_z$ are just the probabilities associated to
each point. An example of such a sequence is given in Fig.~\ref{fig:intro}.

\begin{figure}[tb]
\begin{center}
\unitlength = 50pt
\begin{picture}(2.5,2.3)(-0.5,-0.5)
\put(0,0){\line(1,0){1.7}}
\put(0,0){\line(0,1){1.7}}
\put(1,-0.1){\makebox(0,0)[t]{$1$}}
\put(-0.1,1){\makebox(0,0)[r]{$1$}}
\put(1.15,0.15){\makebox(0,0)[b]{$\frac{1}{2}$}}
\put(0.15,1.2){\makebox(0,0)[l]{$\frac{1}{2}$}}
\thicklines
\put(1,0){\makebox(0,0){$\bullet$}}
\put(0,1){\makebox(0,0){$\bullet$}}
\end{picture}
\qquad
\begin{picture}(2.5,2.3)(-0.5,-0.5)
\put(0,0){\line(1,0){1.7}}
\put(0,0){\line(0,1){1.7}}
\put(1,0){\line(0,-1){0.05}}
\put(1,-0.1){\makebox(0,0)[t]{$1$}}
\put(-0.1,1){\makebox(0,0)[r]{$1$}}
\put(1.15,1.2){\makebox(0,0)[l]{$\frac{1}{2}$}}
\put(0.15,1.2){\makebox(0,0)[l]{$\frac{1}{2}$}}
\thicklines
\put(1,1){\makebox(0,0){$\bullet$}}
\put(0,1){\makebox(0,0){$\bullet$}}
\end{picture}
\qquad
\begin{picture}(2.5,2.3)(-0.5,-0.5)
\put(0,0){\line(1,0){1.7}}
\put(0,0){\line(0,1){1.7}}
\put(0.5,0){\line(0,-1){0.05}}
\put(0,1){\line(-1,0){0.05}}
\put(0.5,-0.1){\makebox(0,0)[t]{$1/2$}}
\put(-0.1,1){\makebox(0,0)[r]{$1$}}
\put(0.65,1.15){\makebox(0,0)[l]{$1$}}
\thicklines
\put(0.5,1){\makebox(0,0){$\bullet$}}
\end{picture}
\caption{A point game sequence with three configurations. 
Numbers outside the axes label location and numbers inside the axes label
probability. The sequence corresponds to a protocol with $P_A^*=1$ and
$P_B^*=1/2$, that is, a protocol where Alice flips a coin and announces the
outcome.}
\label{fig:intro}
\end{center}
\end{figure}
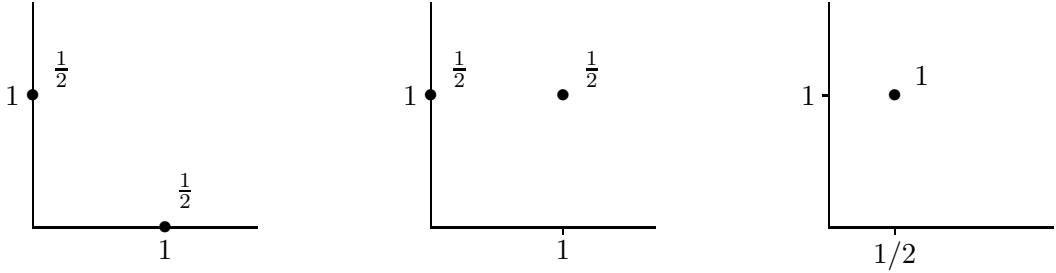

The boundary conditions of the sequence are as follows: The stating
configuration always contains two points, each carrying probability one
half, with one point located at $x=1$ and $y=0$ and the other point at
$x=0$ and $y=1$. The final configuration must contain a single point, which
carries unit probability (as required by conservation of probability).

Each of these sequences can be translated into coin-flipping protocols such
that the amount of cheating allowed is bounded by the location of the final
point. In particular, if the final point is located at $(x,y)$ then the
resulting protocol will satisfy $P_A^*\leq y$ and $P_B^*\leq x$, and hence
the bias is bounded by $\max(x,y)-1/2$.

We call these sequences ``point games'' and they are completely equivalent
to standard protocols described by unitaries. There exists constructive
mappings from point games to standard protocols and vice versa. The optimal
coin-flipping protocol can be constructed and tightly bounded by a point
game. Hence, rather than searching for optimal protocols, one can
equivalently search for optimal point games. These point games are
formalized in Section~\ref{sec:Kit} and examples are given in
Section~\ref{sec:examples}.

The configurations above are roughly related to standard semidefinite
programing objects as follows: The $x$ coordinates are the eigenvalues of
the dual SDP operators on Alice's Hilbert space, the $y$ coordinates are
eigenvalues of the dual SDP operators on Bob's Hilbert space, and the
weights are the probabilities assigned by the honest state to each of these
eigenspaces.

The obscure condition of Eq.~(\ref{eq:constraint}) can best be understood
if we describe the points on the line before and after the transition by
functions $p(z),p'(z):[0,\infty)\rightarrow[0,\infty)$ with finite support.
We then are essentially requiring that $p'(z)-p(z)$ belong to the cone dual
to the set of operator monotone functions with domain $[0,\infty)$. For
those unfamiliar with operator monotone functions, their definition and a
few properties are discussed later in the paper.

We note an unusual convention that was used above and throughout most of
the paper: the description of point games follows a \textbf{reverse time
convention} where the final measurement occurs at $t=0$ and the initial
state preparation occurs at $t=n>0$. The motivation for this will become
clear as the formalism is developed.

Kitaev further simplified his formalism so that an entire point game can be
described by a single pair of functions $h(x,y)$ and $v(x,y)$ that take
real values and have finite support. The main constraint is that on every
horizontal line of $h(x,y)$ and every vertical line of $v(x,y)$, the sum of
the weights must be zero, and the weighted average of $\frac{\lambda
z}{\lambda + z}$ must be non-negative for every
$\lambda\in(0,\infty)$. Furthermore, $h(x,y)+v(x,y)$ must be zero
everywhere except at three points: $(1,0)$ and $(0,1)$ where it has value
$-1/2$, and a third point $(x,y)$, where it has value $1$, and which is the
equivalent of the final point of the original point games. This variant of
point games is described in Section~\ref{sec:Kit2}.

A simple example can be constructed using Fig.~\ref{fig:introlad16}. The
labeled points with outgoing horizontal arrows appear in $h(x,y)$ with
negative sign, whereas those with incoming horizontal arrows appear in
$h(x,y)$ with positive sign. Similarly for $v(x,y)$ and the vertical
arrows. The final point at $(\frac{2}{3},\frac{2}{3})$ appears in both
functions with positive coefficient and magnitude $1/2$. That means that
the point game corresponds to a protocol with $P_A^*=P_B^*=2/3$, or bias
$1/6$, and is a variant of the author's previous best protocol
\cite{me2005}.

\begin{figure}[tb]
\begin{center}
\unitlength = 30pt
\begin{picture}(11,10)(-1,-1)
\put(0,0){\line(1,0){9}}
\put(0,0){\line(0,1){9}}
\put(2,0){\line(0,-1){0.1}}
\put(4,0){\line(0,-1){0.1}}
\put(5,0){\line(0,-1){0.1}}
\put(6,0){\line(0,-1){0.1}}
\put(7,0){\line(0,-1){0.1}}
\put(8,0){\line(0,-1){0.1}}
\put(0,2){\line(-1,0){0.1}}
\put(0,4){\line(-1,0){0.1}}
\put(0,5){\line(-1,0){0.1}}
\put(0,6){\line(-1,0){0.1}}
\put(0,7){\line(-1,0){0.1}}
\put(0,8){\line(-1,0){0.1}}
\put(2,-0.2){\makebox(0,0)[t]{$\frac{2}{3}$}}
\put(3,-0.2){\makebox(0,0)[t]{$1$}}
\put(4,-0.2){\makebox(0,0)[t]{$\frac{4}{3}$}}
\put(5,-0.2){\makebox(0,0)[t]{$\frac{5}{3}$}}
\put(6,-0.2){\makebox(0,0)[t]{$2$}}
\put(7,-0.2){\makebox(0,0)[t]{$\frac{7}{3}$}}
\put(8,-0.2){\makebox(0,0)[t]{$\frac{8}{3}$}}
\put(-0.2,2){\makebox(0,0)[r]{$\frac{2}{3}$}}
\put(-0.2,3){\makebox(0,0)[r]{$1$}}
\put(-0.2,4){\makebox(0,0)[r]{$\frac{4}{3}$}}
\put(-0.2,5){\makebox(0,0)[r]{$\frac{5}{3}$}}
\put(-0.2,6){\makebox(0,0)[r]{$2$}}
\put(-0.2,7){\makebox(0,0)[r]{$\frac{7}{3}$}}
\put(-0.2,8){\makebox(0,0)[r]{$\frac{8}{3}$}}
\put(3.2,0.2){\makebox(0,0)[bl]{$\frac{1}{2}$}}
\put(0.2,3.2){\makebox(0,0)[bl]{$\frac{1}{2}$}}
\put(1.95,1.95){\makebox(0,0)[tr]{$\frac{1}{2}+\frac{1}{2}$}}
\put(4.1,1.9){\makebox(0,0)[tl]{$1$}}
\put(5.1,2.9){\makebox(0,0)[tl]{$1$}}
\put(6.1,3.9){\makebox(0,0)[tl]{$1$}}
\put(7.1,4.9){\makebox(0,0)[tl]{$1$}}
\put(1.9,4.1){\makebox(0,0)[br]{$1$}}
\put(2.9,5.1){\makebox(0,0)[br]{$1$}}
\put(3.9,6.1){\makebox(0,0)[br]{$1$}}
\put(4.9,7.1){\makebox(0,0)[br]{$1$}}
\put(2.88,2.1){\makebox(0,0)[br]{$\frac{3}{2}$}}
\put(2.1,2.88){\makebox(0,0)[tl]{$\frac{3}{2}$}}
\put(3.88,3.1){\makebox(0,0)[br]{$2$}}
\put(3.1,3.88){\makebox(0,0)[tl]{$2$}}
\put(4.88,4.1){\makebox(0,0)[br]{$2$}}
\put(4.1,4.88){\makebox(0,0)[tl]{$2$}}
\put(5.88,5.1){\makebox(0,0)[br]{$2$}}
\put(5.1,5.88){\makebox(0,0)[tl]{$2$}}
\put(6.88,6.1){\makebox(0,0)[br]{$2$}}
\put(6.1,6.88){\makebox(0,0)[tl]{$2$}}
\thicklines
\put(3,0){\vector(0,1){2}}
\put(0,3){\vector(1,0){2}}
\put(3,2){\vector(-1,0){1}}
\put(2,3){\vector(0,-1){1}}
\put(3,2){\vector(1,0){1}}
\put(2,3){\vector(0,1){1}}
\put(4,3){\vector(-1,0){2}}
\put(3,4){\vector(0,-1){2}}
\put(2,4){\vector(1,0){1}}
\put(4,2){\vector(0,1){1}}
\put(4,3){\vector(1,0){1}}
\put(3,4){\vector(0,1){1}}
\put(5,4){\vector(-1,0){2}}
\put(4,5){\vector(0,-1){2}}
\put(3,5){\vector(1,0){1}}
\put(5,3){\vector(0,1){1}}
\put(5,4){\vector(1,0){1}}
\put(4,5){\vector(0,1){1}}
\put(6,5){\vector(-1,0){2}}
\put(5,6){\vector(0,-1){2}}
\put(4,6){\vector(1,0){1}}
\put(6,4){\vector(0,1){1}}
\put(6,5){\vector(1,0){1}}
\put(5,6){\vector(0,1){1}}
\put(7,6){\vector(-1,0){2}}
\put(6,7){\vector(0,-1){2}}
\put(5,7){\vector(1,0){1}}
\put(7,5){\vector(0,1){1}}
\put(7,6){\vector(1,0){1}}
\put(6,7){\vector(0,1){1}}
\put(7,7){\vector(-1,0){1}}
\put(7,7){\vector(0,-1){1}}
\put(7.6,7.6){\makebox(0,0){$\cdot$}}
\put(7.7,7.7){\makebox(0,0){$\cdot$}}
\put(7.8,7.8){\makebox(0,0){$\cdot$}}
\put(3,0){\makebox(0,0){$\bullet$}}
\put(0,3){\makebox(0,0){$\bullet$}}
\end{picture}
\caption{A coin-flipping protocol with bias $1/6$.}
\label{fig:introlad16}
\end{center}
\end{figure}

The power of Kitaev's formalism is evident from the previous example as a
complete protocol can be described by a single
picture. Section~\ref{sec:zero} discusses in detail how to build and
analyze such structures. Among the issues addressed are how to truncate the
above infinite ladder so that the resulting figure has only a finite number
of points as required by our description of Kitaev's formalism.

To achieve zero bias in coin flipping, one can use similar constructions,
but with more complicated ladders heading off to infinity. In particular,
for every integer $k\geq 0$ we will build a protocol with
\be
P_A^*=P_B^*=\frac{k+1}{2k+1}
\ee
\noindent
(technically, for each $k$ we will have a family of protocols that will
converge to the above values). The case $k=0$ allows both players to
maximally cheat, the case $k=1$ is the author's bias $1/6$ protocol, and
the limit $k\rightarrow \infty$ achieves arbitrarily small bias. The
details of this construction can also be found in Section~\ref{sec:zero}.

All the new protocols are formulated in the language of Kitaev's
formalism. Sadly, mechanically transforming these protocols back into the
language of unitaries, while possible, does not lead to particularly simple
or efficient protocols (i.e., in terms of laboratory resources). Finding
easy to implement protocols with a small bias remains an interesting open
problem. A number of other open problems can be found at the end of
Section~\ref{sec:end}.

A first stab at finding good easy to implement protocols is given by
Appendix~\ref{sec:ddb}. The section transforms the author's original bias
$1/6$ protocol (which uses a number of qubits linear in the number of
messages) into a new form that uses constant space. In fact, the total
space needed is one qutrit for each of Alice and Bob, and one qubit used to
send messages.

The key idea is to use early measurements to prune states that are known to
be illegal. While in a theoretical sense measurements can always be delayed
to the last step, their frequent use can provide practical simplifications
(as is well know in key distribution). In fact, all the early measurements
are of the flying qubit in the computational basis.

The protocol in Appendix~\ref{sec:ddb} is described in the standard
language of unitaries, and an analysis is sketched using Kitaev's first
formalism. As a bonus, the resulting protocol is related to the ancient and
most holy game of Dip-Dip-Boom, also described therein.

Appendix~\ref{sec:strong} proves strong duality for coin flipping, which is
an important lemma needed for both Kitaev's first and second
formalisms. While strong duality does not hold in general semidefinite
programs, it does in most, and there exists a number of lemmas that provide
sufficient conditions. Unfortunately, some of the simplest lemmas do not
directly apply to coin-flipping. Instead, the appendix directly proves
strong duality using simple arguments from Euclidean geometry. While all
the ideas in this section are taken from standard textbooks, the
presentation is still somewhat clever and novel.

Finally, Appendix~\ref{sec:f2m} proves another mathematical lemma needed
for Kitaev's second formalism. It is the key step needed to turn the
functions that underlie the point games back into matrices out of which
states and unitaries can be constructed. Though some of the ideas are
potentially novel, mostly it deals with standard technical issues from the
theory of matrices.

As a final goody, Section~\ref{sec:cheat} includes a brief discussion and
example of how to extend the formalism to include cheat detection. Of
course, because weak coin-flipping can be achieved with arbitrarily small
bias, adding in cheat detection isn't particularly useful. However, similar
techniques may prove helpful in studying cheat detection for strong coin
flipping and other secure computation problems.

\textit{Author's note:} Sections~\ref{sec:Kit} and~\ref{sec:Kit2} are based
on my recollection of a couple of discussions with Kitaev and a subsequent
group meeting he gave (of which I sadly kept no written record). As I have
had to reconstruct some of the details, and as I have strived to move the
discussion to finite dimensional spaces, some of Kitaev's original elegance
has been replaced by a more pedantically constructive (and hopefully
pedagogical) approach. I claim no ownership of the main ideas in these
sections, though am happy to accept the blame for any errors in my write
up. You should also know that all the terms such as UBP, ``point game,''
and ``valid transitions'' are my own crazy invention, and are unlikely to
be familiar to those who have leaned Kitaev's formalism from other sources.

At this point, those familiar with the definition and history of coin
flipping may wish to skip ahead to Section~\ref{sec:Kit}. Good luck!

\subsection{\label{sec:def}Coin flipping defined}

Coin flipping is a formalization of the notion of flipping or tossing a
coin under the constraints that the participants are mutually distrustful
and far apart.

The two players involved in coin flipping, traditionally called Alice and
Bob, must agree on a single random bit which represents the outcome of the
coin flip. As Alice and Bob do not trust each other, nor anyone else, they
each want a protocol that prevents the other player from
cheating. Furthermore, because they are far apart, the protocol must be
implementable using only interaction over a communication device such as a
telephone.  The problem is known in the classical literature as
``coin flipping by telephone'' and was first posed by Manuel Blum in 1981
\cite{Blum}.

There are two variants of coin flipping. In the first variant, called weak
coin flipping, Alice and Bob each have a priori a desired coin outcome.
The outcomes can be labeled as ``Alice wins'' and ``Bob wins,'' and we
do not care if the players cheat in order to increase their own probability
of losing. In the second variant of coin flipping, called strong coin
flipping, there are no a priori desired outcomes and we wish to prevent
either player from biasing the coin in either direction. 

Obviously, strong coin flipping is at least as hard as weak coin flipping
and in general it is harder. However, this paper is mainly concerned with
weak coin flipping which we often simply refer to as ``coin flipping''.

To be more precise in our definition, weak coin flipping is a two-party
communication protocol that begins with a completely uncorrelated state and
ends with each of the participants outputting a single bit. We say that
Alice wins on outcome 0 and Bob wins on outcome 1. The requirements are:
\begin{enumerate}
\item When both players are honest, Alice's output is uniformly random and
equal to Bob's output.
\item If Alice is honest but Bob deviates from the protocol, then no matter
what Bob does, the probability that Alice outputs one (i.e., Bob wins) is
no greater than $P_B^*$.
\item Similarly, if Bob is honest but Alice deviates from the protocol,
then the maximum probability for Bob to declare Alice the winner is $P_A^*$.
\end{enumerate}

\noindent
The parameters $P_A^*$ and $P_B^*$ define the protocol. Ideally we want
$P_A^*=P_B^*=1/2$. Unfortunately, this is not always possible. We therefore
introduce the bias $\max(P_A^*,P_B^*)-1/2$ as a measure of the security of
the protocol. Our goal is to find a protocol with the smallest bias
possible.

Note that the protocol places no restrictions on the output of a cheating
player, as these are impossible to enforce. In particular, when one player
is cheating the outputs do not have to agree, and when both players are
cheating the protocol is not required to satisfy any properties. This also
means that if an honest player ever detects that their opponent has
deviated from the protocol (i.e., the other player stops sending messages
or sends messages of the wrong format) then the honest player can simply
declare victory rather than aborting. This will be an implicit rule in all
our weak coin-flipping protocols.

Occasionally, it is worth extending the definition of coin flipping to case
where the output is not uniformly random even when both players are
honest. In such a case we denote by $P_A$ the honest probability for Alice
to win and by $P_B=1-P_A$ the honest probability for Bob to win.

\subsubsection{Communication model}

It is not hard to see, that in a classical world, and without any further
assumptions, at least one player can guarantee victory. For instance, if
one of the players were in charge of flipping the coin, the other player
would have no way of verifying via a telephone that the outcome of the coin
is the one reported by the first player.

Coin flipping can be achieved in a classical setting by adding in certain
computational assumptions \cite{Blum}. However, some of these assumptions
will no longer be true once quantum computers become available. Coin
flipping can also be achieved in a relativistic setting \cite{Ken99}
if Alice and Bob's laboratories are assumed to satisfy certain spatial
arrangements. However, these requirements may not be optimal for today's
on-the-go coin flippers.

In this paper we shall focus on the quantum setting, where Alice and Bob
each have a quantum computer with as much memory as needed, and are
connected by a noiseless quantum channel. They are each allowed to do
anything allowed by the laws of quantum mechanics other than directly
manipulate their opponents qubits. The resulting protocols will have
information theoretic security.

Although at the moment such a setting seems impractical, if ever
quantum computers are built and are as widely available as classical
computers are today, then quantum weak coin flipping may become viable.

\subsubsection{On the starting state}

The starting state of coin-flipping protocols is by definition completely
uncorrelated, which means that Alice and Bob initially share neither
classical randomness nor quantum entanglement (though they do share a
common description of the protocol).

There are two good reasons for this definition. First, it is easy to see
that given a known maximally entangled pair of qubits, Alice and Bob could
obtain a correlated bit without even using communication. But the same
result can be obtained when starting with a uniformly distributed shared
classical random bit. Such protocols are trivial, and certainly do not
require the power of quantum mechanics. The purpose of coin flipping,
though, is to create these correlations.

Still, at first glance it would appear that by starting with correlated
states we can put the acquisition of randomness, or equivalently the
interaction with a third party, in the distant past. In such a model Alice
and Bob would buy a set of correlated bits from their supermarket and then
used them when needed. The problem is that they can now figure out the
outcomes of the coin flips before committing to them. It is not hard to
imagine that a cheater would have the power to order the sequence of events
that require a coin flip so that he wins on the important ones and loses
the less important ones. We would now have to worry about protecting the
ordering of events and that is a completely different problem.

A good one-shot coin-flipping protocol should not allow the players to
predict the outcome of the coin flip before the protocol has begun, and
enforcing this is the second reason that we require an uncorrelated stating
state.

It might be interesting to explore what happens when the no-correlation
requirement is weakened to a no-prior-knowledge-of-outcome requirement, but
that is beyond the scope of this paper.

\subsubsection{On the security guarantees}

The security model of coin flipping divides the universe into three parts:
Alice's laboratory, Bob's laboratory and the rest of the universe. We
assume that Alice and Bob each have exclusive and complete control over
their laboratories. Other than as a conduit for information between Alice
and Bob, the rest of the universe will not be touched by honest
players. However, a dishonest player may take control of anything outside
their opponent's laboratory, including the communication channel.

The security of the protocols therefore depends on the inability of a
cheater to tamper with their opponent's laboratory. What does this mean?
Abstractly, it can be defined as
\begin{enumerate}
\item
All quantum superoperators that a cheater can
apply must act as the identity on the part of the Hilbert space that is
located inside the laboratory (including the message space when
appropriate).
\item
All operations of honest players are performed flawlessly and without
interference by the cheater.
\item
An honest player can verify that an incoming message has the right
dimension and can abort otherwise.
\end{enumerate}
\noindent
In practice, however, this translates into requirements such as
\begin{enumerate}
\item
The magnetic shielding on the laboratory is good enough to prevent your
opponent from affecting your qubits.
\item
The grad student operating your machinery cannot be bribed to apply the
wrong operations.
\item
A nanobot cannot enter your laboratory though the communication channel.
\end{enumerate}
\noindent
As usual, the fact that a protocol is secure does not mean that it will
protect against the preceding attacks. The purpose of the security analysis
is to prove that the only way to cheat is to attack an opponents
laboratory, thereby guaranteeing that one's security is as good as the
security of one's laboratory.

\subsubsection{On the restriction to unitary operations}

It is customary when studying coin flipping to consider only protocols that
involve unitary operations with a single measurement at the end. It is also
customary to only consider cheating strategies that can be implemented
using unitary operations. 

Nevertheless, any bounds that are derived under such conditions apply to
the most general case which includes players that can use measurements,
superoperators and classical randomness, and protocols that employ
extra classical channels.

The above follows from two separate lemmas, which roughly can be stated as:
\begin{enumerate}
\item Given any protocol $P$ in the most general setting (including
measurements, classical channels, etc.) that has a maximum
bias $\epsilon$ under the most general cheating strategy (including,
measurements, superoperators, etc.) then there exists a second protocol
$P'$ that is specified using only unitary operations with a single
pair of measurements at the end and that also has a maximum
bias $\epsilon$ under the most general cheating strategy (including,
measurements, superoperators, etc.).
\item Given any protocol $P$ specified using only unitary operations with a
single pair of measurements at the end, and given any cheating strategy for
$P$ (which may include measurements, superoperators, etc) that achieves a
bias $\epsilon$, we can find a second cheating strategy for $P$ that also
achieves a bias $\epsilon$ but can be implemented using only unitary
operations.
\end{enumerate}

\noindent
Unfortunately neither of the above statements will be proven here, and the
proofs for the above statements are distributed among a number of published
papers, but good stating points are \cite{Lo:1998pn} and \cite{May96}.

The first statement implies that we need only consider protocols where
the honest actions can be described as a sequence of unitaries. The result
is important for proofs of lower-bounds on the bias, but is not needed for
the main result of this paper.

The reduction of the first statement also applies to protocols that
potentially have an infinite number of rounds, such as rock-paper-scissors,
where measurements are carried out at intermediate steps to determine if a
winner can be declared or if the protocol must go on. Such protocols are
dealt with by proving that there exist truncations that approximate the
original protocol arbitrarily well. In such a case, though, the resulting
bias will only come arbitrarily close to the original bias.

The second statement implies that in our search for the optimal cheating
strategy (or equivalently, in our attempts to upper bound the bias) we need
only consider unitary cheaters. Note that measurements are not even needed
in the final stage as we do not care about the output of dishonest players.

The basic idea of the proof of the second statement is that any allowed
quantum mechanical operation can be expressed as a unitary followed by the
discarding of some Hilbert subspace. If the cheater carries out his
strategy without ever discarding any such subspaces he will not reduce his
probability of victory and will only need to use unitary operations.

The second statement also holds when the honest protocol includes certain
projective measurements such as those used in this paper. In fact, we will
sketch a proof for this second statement in Section~\ref{sec:primal} as we
translate coin flipping into the language of semidefinite programing.

\subsection{\label{sec:hist}A brief history muddled by hindsight}

Roughly speaking, cryptography is composed of two fundamental problems: two
honest players try to complete a task without being disrupted by a third
malicious party (i.e., encrypted communication), and two mutually
distrustful players try to cooperate in a way that prevents the opposing
player from cheating, effectively simulating a trusted third party (i.e.,
choosing a common meeting time while keeping their schedules private).
The second case is commonly known as two-party secure computation.

In the mid 1980s, quantum information had a resounding success in the first
category by enabling key distribution with information theoretic security
\cite{Wie83,BB84}. Optimism was high, and it seemed like quantum
information would be able to solve all problems in the second category as
well. But after many failed attempts at producing protocols for a task
known as bit commitment, it was finally proven by Mayers, Lo and Chau
\cite{May96, Lo:1998pn} that secure quantum bit commitment was impossible.

One of the reasons people had focused on bit commitment is that it is a
powerful primitive from which all other two-party secure computation
protocols can be constructed \cite{Yao95}. Its impossibility means that all
other universal primitives for two-party secure computation must also be
unrealizable using quantum information \cite{Lo97}.

We note that when we speak of possible and impossible with regards to
quantum information we always mean with information theoretic security
(i.e., without placing bounds on the computational capacity of the
adversary). Classically all two-party secure-computation tasks can be
realized under certain complexity assumptions \cite{Yao82} but become
impossible if we demand information theoretic security. Surprisingly, most
multiparty secure-computations tasks can be done classically with
information theoretic security so long as all parties share private
pairwise communication channels and the number of cheating players is
bounded by some constant fraction (dependent on the exact model) of the
total number of players \cite{CCD88,BGW88,RB89}. Similar results hold
in the multiparty quantum case \cite{GGS02}.

Given the impossibility of quantum bit commitment, and the fact that most
multiparty problems can already be solved with classical information, the
new goal in the late 1990s became to find any two-party task that is
modestly interesting and can be realized with information theoretic
security using quantum information.

One of the problems that the literature converged on \cite{Goldenberg:1998bx}
was a quantum version of the problem of flipping a coin over the
telephone \cite{Blum}. Initially the focus was on strong coin flipping and
Ambainis \cite{Ambainis2002} and Spekkens and Rudolph \cite{Spekkens2001}
independently proposed protocols that achieve a bias of
$1/4$. Unfortunately, shortly thereafter Kitaev \cite{Kitaev} (see also
\cite{Ambainis2003}) proved a lower bound of $1/\sqrt{2}-1/2$ on the bias.

Research continued on weak coin flipping \cite{ker-nayak, Spekkens2002,
Spekkens2003} and the best known bias prior to the author's own work was
$1/\sqrt{2}-1/2\simeq 0.207$ by Spekkens and Rudolph \cite{Spekkens2002}.
The best lower bound was proven by Ambainis \cite{Ambainis2002} and states
that the number of messages must grow at least as $\Omega(\log \log
\frac{1}{\epsilon})$. In particular, it implies that no protocol with a
fixed number of messages can achieve an arbitrarily small bias.

At this point most known protocols used at most a few rounds of
communication. The first non-trivial many-round coin-flipping protocol was
published in \cite{me2004} by the author and achieved a bias of $0.192$ in
the limit of arbitrarily many messages. In subsequent work \cite{me2005} it
was shown that this protocol (and many of the good protocols known at the
time) were part of a large family of quantized classical public-coin
protocols. Furthermore, an analytic expression was given for the bias of
each protocol in the family, and the optimal protocol for each number of
messages was identified. Sadly, the best bias that can be achieved in this
family is $1/6$, and this only in the limit of arbitrarily many messages (a
new formulation of this bias $1/6$ protocol can be found in
Appendix~\ref{sec:ddb}, wherein we use early measurements to reduce the
space needed to run the protocol to a qutrit per player and a qubit for
messages).

Independently, Kitaev created a new formalism for studying two player
adversarial games such as coin flipping
\cite{Kit04}, which built on his earlier work \cite{Kitaev}. The formalism
describes the set of possible protocols as the dual to the cone of two
variables functions that are independently operator monotone in each
variable. Though the result was never published, we include in
Sections~\ref{sec:Kit} and \ref{sec:Kit2} a description of the
formalism. This formalism, which we shall refer to as Kitaev's second
coin-flipping formalism, is the crucial idea behind the results in the
present paper. A different extension of Kitaev's original formalism was
also proposed by Gutoski and Watrous \cite{GW07}.

Looking beyond coin flipping, there is the intriguing possibility that we
can still achieve most of protocols of two-party secure computation if we
are willing to loosen our requirements: instead of requiring that cheating
be impossible, we require that a cheater be caught with some non-zero
probability. Quantum protocols that satisfy such requirements for bit
commitment have already been constructed \cite{Aharonov00, Hardy99} though
the amount of potential cheat detection is known to be bounded
\cite{me2003-2}.

The possibility of some interesting quantum two-party protocols with
information theoretic security, plus many more protocols built using cheat
detection, may mean that ultimately quantum information will fulfill its
potential in the area of secure computation. But more work needs to be done
in this direction, and we hope that the results and techniques of the
present paper will be helpful.

\section{\label{sec:Kit}Kitaev's second coin-flipping formalism}

The goal of this section is to describe Kitaev's formalism which relates
coin flipping to the dual cone of a certain set of operator monotone
functions \cite{Kit04}.

The first step in the construction involves formalizing the problem of
coin flipping and proving the existence of certain upper bound certificates
for $P_A^*$ and $P_B^*$. This is done in Section~\ref{sec:kit1} and we
refer to the result as Kitaev's first coin-flipping formalism as most of
the material was used in the construction of the lower bounds on strong
coin flipping \cite{Kitaev}.

The next step, carried out in Section~\ref{sec:UBP}, involves using these
certificates to change the maximization over cheating strategies to a
minimization over certificates, and overall to transform the problem of
finding the best coin-flipping protocol into a minimization over objects we
call upper-bounded protocols (UBPs).

The third step involves stripping away most of the irrelevant information
of the UBPs to end up with a sequence of points moving around in the
plane. These ``point games'' are the main object of study of Kitaev's
second formalism. They come in two varieties: time dependent (which are
studied in Section~\ref{sec:TDPG}) and time independent (whose description
is delayed to Section~\ref{sec:TIPG}).

\subsection{\label{sec:kit1}Kitaev's first coin-flipping formalism}

The first goal is to formalize coin-flipping protocols using the standard
quantum communication model of a sequence of unitaries with measurements
delayed to the end.

\begin{definition}
A \textbf{coin-flipping protocol} consists of the following data
\begin{itemize}
\item $\HA$, $\HM$ and $\HB$, three finite-dimensional Hilbert spaces 
corresponding to Alice's qubits, the message channel and Bob's qubits 
respectively. We assume that each Hilbert space is equipped with a 
orthonormal basis of the form $\ket{0},\ket{1},\ket{2},\dots$ called 
the computational basis.
\item $n$, a positive integer describing the number of messages.\\
For simplicity we shall assume $n$ is even.
\item A tensor product initial state: 
$\ket{\psi_0} =
\ket{\psi_{A,0}}\otimes\ket{\psi_{M,0}}\otimes\ket{\psi_{B,0}}
\in\HA\otimes\HM\otimes\HB$.
\item A set of unitaries $U_1,\dots,U_n$ on $\HA\otimes\HM\otimes\HB$
of the form
\be
U_i = \begin{cases}
U_{A,i}\otimes I_\HB & \text{for $i$ odd,}\\
I_\HA\otimes U_{B,i} & \text{for $i$ even,}
\end{cases}
\ee
\noindent
where $U_{A,i}$ acts on $\HA\otimes\HM$ and $U_{B,i}$ acts on $\HM\otimes\HB$.
\item $\lp\{\Pi_{A,0},\Pi_{A,1}\rp\}$, a POVM on $\HA$.
\item $\lp\{\Pi_{B,0},\Pi_{B,1}\rp\}$, a POVM on $\HB$.
\end{itemize}
Furthermore, the above data must satisfy
\be
\Pi_{A,1} \otimes I_{\HM} \otimes \Pi_{B,0} \ket{\psi_n} =
\Pi_{A,0} \otimes I_{\HM} \otimes \Pi_{B,1} \ket{\psi_n} = 0,
\label{eq:povmreq}
\ee
\noindent
where $\ket{\psi_n} = U_{n} \cdots U_{1} \ket{\psi_0}$.
\end{definition}

Given the above data, the protocol is run as follows:
\begin{enumerate}
\item Alice starts with $\HA$ and Bob starts with $\HM\otimes\HB$.
They initialize their state to $\ket{\psi_0}$.
\item For $i=1$ to $n$:\\
If $i$ is odd Alice takes $\HM$ and applies $U_{A,i}$.\\
If $i$ is even Bob takes $\HM$ and applies $U_{B,i}$.
\item Alice measures $\HA$ with $\lp\{\Pi_{A,0},\Pi_{A,1}\rp\}$ 
and Bob measures $\HB$ with $\lp\{\Pi_{B,0},\Pi_{B,1}\rp\}$. 
They each output zero or one based on the outcome of the measurement.
\end{enumerate}

\noindent
When both Alice and Bob are honest, the above protocol starts off with the
state $\ket{\psi_0}$ and proceeds through the states
\be
\ket{\psi_i} = U_{i} \cdots U_{1} \ket{\psi_0}.
\label{eq:psik}
\ee
The final probabilities of winning are given by
\be
P_A &=& \lp | \Pi_{A,0} \otimes I_{\HM} \otimes \Pi_{B,0} \ket{\psi_n} \rp|^2 
= \Tr\lp[ \Pi_{B,0} \Tr_{\HA\otimes\HM} \ket{\psi_n}\bra{\psi_n} \rp],
\nonumber\\
P_B &=& \lp | \Pi_{A,1} \otimes I_{\HM} \otimes \Pi_{B,1} \ket{\psi_n} \rp|^2 
= \Tr\lp[ \Pi_{A,1} \Tr_{\HM\otimes\HB} \ket{\psi_n}\bra{\psi_n} \rp],
\label{eq:papb}
\ee
\noindent
where Eq.~(\ref{eq:povmreq}) guarantees the second
equalities above and the condition $P_A+P_B=1$. In general, we will also
want to impose $P_A=P_B=1/2$ to obtain a standard coin flip.

How many messages does the above protocol require? Traditionally, we have
one message after each unitary. We also need an initial message before the
first unitary so that Alice can get $\HM$. We will think of this as the
zeroth message. In total, we have $n+1$ messages. However, the first and
last message are somewhat odd: Alice never looks at the last message, so we
could have never sent it. Also, in principle, Alice could have started with
$\HM$ and initialized it herself, which would at most reduce Bob's cheating
power. So the whole protocol could be run with only $n-1$ messages.

However, the moments in time when the message qubits are flying between
Alice and Bob (all $n+1$ of them), mark particularly good times to examine
the state of our system. In particular, we are interest in the state of
$\HA$ and $\HB$ at these times which, when both players are honest, will be
\be
\sigma_{A,i} = \Tr_{\HM\otimes\HB}  \ket{\psi_i}\bra{\psi_i},
\qquad\qquad
\sigma_{B,i} = \Tr_{\HA\otimes\HM}  \ket{\psi_i}\bra{\psi_i},
\ee
\noindent
for $i=0,\dots,n$.

\subsubsection{\label{sec:primal}Primal SDP}

Now that we have formalized the protocol, we proceed with the formalization
of the optimization problem needed to find the maximum probabilities with
which the players can win by cheating. The resulting problems will be
semidefinite programs (SDPs).

We will study the case of Alice honest and Bob cheating (the other case
being nearly identical). As usual, we do not want to make any assumptions
about the operations that Bob does, or even the number of qubits that he
may be using, therefore we must focus entirely on the state of Alice's
qubits.

As Alice initializes her qubits independently from Bob, we know what their
state must be during the zeroth message:
\be
\rho_{A,0} = \ket{\psi_{A,0}}\bra{\psi_{A,0}}.
\label{eq:prim0}
\ee
Subsequently we shall lose track of their exact state, but we know by the
laws of quantum mechanics that they must satisfy certain requirements. The
simplest is that, since Bob cannot affect Alice's qubits, then during the
steps when Alice does nothing the state of the qubits cannot change:
\be
\rho_{A,i} = \rho_{A,i-1} \qquad\qquad \text{for $i$ even.}
\label{eq:primeven}
\ee
Note that this is true even if Bob performs a measurement as Alice will
not know the outcome, and therefore her mixed state description will still
be correct.

The more complicated case is the steps when Alice performs a unitary.
Let $\tilde \rho_{A,i}$ be the state of $\HA\otimes\HM$ immediately after
Alice receives the $i$th message (for $i$ even, of course). The laws of
quantum mechanics again require the consistency condition
\be
\Tr_\HM \tilde \rho_{A,i} = \rho_{A,i} \qquad\qquad \text{for $i$ even,}
\label{eq:primtilde}
\ee
\noindent
where $\tilde \rho_{A,i}$ is only restricted by the fact that Bob cannot
affect the state of $\HA$. Note also that the above equation holds valid
even if Bob uses his message to tell Alice the outcome of a previous
measurement.

Now when Alice applies her unitary and sends off $\HM$ she
will be left with the state
\be
\rho_{A,i} = \Tr_\HM\lp[ U_{A,i} \tilde\rho_{A,i-1} U_{A,i}^\dagger\rp]
\qquad\qquad\text{for $i$ odd.}
\label{eq:primodd}
\ee

Finally, Alice's output is determined entirely by the measurement of
$\rho_{A,n}$. In particular, Bob wins with probability
\be
P_{win} = \Tr\lp[ \Pi_{A,1} \rho_{A,n} \rp].
\ee

Now consider the maximization of the above quantity over density operators
$\rho_{A,0},\dots,\rho_{A,n}$ and
$\tilde\rho_{A,0},\dots,\tilde\rho_{A,n-2}$ subject to 
Eqs.~(\ref{eq:prim0},\ref{eq:primeven},\ref{eq:primtilde},\ref{eq:primodd}).
Because the optimal cheating strategy must satisfy the above conditions we
have
\be
P_B^* \leq \max \Tr\lp[ \Pi_{A,1} \rho_{A,n} \rp],
\ee
\noindent
where the maximum is taken subject to the above constraints.

The bound is also tight because any sequence of states consistent with the
above constraints can be achieved by Bob simply by maintaining the
purification of Alice's state. We sketch the proof: we inductively
construct a strategy for Bob that only uses unitaries so that the total
state will always be pure. Assume that Alice has $\rho_{A,i-1}$ (and the
total state is $\ket{\phi_{i-1}}$) and Bob wants to make her transition to
a given $\rho_{A,i}$ consistent with the above constraints.  If $i$ is even
this is trivial. If $i$ is odd, he must make sure to send the right message
so that Alice ends up with the appropriate $\tilde
\rho_{A,i-1}$. But let $\ket{\tilde \phi_{i-1}}$ be any 
purification of $\tilde \rho_{A,i-1}$ into $\HB$. Because the reduced
density operators on $\HA$ of both $\ket{\phi_{i-1}}$ and $\ket{\tilde
\phi_{i-1}}$ are the same, they are related by a unitary on $\HM\otimes\HB$
and by applying this unitary Bob will succeed in this step. By induction he
also succeeds in obtaining the entire sequence, as the base case for $i=0$
is trivial. We therefore have
\be
P_B^* = \max \Tr\lp[ \Pi_{A,1} \rho_{A,n} \rp].
\label{eq:primmax}
\ee
\noindent

\subsubsection{Dual SDP} 

In the last section we found a mathematical description for the problem of
computing $P_B^*$. Unfortunately, it is formulated as a maximization
problem whose solution is often difficult to find. It would be sufficient
for our purposes, though, to find an upper bound on $P_B^*$. Such upper
bounds can be constructed from the dual SDP.

In particular, in this section we will describe a set of simple-to-verify
certificates that prove upper bounds on $P_B^*$. These certificates are
known as dual feasible points.

The certificates will be a set of $n+1$ positive semidefinite operators
$Z_{A,0},\dots,Z_{A,n}$ on $\HA$ whose main property is
\be
\Tr[Z_{A,i-1} \rho_{A,i-1}] \geq \Tr[Z_{A,i} \rho_{A,i}]
\label{eq:dualprop}
\ee
for $i=1,\dots,n$ and for all $\rho_{A,0},\dots,\rho_{A,n}$ consistent
with the constraints of
Eqs.~(\ref{eq:prim0},\ref{eq:primeven},\ref{eq:primtilde},\ref{eq:primodd}).
Additionally, we require
\be
Z_{A,n}= \Pi_{A,1}.
\label{eq:zn}
\ee

Given a solution $\rho_{A,0}^*,\dots,\rho_{A,n}^*$ which attains the
maximum in Eq.~(\ref{eq:primmax}), we can use the above properties to write
\be
\bra{\psi_{A,0}} Z_{A,0} \ket{\psi_{A,0}} =
\Tr[Z_{A,0} \rho_{A,0}^*] \geq \Tr[Z_{A,n} \rho_{A,n}^*] = P_B^*,
\ee
\noindent
obtaining an upper bound on $P_B^*$. The crucial trick is that while we do
not know the complete optimal solution, we do know that
$\rho_{A,0}^*=\ket{\psi_{A,0}}\bra{\psi_{A,0}}$, which gives us a way of
computing the upper bound.

How do we enforce Eq.~(\ref{eq:dualprop})? We do it independently for each
transition: For $i$ odd, Eqs.~(\ref{eq:primtilde},\ref{eq:primodd}) give us
$\rho_{A,i-1} = \Tr_\HM \tilde \rho_{A,i-1}$ and $\rho_{A,i} = \Tr_\HM[
U_{A,i} \tilde\rho_{A,i-1} U_{A,i}^\dagger]$.  We are therefore trying to
impose
\be
\Tr\lp[ \lp(Z_{A,i-1}\otimes I_\HM\rp) \tilde \rho_{A,i-1} \rp]
\geq
\Tr\lp[ \lp(Z_{A,i}\otimes I_\HM\rp) U_{A,i} \tilde \rho_{A,i-1} 
U_{A,i}^\dagger\rp].
\ee
A sufficient condition (which is also necessary if $\tilde \rho_{A,i-1}$ is
arbitrary) is given by
\be
Z_{A,i-1}\otimes I_\HM \geq
U_{A,i}^\dagger \lp(Z_{A,i}\otimes I_\HM\rp) U_{A,i}
\qquad\qquad\text{for $i$ odd.}
\label{eq:zodd}
\ee
\noindent
In general, even when Bob is cheating, not all possible density operators
$\tilde \rho_{A,i-1}$ are attainable (otherwise he would have complete
control over Alice's qubits). Therefore, the above constraint could in
principle be overly stringent. However, we shall prove in the next section
that arbitrarily good certificates can be found even when when using the
above constraint.

Using a similar logic, when $i$ is even we have the relation $\rho_{A,i} =
\rho_{A,i-1}$ and so a sufficient condition on the dual variables is
$Z_{A,i-1} \geq Z_{A,i}$. However, from Alice's perspective these are just
dummy transitions. We introduced extra variables to mark the passage of
time during Bob's actions, but really we want to keep Alice's system
unchanged during these time steps. Therefore, we impose the more stringent
requirement on the dual variables
\be
Z_{A,i-1} = Z_{A,i}\qquad\qquad \text{for $i$ even.}
\label{eq:zeven}
\ee

We can summarize the above as follows:

\begin{definition}
Fix a coin-flipping protocol $P$. A set of positive semidefinite operators
$Z_{A,0},\dots,Z_{A,n}$ satisfying
Eqs.~(\ref{eq:zn},\ref{eq:zodd},\ref{eq:zeven}) is known as a 
\textbf{dual feasible point} (for the problem of cheating Bob given 
a protocol $P$).
\end{definition}

Our arguments above prove:

\begin{lemma}
A dual feasible point $Z_{A,0},\dots,Z_{A,n}$ for a coin-flipping protocol
$P$ constitutes a proof of the upper bound $\bra{\psi_{A,0}} Z_{A,0}
\ket{\psi_{A,0}}\geq P_B^*$.
\label{lemma:dual}
\end{lemma}

The importance of the above upper bounds is that the infimum over dual
feasible points actually equals $P_B^*$. In other words, there exist
arbitrarily good upper bound certificates. This result is known as strong
duality and is proven in Appendix~\ref{sec:strong}.

\subsection{\label{sec:UBP}Upper-Bounded Protocols}

Thus far we have studied Kitaev's first coin-flipping formalism. Given a
protocol it helps us find the optimal cheating strategies for Alice and Bob
by formulating these problems as convex optimizations. In general, however,
we do not have a fixed protocol that we want to study. Rather, we want to
identify the optimal protocol from the space of all possible
protocols. Kitaev's second coin-flipping formalism will help us formulate
this bigger problem as a convex optimization.

The goal is to compute the minimum (over all coin-flipping protocols) of
the maximum (over all cheating strategies for the given protocol) of the
bias.  Alternating minimizations and maximizations are often tricky, but we
can get rid of this problem by dualizing the inner maximization, that is,
by replacing the maximum over cheating strategies with a minimum over the
upper-bound certificates discussed in the last section. The goal becomes to
compute the minimum (over all coin-flipping protocols) of the minimum (over
all upper-bound certificates for the given protocol) of the bias.

But we can go a step further and pair up the protocols and upper bounds to
get a single mathematical object which we call an upper-bounded
protocol. The space of upper-bounded protocols includes bad protocols with
tight upper bounds, good protocols with loose upper bounds and even bad
protocols with loose upper bounds. But somewhere in this space is the
optimal protocol together with its optimal upper bound and by minimizing
the bias in this space we can find it (though, strictly speaking, we
must carry out an infimum not a minimum and will only arrive arbitrarily
close to optimality).

\begin{definition}
An \textbf{upper-bounded (coin-flipping) protocol}, or \textbf{UBP},
consists of a coin-flipping protocol together with two numbers $\beta$ and
$\alpha$, a set of positive semidefinite operators $Z_{A,0},\dots,Z_{A,n}$
defined on $\HA$ and a set of positive semidefinite operators
$Z_{B,0},\dots,Z_{B,n}$ defined on $\HB$ which satisfy the equations
\be
Z_{A,0}\ket{\psi_{A,0}}= \beta\ket{\psi_{A,0}}
&\qquad&
Z_{B,0}\ket{\psi_{B,0}}= \alpha\ket{\psi_{B,0}}
\nonumber\\
Z_{A,i-1}\otimes I_\HM \geq
U_{A,i}^\dagger \lp(Z_{A,i}\otimes I_\HM\rp) U_{A,i}
&&
Z_{B,i-1} = Z_{B,i}
\qquad\qquad\qquad\qquad\qquad\qquad
\text{($i$ odd)}
\nonumber\\
Z_{A,i-1}  = Z_{A,i}
&&
I_\HM\otimes Z_{B,i-1} \geq
U_{B,i}^\dagger \lp(I_\HM\otimes Z_{B,i}\rp) U_{B,i}
\qquad\text{($i$ even)}
\nonumber\\
Z_{A,n}=\Pi_{A,1}
&&
Z_{B,n}=\Pi_{B,0}
\ee
We shall refer to the pair $(\beta,\alpha)$ as the upper bound of the UBP.
\end{definition}

\begin{theorem}
A UBP satisfies
\be
P_B^*\leq \beta,\qquad\qquad
P_A^*\leq \alpha,
\ee
\noindent
where $P_B^*$ and $P_A^*$ are the optimal cheating probabilities of the
underlying protocol.
\label{thm:ubp}
\end{theorem}

Note the reverse order of $\beta$ and $\alpha$. That is because $\beta$,
which is the upper bound on Bob's cheating, must be computed from
quantities that involve operations on Alice's qubits. We will normally list
quantities defined or computed on $\HA$ before those from $\HB$.

The proof of the bound on $P_B^*$ from Theorem~\ref{thm:ubp} follows
directly from Lemma~\ref{lemma:dual}. We shall not prove the equivalent
bound on $P_A^*$ tough it follows from nearly identical arguments.

The main difference between the dual feasible points described in the
previous section and the one used in the definition of UBPs is that the
latter is more restrictive: rather than just setting $\beta =
\bra{\psi_{A,0}}Z_{A,0}\ket{\psi_{A,0}}$ we additionally require that
$\ket{\psi_{A,0}}$ be an eigenvector of $Z_{A,0}$.

Clearly these more restricted dual feasible points still yield the desired
upper bounds, proving the theorem. What we shall argue below, though, is
that we have not sacrificed anything by imposing this additional constraint.
Any upper bound that can be proven with the original certificates can be
proven with these more restricted certificates as well.

We use the fact that for every $\epsilon>0$ there exists a $\Lambda>0$ such
that
\be
\Big(\bra{\psi_{A,0}}Z_{A,0}\ket{\psi_{A,0}} + \epsilon\Big) 
\ket{\psi_{A,0}} \bra{\psi_{A,0}} + 
\Lambda \Big(I - \ket{\psi_{A,0}} \bra{\psi_{A,0}}\Big)
\geq
Z_{A,0}.
\ee
\noindent
The left hand side has $\ket{\psi_{A,0}}$ as an eigenvector as desired,
and by transitivity of inequalities it can be used as the new $Z_{A,0}$.
On Bob's side a similar construction can be used to replace both $Z_{B,0}$
and $Z_{B,1}$ while maintaining their equality. By taking $\epsilon$
arbitrarily small we can get arbitrarily close to the old dual feasible
points, and so the infimum over both sets will be the same.

In summary, we have defined our UBPs and argued that finding the infimum
over this set is equivalent to seeking the optimal coin-flipping protocol.
In particular, we have shown that:

\begin{theorem}
Let $f(\beta,\alpha):\R\times\R\rightarrow\R$ be a function such that
$f(\alpha',\beta')\geq f(\alpha,\beta)$ whenever $\alpha'\geq\alpha$ and
$\beta'\geq\beta$, then
\be
\inf_{\text{proto}} f(P_B^*,P_A^*) = \inf_{UBP} f(\beta,\alpha),
\ee
\noindent
where the left optimization is carried out over all coin-flipping protocols
and the right one is carried out over all upper-bounded protocols. In
particular, the optimal bias $f(\beta,\alpha)=\max(\beta,\alpha)-1/2$ can
be found by optimizing either side.
\label{thm:kitmain1}
\end{theorem}

\subsubsection{Lower bounds and operator monotone functions}

We begin our study of upper-bounded protocols by showing how to place lower
bounds on the set of UBPs. Though we will not prove any new lower
bounds, the ideas presented in this section will motivate the constructions
of the next sections.

The main tool tool that we will be using is the following inequality
\be
\bra{\psi_{i-1}} Z_{A,i-1}\otimes I_\HM\otimes Z_{B,i-1} \ket{\psi_{i-1}}
\geq
\bra{\psi_{i}} Z_{A,i}\otimes I_\HM\otimes Z_{B,i} \ket{\psi_{i}}
\label{eq:lowerb}
\ee
\noindent
which is the bipartite equivalent of Eq.~(\ref{eq:dualprop}). The proof for
$i$ odd is that
\be
\bra{\psi_{i-1}} Z_{A,i-1}\otimes I_\HM\otimes Z_{B,i-1} \ket{\psi_{i-1}}
&\geq&
\bra{\psi_{i-1}} 
\big(U_{A,i}^\dagger \otimes I_\HB\big) 
Z_{A,i}\otimes I_\HM\otimes Z_{B,i} 
\big(U_{A,i} \otimes I_\HB\big) 
\ket{\psi_{i-1}}
\nonumber\\
&&=
\bra{\psi_{i}} Z_{A,i}\otimes I_\HM\otimes Z_{B,i} \ket{\psi_{i}}
\label{eq:lowerbproof}
\ee
\noindent
and the proof for $i$ even is nearly identical.

Iterating we obtain the inequality
\be
\beta\alpha =
\bra{\psi_{0}} Z_{A,0}\otimes I_\HM\otimes Z_{B,0} \ket{\psi_{0}}
&\geq&
\bra{\psi_{n}} Z_{A,n}\otimes I_\HM\otimes Z_{B,n} \ket{\psi_{n}}
\nonumber\\
&&= \bra{\psi_{n}} \Pi_{A,1}\otimes I_\HM\otimes \Pi_{B,0} \ket{\psi_{n}} =
0
\ee
\noindent
or equivalently $P_B^* P_A^* \geq 0$, which admittedly is rather
disappointing.

But there is hope. If we had been studying strong coin flipping and were
interested in the case when both Alice and Bob want to obtain the outcome
one, the above analysis would be correct given the minor change
$Z_{B,n}=\Pi_{B,1}$. In such case the above inequality would read $P_B^*
P_A^*\geq 1/2$ which is Kitaev's bound for strong coin flipping
\cite{Kitaev}.

For historical purposes we note that the results up to this point were
already part of Kitaev's first formalism. Pedagogically, though, it makes
more sense to call the optimizations given a fixed protocol the ``first
formalism'' and the optimizations over all protocols the ``second
formalism.''

Returning to the case of weak coin flipping, we can obtain better bounds 
by inserting operator monotone functions into the above inequalities.
An operator monotone function $f:[0,\infty)\rightarrow [0,\infty)$ is a
function that preserves the ordering of matrices. That is
\be
X\geq Y \Rightarrow f(X)\geq f(Y)
\ee
for all positive semidefinite operators $X$ and $Y$. The simplest example
of an operator monotone function is $f(z)=z$. Another example is $f(z)=1$.
Note that not all monotone functions are operator monotone. The classic
example is $f(z)=z^2$ which is monotone on the domain $[0,\infty)$ but is
not operator monotone. A few more facts about operator monotone functions
are collected in the next section.

How do we use the operator monotone functions? A moment ago we were
studying the expression $\bra{\psi_{i}} Z_{A,i}\otimes I_\HM\otimes Z_{B,i}
\ket{\psi_{i}}$. But we could just as well study the expression 
$\bra{\psi_{i}} Z_{A,i}\otimes I_\HM\otimes f(Z_{B,i}) \ket{\psi_{i}}$ for
any operator monotone function $f$. We could then prove an inequality
similar to Eq.~(\ref{eq:lowerb}), and iterating we would end up with the
condition $\beta f(\alpha) \geq P_B f(0)$, where $P_B$ is the honest
probability of Bob wining. Choosing $f(z)=1$ we can derive the bound
$P_B^*\geq P_B$, which at least has the potential of being saturated.

The next obvious step is to put operator monotone functions on both sides
and study expressions of the form $\bra{\psi_{i}} f(Z_{A,i})\otimes
I_\HM\otimes g(Z_{B,i}) \ket{\psi_{i}}$. But because at any time step only
one of $Z_{A,i}$ and $Z_{B,i}$ increases, we can do even better.

\begin{definition}
A \textbf{bi-operator monotone function} is a function
$f(x,y):[0,\infty)\times[0,\infty)\rightarrow[0,\infty)$ such that when one
of the variables is fixed, it acts as an operator monotone function in the
other variable.

More specifically, given $c\in[0,\infty)$ define $f(\underline
c,z):[0,\infty)\rightarrow[0,\infty)$ to be the function $z\rightarrow
f(c,z)$ (i.e., where the first argument has been fixed). Similarly, let
$f(z,\underline c):[0,\infty)\rightarrow[0,\infty)$ be the function
obtained by fixing the second argument. We say $f(x,y)$ is bi-operator
monotone if both $f(\underline c,z)$ and $f(z,\underline c)$ are operator
monotone for every $c\in[0,\infty)$.

Furthermore, given $f(x,y):[0,\infty)\times[0,\infty)\rightarrow[0,\infty)$ we
extend its definition to act on pairs of positive semidefinite operators as
follows: let $X=\sum_{i}x_i\ket{x_i}\bra{x_i}$ and
$Y=\sum_{i}y_i\ket{y_i}\bra{y_i}$ then
\be
f(X,Y) = \sum_{i,j} f(x_i,y_j) \ket{x_i}\bra{x_i}\otimes \ket{y_j}\bra{y_j}.
\ee
\end{definition}

Bi-operator monotone functions satisfy a few simple properties:
\begin{itemize}
\item If $Y'\geq Y$ then $f(X,Y')\geq f(X,Y)$.
\item If $U$ is unitary then $f(X,U Y U^\dagger)= \big(I\otimes U\big) f(X,Y)
\big(I\otimes U^\dagger\big)$.
\item When acting on a tripartite system $f(X,I\otimes Y)=f(X\otimes I,Y)$.
\end{itemize}

We are now ready to prove our most general inequality. Given a bi-operator
monotone function $f$ then we can write
\be
\bra{\psi_{i-1}} f(Z_{A,i-1},I_\HM\otimes
Z_{B,i-1}) \ket{\psi_{i-1}} \geq
\bra{\psi_{i}} f(Z_{A,i},I_\HM\otimes Z_{B,i}) \ket{\psi_{i}},
\ee
\noindent
whose proof is nearly identical to Eq.~(\ref{eq:lowerbproof}). Iterating,
we obtain the lemma:
\begin{lemma}
Given any bi-operator
monotone function $f$, we obtain a bound on coin-flipping protocols given by
\be
f(P_B^*,P_A^*) \geq
P_B f(1,0) + P_A f(0,1).
\ee
\end{lemma}

Can we use bi-operator monotone functions to prove an interesting bound on
weak coin flipping? Certainly not if we believe that we can achieve
arbitrarily small bias. However we shall see that the spaces of bi-operator
monotone functions and coin-flipping protocols are essentially duals. If
there were no arbitrarily good protocols then we would be able to prove
that fact using the above lemma.

\subsubsection{Some more facts about operator monotone functions}

Operator monotone functions are well studied and a good reference on the
subject is \cite{Bhatia}. As they are also central to the task of
constructing coin-flipping protocols we collect here some of their most
important properties which will be used throughout the paper.

Our main interest are functions that map the set of positive semidefinite
operators to itself. Therefore, when not otherwise stated, we assume all
operator monotone functions have domain $[0,\infty)$ and range contained
in $[0,\infty)$. In general, though, operator monotone functions can be
defined on any real domain.

The space of operator monotone functions (on a fixed domain) forms a convex
cone. If $f(z)$ and $g(z)$ are operator monotone then so are
\be
a f(z) + b g(z)
\ee
\noindent
for any $a\geq0$ and $b\geq 0$. Another simple property is that if $f(z)$
is operator monotone on a domain $(a,b)$ then for any $c\in\R$ we have
$f(z-c)$ is operator monotone on $(a+c,b+c)$.

A very important function which is operator monotone on the domain
$(0,\infty)$ is $f(z)=-1/z$. The proof is given by $Y\geq X>0\Rightarrow
I\geq Y^{-1/2}X Y^{-1/2}\Rightarrow I\leq (Y^{-1/2}X
Y^{-1/2})^{-1}\Rightarrow I\leq Y^{1/2}X^{-1} Y^{1/2} \Rightarrow
-Y^{-1}\geq -X^{-1}$. By shifting and restricting the domain we get
\be
f(z)=-\frac{1}{\lambda+z}
\ee
which is operator monotone on $[0,\infty)$ for $\lambda\in(0,\infty)$,
though the range is negative. The range can be fixed by scaling and adding
in the constant function to get $1-\frac{\lambda}{\lambda+z} =
\frac{z}{\lambda+z}$.

In fact, the above functions together with $f(z)=z$ and $f(z)=1$ span the
extremal rays of the convex cone of operator monotone functions. More
precisely, every operator monotone function
$f:(0,\infty)\rightarrow[0,\infty)$ has a unique integral representation
\be
f(z) = c_1 + c_2 z + \int_{0}^{\infty} \frac{\lambda z}{\lambda + z} d
w(\lambda),
\ee
\noindent
where $c_1,c_2\in\R$ are non-negative and $d w(\lambda)$ is a positive measure
such that $\int_0^\infty \frac{\lambda}{1+\lambda}d w(\lambda) < \infty$.
In particular, they are infinitely differentiable.

The general case, where the domain is the interval $I=(a,b)\subset\R$, is
nearly identical but with the integral ranging over $-\lambda\in\R\setminus
I$. When the domain is the closed interval $[a,b]$ then functions must be
operator monotone on $(a,b)$ and monotone on $[a,b]$ so that $f(a)\leq
\lim_{z\rightarrow a^+} f(z)$ and $f(b)\geq
\lim_{z\rightarrow b^-} f(z)$.

\subsection{\label{sec:TDPG}Time Dependent Point Games}

In previous sections we have paired the honest probability distribution
$\sigma_{A,i}$ with the dual variable $Z_{A,i}$. We have also
paired the full honest state $\ket{\psi_i}$ with the operator
$Z_{A,i}\otimes I\otimes Z_{B,i}$. These pairings have led to interesting
results, but they can also become rapidly unwieldy because they contain too
much information, such as a choice of basis. The goal of this section
is to get rid of most of this excess information and strip the problem to a
bare minimum that still contains the essence of coin flipping.

The key idea for the following discussion is to use the honest state to
define a probability distribution over the eigenvalues of the dual SDP
variables. This idea is captured by the next definition.

\begin{definition}
Given $Z=\sum_{z\in\eig(Z)} z \Pi^{[z]}$, a positive semidefinite matrix
expressed as a sum of its eigenspaces, and $\sigma$, a second positive
semidefinite matrix defined on the same space, we define the function
{\boldmath $\Prob(Z,\sigma):[0,\infty)\rightarrow[0,\infty)$} 
as follows
\be
p(z)=\Prob(Z,\sigma)\quad\Rightarrow\quad
p(z) = 
\begin{cases}
\Tr[\Pi^{[z]} \sigma] & z\in \eig(Z),\\
0 & \text{otherwise.}
\end{cases}
\ee
Similarly, given a vector $\ket{\psi}$ instead of $\sigma$ we define
{\boldmath $\Prob(Z,\ket{\psi}):[0,\infty)\rightarrow[0,\infty)$} by
\be
p(z)=\Prob(Z,\ket{\psi})\quad\Rightarrow\quad
p(z) = 
\begin{cases}
\bra{\psi}\Pi^{[z]}\ket{\psi} & z\in \eig(Z),\\
0 & \text{otherwise.}
\end{cases}
\ee
\end{definition}
\noindent
Note that by construction
$\Prob(Z,\ket{\psi})\equiv\Prob(Z,\ket{\psi}\bra{\psi})$.

We think of the above functions $p(z)$ as belonging to the space of
functions $[0,\infty)\rightarrow[0,\infty)$ with finite support. The
motivation for the construction is that given any function with arbitrary
support $f(z):[0,\infty)\rightarrow\R$ we have
\be
p(z)=\Prob(Z,\sigma)\quad\Rightarrow\quad
\sum_z p(z) f(z) = \Tr[\sigma f(Z)],
\ee
\noindent
where the sum on the left is over the finite support of $p(z)$. 

A similar construction can be used for the bipartite case. Take
$Z_A=\sum_{z_A} z_A \Pi_{A}^{[z_A]}$ on $\HA$, $Z_B=\sum_{z_B} z_B
\Pi_{B}^{[z_B]}$ on $\HB$ and $\ket{\psi}$ on $\HA\otimes\HM\otimes\HB$ and
combine them to form the two-variable function
\be
p(z_A,z_B) =
\begin{cases}
\bra{\psi}\Pi_{A}^{[z_A]}\otimes I_{\HM} \otimes \Pi_{B}^{[z_B]} \ket{\psi} &
z_A\in\eig(Z_A)\text{ and }z_B\in\eig(Z_B),\\
0 & \text{otherwise,}
\end{cases}
\ee
\noindent
which we will denote by $\Prob(Z_A,Z_B,\ket{\psi})$.

Often it will be useful to use an algebraic notation when describing
functions with finite support. For single-variable functions $p(z)$ with
finite support we introduce a basis $\{[z_i]\}$ of functions that take
the value one at $z_i$ and are zero everywhere else. For instance, a function
with two nonzero values $p(z_1)=c_1$ and $p(z_2)=c_2$ can be written as
\be
p = c_1 [z_1] + c_2 [z_2].
\ee
\noindent
Similarly, for the bipartite we use $\{[x,y]\}$ as basis elements for the
functions $p(x,y)$ with finite support. For instance, the initial state 
$\Prob(Z_{A,0},Z_{B,0},\ket{\psi_0})$ of a UBP always has the form
\be
1[\beta,\alpha].
\label{eq:tn}
\ee

On the other hand, because $Z_{A,n}=\Pi_{A,1}$ and $Z_{B,n}=\Pi_{B,0}$ are
just the projections onto the opposing player's winning space, the final
state $\Prob(Z_{A,n},Z_{B,n},\ket{\psi_n})$ of a UBP always has the form
\be
P_B[1,0] + P_A[0,1]
\label{eq:t0}
\ee
as can be verified from Eq.~(\ref{eq:papb}). Note, however, that
the label $0$ on $\Pi_{B,0}$ refers to the coin outcome and not the
associated eigenvalue. Because $Z_{B,n}=1\Pi_{B,0}+0\Pi_{B,1}$, the
projector onto the one eigenvalue of $Z_{B,n}$ is given by
$\Pi_{B}^{[1]}=\Pi_{B,0}$ and similarly $\Pi_{B}^{[0]}=\Pi_{B,1}$.

In fact, given a UBP we can compute for every $i$ the functions
$\Prob(Z_{A,i},Z_{B,i},\ket{\psi_i})$, which allows us to visualize the UBP
as a movie of sorts where at each time step there is a finite set of points
in the plane. The points move around between time steps according to some
rules which we will determine in a moment. We call these sequences
\textit{point games} and we shall show that they contain all the important
information about UBPs.

One important convention that we introduce, though, is that point games are
always described in \textbf{reverse time order}. Henceforth, we will refer
to Eq.~(\ref{eq:t0}) as the first or $t=0$ point configuration whereas
Eq.~(\ref{eq:tn}) will be referred to as the last or $t=n$ point
configuration. More generally, the point configuration at $t=i$ will be
constructed from the operators $Z_{A,n-i}$, $Z_{B,n-i}$ and
$\ket{\psi_{n-i}}$.

The motivation for reversing the time order is as follows: first we get to
begin at a known starting configuration such as $0.5[1,0]+0.5[0,1]$ (for
the main case of interest $P_A=P_B=1/2$). From there, we can move the
points around following the rules of point games until they merge into a
single point at some location $[\beta,\alpha]$. We can do this without
fixing in advance the number of steps, but rather using the existence of a
single point as an end condition.  The sequence of moves will then encode a
UBP with upper bound $(\beta,\alpha)$.

So what are these rules for moving points around? Let begin with the
one-variable case and examine a transition between $p_i(z)$, constructed
from $Z_{A,n-i}$ and $\sigma_{A,n-i}$, and $p_{i+1}(z)$, constructed from
$Z_{A,n-i-1}$ and $\sigma_{A,n-i-1}$. A necessary condition is
\be
\sum_z p_i(z) f(z) \leq \sum_z p_{i+1}(z) f(z)
\ee
\noindent
for every operator monotone function $f$, where again the sums range over the
finite supports of the respective probability distributions.

The condition is trivially necessary on the time transitions when Bob acts
and Alice does nothing because then $p_i=p_{i+1}$. To prove that the condition
is necessary for the non-trivial transitions recall the usual relation
$Z_{A,n-i-1}\otimes I_\HM \geq U_{A,n-i}^\dagger (Z_{A,n-i}\otimes I_\HM )
U_{A,n-i}$. Also let $\tilde \sigma_{A,n-i-1} =
\Tr_\HB[\ket{\psi_{n-i-1}}\bra{\psi_{n-i-1}}]$ so that
$\sigma_{A,n-i-1}=\Tr_\HM[\tilde \sigma_{A,n-i-1}]$ and 
$\sigma_{A,n-i}=\Tr_\HM[U_{A,n-i} \tilde \sigma_{A,n-i-1}U_{A,n-i}^\dagger]$
and therefore
\be
\sum_z p_i(z) f(z) &=& \Tr[\sigma_{A,n-i} f(Z_{A,n-i})]
= \Tr[\tilde \sigma_{A,n-i-1} 
f(U_{A,n-i}^\dagger  Z_{A,n-i} \otimes I_{\HM} U_{A,n-i})]
\label{eq:validproof}
\\\nonumber
&\leq&
\Tr[\tilde \sigma_{A,n-i-1} f(Z_{A,n-i-1} \otimes I_{\HM})]
= \Tr[\sigma_{A,n-i-1} f(Z_{A,n-i-1})] = \sum_z p_{i+1}(z) f(z).
\ee
Sadly, there are no sufficient conditions that can be derived just by
looking at one side of the problem. Nevertheless, it will be useful to
define the transitions that satisfy the above constraints as the valid set
of transitions, as this notion will be used as a building block when
studying the more complicated bipartite transitions.

\begin{definition}
Let $p_i(z)$ and $p_{i+1}(z)$ be two functions $[0,\infty)\rightarrow
[0,\infty)$ with finite support. We say $p_i(z)\rightarrow p_{i+1}(z)$ is
a \textbf{valid transition} if $\sum_z p_i(z) = \sum_z p_{i+1}(z)$,
and for every operator monotone function $f$ we have
\be
	\sum_z p_i(z) f(z) \leq \sum_z p_{i+1}(z) f(z).
\ee
\noindent
Furthermore, we say that the transition is \textbf{strictly valid} if
$\sum_z p_i(z) f(z) < \sum_z p_{i+1}(z) f(z)$ for every non-constant
operator monotone function $f$.
\end{definition}

\noindent
For reasons to become clear below, we do not restrict the functions
$p_i(z)$ and $p_{i+1}(z)$ to sum to one. However, we do require that their
sum be equal so that probability is conserved.

Because we have a characterization of the extremal rays of the cone of
operator monotone functions, it is sufficient when checking for valid
transitions to use the set of functions $f_\lambda(z)=\frac{\lambda
z}{\lambda+z}$ for $\lambda\in(0,\infty)$. As for the other two extremal
functions, the constraint from $f(z)=1$ is independently imposed as
conservation of probability, and the constraint from $f(z)=z$ follows from
the limit $\lambda\rightarrow\infty$ of $f_\lambda(z)$. It is worth stating
this explicitly:

\begin{lemma}
Let $p_i(z)$ and $p_{i+1}(z)$ be two functions $[0,\infty)\rightarrow
[0,\infty)$ with finite support. The transition $p_i(z)\rightarrow
p_{i+1}(z)$ is valid if and only if $\sum_z \lp(p_{i+1}(z)- p_i(z)\rp)=0$
and for every $\lambda\in(0,\infty)$ we have $\sum_z  \frac{\lambda
z}{\lambda+z}\lp(p_{i+1}(z) - p_i(z) \rp)\geq 0$.
\end{lemma}

Not surprisingly, a set of necessary conditions for bipartite transitions
constructed from UBPs is that
\be
\sum_{x,y} p_i(x,y) f(x,y) \leq \sum_{x,y} p_{i+1}(x,y)
f(x,y)
\ee
\noindent
for every bi-operator monotone functions $f$, where as usual the sums are
over the finite support of the respective probability distributions.
Unfortunately, the space of bi-operator monotone functions is not as well
characterized as the space of operator monotone functions, and therefore it
would be better to have a set of conditions that are constructed from the
latter:

\begin{definition}
Let $p_i(x,y)$ and $p_{i+1}(x,y)$ be two functions
$[0,\infty)\otimes [0,\infty)\rightarrow [0,\infty)$ with finite support.
We say $p_i(x,y)\rightarrow p_{i+1}(x,y)$ is a \textbf{valid
transition} if either
\begin{enumerate}
\item for every $c\in[0,\infty)$ the transition
$p_i(z,\underline c)\rightarrow p_{i+1}(z,\underline c)$ is valid, or
\item for every $c\in[0,\infty)$ the transition
$p_i(\underline c,z)\rightarrow p_{i+1}(\underline c,z)$ is valid,
\end{enumerate}
\noindent
where as before $p_i(z,\underline c)$ is the one-variable function obtained
by fixing the second input. We call the first case a \textbf{horizontal
transition} and the second case a \textbf{vertical transition}.
\end{definition}

\noindent
The first case occurs when Alice applies a unitary and the second case when
Bob applies a unitary. As opposed to the single variable transitions, the
bipartite condition of validity is not transitive. However, we can define
the notion of transitively valid for two functions if there is a sequence
of functions beginning with the first one and ending with the second one,
such that each transition is valid. The main object of study for this
section will be transitively valid transitions of the form
$P_B[1,0]+P_A[0,1]\rightarrow 1[\beta,\alpha]$, where we always assume
$P_A,P_B\geq 0$ are some fixed numbers such that $P_A+P_B=1$.

\begin{definition}
A \textbf{time dependent point game (TDPG)} is a sequence
$p_0(x,y),\dots,p_n(x,y)$ of functions
$[0,\infty)\otimes [0,\infty)\rightarrow
[0,\infty)$ with finite support, such that every transition
$p_i(x,y)\rightarrow p_{i+1}(x,y)$ is valid and such that the first
and last distributions have the form
\be
p_0 = P_B[1,0] + P_A[0,1],
\qquad\text{and}\qquad
p_n = 1[\beta,\alpha].
\ee
We say that $[\beta,\alpha]$ is the final point of the TDPG.
\end{definition}

\noindent
The above definition can be extended to games beyond coin-flipping which
have many possible outcomes. If outcome $i$ has honest probability $q_i$,
and pays $a_i\geq 0$ to Alice and $b_i\geq 0$ to Bob, the same formalism
applies if we use as starting state $p_0 =
\sum_i q_i [b_i,a_i]$. This paper will focus exclusively on weak
coin-flipping though.

Our first task will be to prove that given a UBP with bound
$(\beta,\alpha)$ we can build a TDPG with final point $[\beta,\alpha]$.  We
have already done most of the work by constructing the probability
distributions out of the UBP and showing that they have the right initial
and final states. What remains to be shown is that the transitions
$p_i(x,y)\rightarrow p_{i+1}(x,y)$ are valid.

We focus on the transitions when Alice applies a unitary, the other case
being nearly identical. Given a UBP, we construct the distribution
$p_i=\Prob(Z_{A,n-i},Z_{B,n-i},\ket{\psi_{n-i}})$ and the distribution
$p_{i+1}=\Prob(Z_{A,n-i-1},Z_{B,n-i-1},\ket{\psi_{n-i-1}})$. The UBP
operators satisfy the usual relations $Z_{A,n-i-1}\otimes I_\HM \geq
U_{A,n-i}^\dagger (Z_{A,n-i}\otimes I_\HM) U_{A,n-i}$,
$Z_{B,n-i-1}=Z_{B,n-i}$ and $\ket{\psi_{n-i}}= U_{A,n-i}\otimes I_B
\ket{\psi_{n-i-1}}$. We expand the state $\ket{\psi_{n-i-1}}$ as
\be
\ket{\psi_{n-i-1}} = \sum_{y} \ket{\phi_{y}}\otimes \ket{y},
\ee
\noindent
where $\ket{y}$ are normalized eigenvectors of $Z_{B,n-i}$ and
$\ket{\phi_{y}}$ are non-normalized states on $\HA\otimes\HM$.  Because
this is not necessarily a Schmidt decomposition, the vectors
$\ket{\phi_{y}}$ are not necessarily orthogonal, however this will not be a
problem. The key idea now is to note that given a fixed $y$, the function
$p_{i+1}(x,\underline{y})=\Prob(Z_{A,n-i-1},\rho_{n-i-1,y})$ where
$\rho_{n-i-1,y}\equiv\Tr_{\HM}[\ket{\phi_{y}}\bra{\phi_{y}}]$. Similarly,
the function $p_{i}(x,\underline{y})=\Prob(Z_{A,n-i},\rho_{n-i,y})$ where 
$\rho_{n-i,y}\equiv\Tr_{\HM}[U_{A,n-i}\ket{\phi_{y}}\bra{\phi_{y}}
U_{A,n-i}^\dagger]$. The relationship between these quantities is the same
as it was the analysis of one-sided probability transitions, and the proof
of Eq.~(\ref{eq:validproof}) goes through with $\rho_{n-i-1,y}$,
$\rho_{n-i,y}$ and $\ket{\phi_{y}}\bra{\phi_{y}}$ taking the place of
$\sigma_{A,n-i-1}$, $\sigma_{A,n-i}$ and $\tilde
\sigma_{A,n-i-1}$. Therefore, for every $y\in[0,\infty)$ the transition
$p_i(x,\underline{y})\rightarrow p_{i+1}(x,\underline{y})$ is valid
and therefore the full transition $p_i(x,y)\rightarrow
p_{i+1}(x,y)$ is valid as well.

We have just proven that given a UBP with bound $(\beta,\alpha)$ we can
construct a TDPG with final point $[\beta,\alpha]$. The converse is also
true in the following sense: given any TDPG with final point
$[\beta,\alpha]$ and an $\epsilon>0$ there exists a UBP with bound
$(\beta+\epsilon,\alpha+\epsilon)$. This is enough because we are only
concerned with infimums, and the infimums over both sets will be
equal. Constructing UBPs from TDPGs will require a fair amount of work to
be done in the next couple of sections. Some readers may prefer to first
study the TDPG examples in Section~\ref{sec:examples}.

\subsubsection{\label{sec:proj}Coin-flipping protocols with projections}

The description of the coin-flipping protocols built from TDPGs will be
greatly simplified if we can use measurements at intermediate steps
throughout the protocol. Of course, there is nothing special about
protocols that involve measurements, as these can always be delayed to the
last step. However, doing so requires simulating the measurement with a
unitary and keeping the simulated outcomes in some extra qubits, which
adds extra complexity to the description of the protocol. Early
measurements can result in significant improvements in the description of a
protocol and the number of qubits employed. An example of this is given in
Appendix~\ref{sec:ddb} which takes the author's original bias $1/6$
protocol requiring arbitrarily many qubits and reduces the space used to a
single qutrit per player plus a single qubit for messages.

In fact, for such simplifications we need only allow a special kind of
projective measurement: at certain steps the players will use a two outcome
POVM of the form $\{E,I-E\}$, where E is a projector (i.e.,
$E^\dagger=E^2=E$). The protocol will be set up so that if both players
play honestly the first outcome will always be obtained, and if the second
outcome is observed the players will immediately abort (at which point they
can declare themselves the winner). We will place one such projection
immediately after each unitary.

The goal of this section is to formalize the needed notion of protocols
with projections, describe their dual feasible points, and show how they
are equivalent to regular protocols. It can be safely skipped by those
familiar with the result.

\begin{definition}
A \textbf{coin-flipping protocol with projections} is a coin-flipping
protocol with the addition of $n$ projection operators $E_1,\dots,E_n$ of
the form
\be
E_i = \begin{cases}
E_{A,i}\otimes I_\HB & \text{for $i$ odd,}\\
I_\HA\otimes E_{B,i} & \text{for $i$ even,}
\end{cases}
\ee
such that $E_i\ket{\psi_{i}}=\ket{\psi_{i}}$ for every $i=1,\dots,n$
\end{definition}

\noindent
The protocol is implemented as before, except that immediately after
implementing $U_i$ the acting player measures using $\{E_i,I-E_i\}$ and
aborts on the second outcome.

To prove the equivalence of protocols with and without measurements, not
only do we need to construct a canonical unitary-based protocol for every
protocol with measurements, but we also need to show that the new protocol
does not allow for any extra cheating. This is done by constructing a
canonical map from the dual feasible points of the protocol with
measurements to the dual feasible points of the unitary protocol such that
the upper bounds are preserved (or at least come arbitrarily close to each
other). Effectively, we aim to construct a map from ``UBPs with
measurements'' to regular UBPs. As most of the constructions are fairly
standard, we will only sketch the details.

As usual we focus on the case of honest Alice and cheating Bob. The primal
SDP requires $\rho_{A,0} = \ket{\psi_{A,0}}\bra{\psi_{A,0}}$, $\rho_{A,i} =
\rho_{A,i-1}$ for $i$ even, and $P_{win} = \Tr\lp[ \Pi_{A,1} \rho_{A,n}
\rp]$ as before. However, the new element is that for $i$ odd we have
\be
\rho_{A,i} = \Tr_\HM\lp[ E_{A,i} U_{A,i} \tilde\rho_{A,i-1}
U_{A,i}^\dagger E_{A,i}\rp],
\qquad\text{for}\qquad
\Tr_\HM \tilde \rho_{A,i-1} \leq \rho_{A,i-1}.
\ee
\noindent
The trace of $\rho_{A,i}$ is no longer unity but rather it encodes the
probability that we have reached step $i$ without aborting and
$\rho_{A,i}/\Tr[\rho_{A,i}]$ is the state at step $i$ given that no aborts
have occurred. The inequality on the right equation allows for Bob to abort,
though it is certainly never optimal for him to do so.

The dual SDP has as before $Z_{A,n}= \Pi_{A,1}$, $Z_{A,i-1}  = Z_{A,i}$ for $i$
even, and $\beta=\bra{\psi_{A,0}}Z_{A,0}\ket{\psi_{A,0}}$ but now
for $i$ odd we impose the condition
\be
Z_{A,i-1}\otimes I_\HM \geq
U_{A,i}^\dagger E_{A,i} \lp(Z_{A,i}\otimes I_\HM\rp) E_{A,i} U_{A,i}.
\ee

Given a protocol with projections and a dual feasible point let us build a
unitary protocol with a matching dual feasible point. We will put primes on
all expressions of the new protocol that differ from the one with
measurements.

The number of rounds and the message space will be the same, but we will
add $n$ qubits to both $\HA$ and $\HB$ so that $\HAP =
\lp(\C^2\rp)^{\otimes n}\otimes\HA$ and $\HBP =
\HB\otimes\lp(\C^2\rp)^{\otimes n}$. These extra qubits will store the
measurement outcomes (though technically we only need $n/2$ qubits on each
side). The new initial state will set all the new qubits to zero
$\ket{\psi_{A,0}'}=\ket{0}\otimes\ket{\psi_{A,0}}$ and
$\ket{\psi_{B,0}'}=\ket{\psi_{B,0}}\otimes\ket{0}$, and when playing
honestly they will always remain zero. The final projectors only give the
victory to the other player if all the extra qubits are zero so that
$\Pi_{A,1}'=\ket{0}\bra{0}\otimes \Pi_{A,1}$ and $\Pi_{A,0}'=I-\Pi_{A,1}'$
and similarly $\Pi_{B,0}'= \Pi_{B,0}\otimes\ket{0}\bra{0}$ and
$\Pi_{B,1}'=I-\Pi_{B,0}'$. Finally, the new unitaries simply simulate a
measurement after applying the regular operation
\be
U_{A,i}' = M_{A,i} \lp(I\otimes U_{A,i}\rp)&\qquad&\text{for $i$ odd,}\\
U_{B,i}' = M_{B,i} \lp(U_{B,i}\otimes I\rp)&\qquad&\text{for $i$ even,}
\ee
\noindent
where $M_{A,i}$ and $M_{B,i}$ are controlled unitaries with target given by
the original $\HA$ (or $\HB$) and new qubit number $i$ and control given
by the other $n-1$ new qubits. The matrices acts as the identity unless all
the $n-1$ control qubits are zero, in which case they apply the operation
\be
M_{A,i}\rightarrow \mypmatrix{E_{A,i}&I-E_{A,i}\cr I-E_{A,i}&E_{A,i}},
\qquad\qquad
M_{B,i}\rightarrow \mypmatrix{E_{B,i}&I-E_{B,i}\cr I-E_{B,i}&E_{B,i}}.
\ee
\noindent
where the blocks correspond to the computational basis of new qubit $i$.

It is not hard to check that the new protocol is indeed a valid
coin-flipping protocol and that the honest probabilities of winning $P_A$
and $P_B$ are the same as in the original protocol.

Now we take a dual feasible point for the original protocol and fix
$\epsilon>0$. We will construct a dual feasible point for the new protocol
with $\beta'= \beta+n\epsilon$. For $i$ even define
\be
Z_{A,i}' = \ket{0}\bra{0} \otimes (Z_{A,i}+(n-i)\epsilon I) + \Lambda_i F_i,
\ee
\noindent
where $F_i$ is a projector onto the space such that at least one of the new
qubits labeled $i+1$ through $n$ is non-zero, and $\Lambda_i\geq 0$ is a
constant to be determined in a moment. By construction
$Z_{A,n}'=\ket{0}\bra{0} \otimes Z_{A,n} = \Pi_{A,1}'$ and $\beta' =
\bra{\psi_{A,0}'}Z_{A,0}'\ket{\psi_{A,0}'} =
\bra{\psi_{A,0}}Z_{A,0}+n\epsilon\ket{\psi_{A,0}}=\beta+n\epsilon$ as required.
We can also set $Z_{A,i-1}'=Z_{A,i}'$ for $i$ even, so that the only
constraint that remains to be checked is
\be
Z_{A,i-2}'\otimes I_\HM \geq
{U_{A,i-1}'}^\dagger \lp(Z_{A,i}'\otimes I_\HM\rp) U_{A,i-1}'.
\ee
We will describe a decomposition
$\HAP=\HH_1\oplus\HH_2\oplus\HH_3\oplus\HH_4$ such that both sides of the
above inequality are block diagonal with respect to it, and therefore we can
check the inequality on each block separately. The decomposition is
obtained by looking at the $n$ new qubits of $\HAP$ from last (qubit $n$)
to first (qubit $1$) and picking out the first non-zero qubit.
\begin{itemize}
\item $\HH_1$ contains vectors where the first non-zero qubit is one of
$i+1,\dots,n$.
\item $\HH_2$  contains vectors where the first non-zero qubit is either
$i-1$ or $i$, but excluding the vector where qubit $i-1$ is the only non-zero.
\item $\HH_3$ contains vectors where the first non-zero qubit is one of
$1,\dots,i-2$.
\item $\HH_4$ contains the space where all new qubits are zero or where
all qubits but qubit $i-1$ are zero.
\end{itemize}

On $\HH_1$ we have $F_{i-2}=F_{i}=I$ so the inequality reads
$\Lambda_{i-2}I\geq\Lambda_{i} I$, and is satisfied so long as $\Lambda_i$
is a decreasing sequence. On $\HH_2$ we have $F_{i-2}=I$ and $F_{i}=0$ so
the inequality reads $\Lambda_{i-2}I\geq 0$. On $\HH_3$ we have
$F_{i-2}=F_{i}=0$ so the inequality reads $0\geq 0$. Finally, $\HH_4$ is
the only space on which $M_{A,i-1}$ acts non-trivially. Writing
$X=Z_{A,i-2}\otimes I_\HM + (n-i+2)\epsilon I$ and $Y=Z_{A,i}\otimes I_\HM +
(n-i)\epsilon I$ and using $U\equiv U_{A,i-1}$, $E\equiv E_{A,i-1}$ we need
to check the block diagonal inequality
\be
\mypmatrix{X & 0 \cr 0 &
\Lambda_{i-2} I}
&\geq&
\mypmatrix{U^\dagger & 0 \cr 0 & U^\dagger}
\mypmatrix{E & I-E \cr I-E & E}
\mypmatrix{Y & 0 \cr 0 & 0}
\mypmatrix{E & I-E \cr I-E & E}
\mypmatrix{U & 0 \cr 0 & U}
\nonumber\\\nonumber
&& = \mypmatrix{U^\dagger E Y E U & U^\dagger E Y (I - E) U\cr
U^\dagger (I - E) Y E U & U^\dagger (I - E) Y (I - E) U}.
\ee
\noindent
From the constraints on the original dual feasible point we had
$Z_{A,i-2}\otimes I_\HM \geq U^\dagger E (Z_{A,i}\otimes I_\HM) E U =
U^\dagger E Y E U - (n-i)\epsilon E$ and therefore $X > U^\dagger E Y E
U$. In turn, this implies that for sufficiently large $\Lambda_{i-2}$ the
whole matrix inequality holds, and then we just need to make sure that
$\Lambda_{i-2}$ is also larger than $\Lambda_i$. We have not specified yet
$\Lambda_n$ but this one can be chosen to be zero as it effectively never
appears anywhere. That concludes the proof of the equivalence of protocols
with and without projective measurement.

\subsubsection{\label{sec:compiling}Compiling TDPGs into UBPs}

We had previously shown how to construct a TDPG out of a UBP. In this
section we will describe the reverse construction, thereby proving 
TDPGs and UBPs equivalent.

We assume that all transitions in the given TDPG alternate between
horizontal and vertical (i.e., if we have a sequence of two valid
horizontal transitions we can combine them into a single valid transition
by removing the middle step). We also assume that the first transition $p_0
\rightarrow p_1$ is vertical and the last one $p_{n-1}\rightarrow p_n$ is
horizontal (which can be accomplished by adding trivial transitions at the
beginning or end). All TDPGs obtained from UBPs have this form.

We will also need to assume that in the given TDPG all non-trivial
one-variable transitions are strictly valid. This is justified by the
following lemma.

\begin{lemma}
Given any TDPG $p_0,\dots,p_n$ with final point $[\beta,\alpha]$ and an
$\epsilon>0$, there exists a second TDPG $q_0,\dots,q_m$ with final point 
$[\beta+\epsilon/2,\alpha+\epsilon/2]$ such that every
non-trivial one-variable transition is strictly valid.
\end{lemma}

\begin{proof}
We construct the new TDPG by shifting up (or left) each set of points in
$q_i$ relative to its predecessor. More specifically let $\lceil i/2
\rceil$ and $\lfloor i/2 \rfloor$ be $i/2$ rounded up and down
respective. There are respectively the number of vertical and horizontal
transitions that have occurred to reach $p_i$.  Now define the new TDPG by
\be
q_i(x,y) = p_i(x-\lp\lfloor\frac{i}{2}\rp\rfloor\frac{\epsilon}{n},
y-\lp\lceil\frac{i}{2}\rp\rceil\frac{\epsilon}{n}).
\ee
The final point is $[\beta+\frac{\epsilon}{2},\alpha+\frac{\epsilon}{2}]$
as required.  Also each non-trivial transition is now strictly valid: for
instance, if $q_i\rightarrow q_{i+1}$ is a horizontal transition and
$y\in[0,\infty)$ such that $\sum_z q_{i}(z,\underline y)\neq0$, then (using
$y'=y-\lp\lceil\frac{i}{2}\rp\rceil\frac{\epsilon}{n}$)
\be
\sum_z p_{i+1}'(z,\underline y) f(z) 
&=& \sum_z p_{i+1}(z,\underline y') f(z +
\lp\lfloor\frac{i}{2}\rp\rfloor\frac{\epsilon}{n} + \frac{\epsilon}{n})
\\\nonumber
&>& \sum_z p_{i+1}(z,\underline y') f(z + 
\lp\lfloor\frac{i}{2}\rp\rfloor\frac{\epsilon}{n})
\geq \sum_z p_{i}(z,\underline y') 
f(z +\lp\lfloor\frac{i}{2}\rp\rfloor\frac{\epsilon}{n})
= \sum_z p_{i}'(z,\underline y) f(z)
\ee
\noindent
for every non-constant operator monotone $f$. The first inequality follows
because non-constant operator monotone functions are strictly monotone and
the second because $f(z +c)$ is operator monotone for $c\geq0$.
\end{proof}

\noindent
The next step in our argument relies on two lemmas which we state below
and prove in Appendix~\ref{sec:f2m}. They are the essential ingredient
which takes pairs of functions, such as those in a TDPG, and compiles them
back into the language of matrices.

\begin{lemma}
Let $p(z)$ and $q(z)$ be functions $[0,\infty)\rightarrow[0,\infty)$ with
finite support. If $p(z)\rightarrow q(z)$ is strictly valid then there
exists positive semidefinite matrices $X$ and $Y$ and an (unnormalized)
vector $\ket{\psi}$ such that $X\leq Y$, $p=\Prob(X,\ket{\psi})$ and
$q=\Prob(Y,\ket{\psi})$.
\label{lemma:f2mmain}
\end{lemma}

\begin{lemma}
The matrices $X$ and $Y$ in Lemma~\ref{lemma:f2mmain} can be chosen such that
\begin{enumerate}
\item The spectrum of $X$ is equal to $\{0\}\cup S(p)$,
with all non-zero eigenvalues occurring once.
\item The spectrum of $Y$ is equal to $\{\Lambda\}\cup S(q)$,
for some large $\Lambda>0$, with all other eigenvalues occurring once.
\item The dimension of $X$ and $Y$ is no greater than 
$|S(p)|+|S(q)|-1$.
\end{enumerate}
\noindent
where $S(p)$ and $S(q)$ are respectively the supports of $p$ and $q$.
\label{lemma:stdformmain}
\end{lemma}

Because we are now assuming that every non-trivial one-variable transition
is strictly valid, we can use Lemma~\ref{lemma:f2mmain} to turn the
transitions into matrices from which we will extract unitaries. However,
first we need to standardize the Hilbert space on which all of these
matrices are defined.

Let us fix a finite set $S$ of non-negative numbers, and assume we are
given a strictly valid transition $p\rightarrow q$ such that the supports
of both $p$ and $q$ are contained in $S$. What we want to argue is that we
can choose $X$ and $Y$ of Lemma~\ref{lemma:f2mmain} so that
their spectrum is exactly $S$ (union zero for $X$ or union some large value
$\Lambda$ for $Y$) and such that the only degenerate eigenvalues are zero
for $X$ and $\Lambda$ for $Y$.

The argument is simple, we start with $X$ and $Y$ satisfying the
requirements of Lemma~\ref{lemma:stdformmain}. If the $\Lambda$ appearing
in $Y$ is not larger than the maximum of $S$ we can simply increase it. Now
we just start appending in the missing eigenvalues from $S$ one at a time
(increasing the dimension of the matrices by using a direct sum). If some
value $c\in S$ is in $X$ but not in $Y$ we can append it to $Y$ if at the
same time we append a zero to $X$. Similarly if $c\in S$ is in $Y$ but not
in $X$ we can append it to $X$ if at the same time we append an extra
$\Lambda$ eigenvalue to $Y$. If $c\in S$ appears in neither matrix we
append it to both at the same time. The dimension of the new matrices so
constructed is no larger than $2|S|$. The dimension can be made exactly
equal to $2|S|$ by appending in extra zeros to $X$ and $\Lambda$s to $Y$.

To extract unitaries from these matrices, note that given any basis, we can
find a unitary $U$ such that $U X U^\dagger$ is diagonal and $U\ket{\psi}$
has non-negative coefficients with respect to this basis. In particular, if
we choose the basis that diagonalizes $Y$ and such that $\ket{\psi}$ has
non-negative coefficients, then we can find such a unitary $U$ and get
\be
Y_d \geq U^\dagger X_d U,
\ee
\noindent
where $X_d$ and $Y_d$ are diagonal in the computational basis with the same
spectrum as $X$ and $Y$ respectively. Additionally, by construction the
non-negative coefficients in the computational basis must have the form
$\ket{\psi}=\sum_i \sqrt{q_i}\ket{i}$ and
$U\ket{\psi}=\sum_i\sqrt{p_i}\ket{i}$. Putting everything together, we
obtain the following rather surprising lemma.

\begin{lemma}
Let $S$ be a finite set of non-negative numbers. Let $p\rightarrow q$ be a
strictly valid transition such that the support of both $p(z)$ and $q(z)$ are
contained in $S$. Let $\HH$ be the Hilbert space spanned by
$\{\ket{i,s}:i\in\{0,1\},s\in S\}$ and define
\be
Z = \sum_{s\in S} s \ket{0,s}\bra{0,s}.
\ee
Then there exists a sufficiently large number $\Lambda>0$ and a unitary $U$
such that
\be
U \sum_{s\in S} \sqrt{q(s)} \ket{0,s} = \sum_{s\in S} \sqrt{p(s)} \ket{0,s}
\qquad\text{and}\qquad
Z + \Lambda P_1 \geq U^\dagger Z U,
\label{eq:magicU}
\ee
\noindent
where $P_1=\sum_s\ket{1,s}\bra{1,s}$ is the projector onto the space where
the first qubit is one.
\label{lemma:magicL}
\end{lemma}

\noindent
Note that the result is non-trivial. If $p\rightarrow q$ is not valid 
then no such unitary exists.

In the protocol below we will essentially be able to choose all our dual
operators equal to $\sum_{s\in S} s \ket{0,s}\bra{0,s} + \Lambda P_1$. We
then use our projective measurements to reset the eigenvalue $\Lambda$ to
zero on every transition so we can apply the above lemma. The unitaries,
though, require more care. During every bipartite transition
$p_i(x,y)\rightarrow p_{i+1}(x,y)$ we have many one-variable strictly valid
transitions $p_i(x,\underline y)\rightarrow p_{i+1}(x,\underline y)$ during
Alice's turn (or $p_i(\underline x,y)\rightarrow p_{i+1}(\underline x, y)$
during Bob's turn), each of which defines a different unitary in the above
lemma.What Alice needs to do on her turn is to apply a block diagonal
unitary with each block corresponding to a different strictly valid
transition. In other words she needs to be able to apply a controlled
unitary with control given by Bob's state. That is what the messages are
used for. Bob will store his state in an entangled subspace of
$\HM\otimes\HB$ so that Alice can use it as control but not change it too
much. Similarly, Alice will store her state in an entangled subspace of
$\HA\otimes\HM$ so that Bob can access it. This construction is similar to
one used in
\cite{KMP04}.

We are now ready to construct the protocol. Fix a TDPG by $p_0,\dots,p_n$
with strictly valid transitions and final point $[\beta,\alpha]$. The
protocol will be defined with the same $n$ representing the number of
messages.

To define the relevant Hilbert spaces let $S_A$ (resp. $S_B$) be the finite
set of $x$ coordinates (resp. $y$ coordinates) of points that are assigned
non-zero probability by $p_i(x,y)$ for some $i$. By construction
$0,1,\beta\in S_A$ and $0,1,\alpha \in S_B$. Set
\be
\HA &=& \vspan\{\ket{i,s_a}:i\in\{0,1\},s_a\in S_A\},\\
\HM &=& \vspan\{\ket{s_a,s_b}:s_a\in S_A,s_b\in S_B\},\\
\HB &=& \vspan\{\ket{s_b,i}:s_b\in S_B,i\in\{0,1\}\}.
\ee
\noindent
It will ocassionally be useful to write $\HM=\HAP\otimes\HBP$ where
$\HAP = \vspan\{\ket{s_a}:s_a\in S_A\}$ and 
$\HBP = \vspan\{\ket{s_b}:s_b\in S_B\}$.

The initial states will be
\be
\ket{\psi_{0,A}} = \ket{0,\beta},\qquad
\ket{\psi_{0,M}} = \ket{\beta,\alpha},\qquad
\ket{\psi_{0,B}} = \ket{\alpha,0}.
\ee
The projections that follow the unitaries will have the form
\be
E_{A,i} = E_A &\equiv& \sum_{s_a\in S_A} \ket{0,s_a}\bra{0,s_a}
\otimes \ket{s_a}\bra{s_a}\otimes I_\HBP,\\
\qquad
E_{B,i} = E_B &\equiv& I_{\HAP}\otimes \sum_{s_b\in S_B} \ket{s_b}\bra{s_b}
\otimes \ket{s_b,0}\bra{s_b,0},
\ee
\noindent
where $E_{A,i}$ is defined for $i$ odd and $E_{B,i}$ for $i$
even. Basically $E_A$ acts on $\HA\otimes\HM$ and ensures that 
the first qubit is $0$ and the registers $s_a$ in $\HA$ and $\HM$ agree.

The final measurement operators will have the form
\be
\Pi_{A,1} = \ket{0,1}\bra{0,1},\qquad \Pi_{B,0} = \ket{1,0}\bra{1,0}
\ee
\noindent
with $\Pi_{A,0}=I-\Pi_{A,1}$ and $\Pi_{B,1}=I-\Pi_{B,0}$.

All that remains is to describe the unitaries, which will be done in a
moment. The unitaries are to be chosen so that the honest state during the
$i$th messages is
\be
\ket{\psi_i} = \sum_{s_a\in S_A, s_b\in S_B} \sqrt{p_{n-i}(s_a,s_b)}\,
\ket{0,s_a}\otimes\ket{s_a,s_b}\otimes\ket{s_b,0},
\label{eq:TDPGpsi}
\ee
\noindent
where we remind the reader of our ``reverse time'' convention of TDPGs
relative to protocols. The above definition agrees with our choice of
initial state $\ket{\psi_0}$. It is also easy to see that the projection
operations that follow the unitaries will always succeed if both players are
honest. Finally, we have $\ket{\psi_n}= P_B\ket{0,1,1,0,0,0} +
P_A\ket{0,0,0,1,1,0}$ so Alice and Bob will agree on the coin outcome,
which will have the required probability distribution.

Before defining the unitaries, we fix a single $\Lambda>0$ larger than all
elements of $S_A$ and $S_B$ and large enough so that all the strictly valid
transitions in the given TDPG can be turned into unitaries satisfying
the constraints of Eq.~(\ref{eq:magicU}).

To construct the unitaries for Alice it is useful to define the subspace
$\HAB\subset\HA\otimes\HAP$ given by
$\HAB=\{\ket{i,s_a,s_a}:i\in\{0,1\},s_a\in S_A\}$ and let $P_\HAB^\perp$ be
the the projector onto the complement of $\HAB$ in $\HA\otimes\HAP$.  Let
$i$ be odd so that we can build $U_{A,i}$ out of the horizontal transition
$p_{n-i}\rightarrow p_{n-i+1}$. The unitary will be block diagonal of the
form
\be
U_{A,i} = \sum_{s_b\in S_B} \lp(U_{A,i}^{(s_b)} + P_\HAB^\perp \rp)
\otimes \ket{s_b}\bra{s_b},
\ee
\noindent
where $U_{A,i}^{(s_b)}$ can be viewed as a unitary operator on $\HAB$.
If $p_{n-i}(s_a,\underline{s_b})= p_{n-i+1}(s_a,\underline{s_b})$ then we
choose $U_{A,i}^{(s_b)}=I$ otherwise
$p_{n-i}(s_a,\underline{s_b})\rightarrow p_{n-i+1}(s_a,\underline{s_b})$
is strictly valid and by Lemma~\ref{lemma:magicL} we can choose
$U_{A,n-i}^{(s_b)}$ on $\HAB$ such that
\be
\bar{Z}_A + \Lambda \bar{P}_1
\geq {U_{A,i}^{(s_b)}}^\dagger \bar{Z}_A U_{A,i}^{(s_b)}
\ee
for $\bar{Z}_A = \sum_{s_a} s_a \ket{0,s_a,s_a}\bra{0,s_a,s_a}$,
$\bar{P}_1 = \sum_{s_a} \ket{1,s_a,s_a}\bra{1,s_a,s_a}$ and such that
\be
U_{A,i}^{(s_b)}\sum_{s_a\in S_A}\sqrt{p_{n-i+1}(s_a,s_b)}\ket{0,s_a,s_a}
= \sum_{s_a\in S_A}\sqrt{p_{n-i}(s_a,s_b)}\ket{0,s_a,s_a}.
\ee
We can now directly verify the equation $U_{A,i}\otimes I_\HB
\ket{\psi_{i-1}} = \ket{\psi_{i}}$ with states given by
Eq.~(\ref{eq:TDPGpsi}). The unitaries $U_{B,i}$ for Bob are defined
analogously, and otherwise we have completed the description of the
coin-flipping protocol associated to the TDPG.

What remains is to describe dual feasible points for the above protocol
that prove the bounds $P_A^*\leq\alpha$ and $P_B^*\leq\beta$. If we define
the operators
\be
Z_A &=& \sum_{s_a} s_a \ket{0,s_a}\bra{0,s_a} + \Lambda \sum_{s_a}
\ket{1,s_a}\bra{1,s_a}\\
Z_B &=& \sum_{s_b} s_b \ket{0,s_b}\bra{0,s_b} + \Lambda \sum_{s_b}
\ket{1,s_b}\bra{1,s_b}
\ee
on $\HA$ and $\HB$ respectively. The desired dual feasible points are given
by $Z_{A,i}=Z_{A}$ and $Z_{B,i}=Z_{B}$ for all $i$ except that we must set
$Z_{A,n}=Z_{A,n-1}=\Pi_{A,1}$ and $Z_{B,n}=\Pi_{B,0}$ as required by our
slightly inflexible definitions. As usual, we will only verify the case of
Alice honest and Bob cheating as the other case is nearly identical.

We trivially have $Z_{A,0}\ket{\psi_{A,0}}=\beta \ket{\psi_{A,0}}$. 
The main constraint that we need to verify is
\be
Z_{A,i-1}\otimes I_\HM \geq U_{A,i}^\dagger E_{A} \lp( 
Z_{A,i}\otimes I_\HM \rp)
E_{A} U_{A,i}
\label{eq:TDPGdual}
\ee
\noindent
for $i$ odd. The special case of $i=n-1$ will be proven if we show that the
above inequality holds with $Z_{A,n-1}=Z_A$ because $Z_A\geq \Pi_{A,1}$ and
inequalities are transitive.

First we note that in all cases
\be
E_{A} \lp( Z_{A,i}\otimes I_\HM \rp) E_{A} = \bar{Z}_A \otimes I_\HBP
\equiv \sum_{s_a} s_a \ket{0,s_a,s_a}\bra{0,s_a,s_a}\otimes I_\HBP,
\ee
\noindent
where $\bar{Z}_A$ has support on $\HAB$. The unitary $U_{A,i}$ maps
$\HAB\otimes\HBP$ to itself, so the right hand side of
Eq.~(\ref{eq:TDPGdual}) has support on $\HAB\otimes\HBP$. The operator
$Z_{A,i-1}\otimes I_\HM$ is block diagonal with respect to the
decomposition of $\HA\otimes\HM$ into $\HAB\otimes\HBP$ and its complement,
and so the inequality is trivially satisfied in the latter space. In
$\HAB\otimes\HBP$ what remains to be shown is
\beq
\lp(\bar{Z}_A + \Lambda \sum_{s_a\in S_A} \ket{1,s_a,s_a}\bra{1,s_a,s_a}\rp)
\otimes I_\HBP \geq U_{A,i}^\dagger 
\lp( \bar{Z}_A \otimes I_\HBP \rp) U_{A,i} =
\sum_{s_b\in S_B} {U_{A,i}^{(s_b)}}^\dagger \bar{Z}_A U_{A,i}^{(s_b)}
\otimes\ket{s_b}\bra{s_b} 
\eeq
\noindent
and the inequality follows from the definition of $U_{A,i}^{(s_b)}$.

What we have proven is that we can take a TDPG with strictly valid
transitions and final point $[\beta,\alpha]$ and turn into a UBP with
projections with bound $(\beta,\alpha)$. However, in the last section we
proved that UBPs with projections come arbitrarily close to regular UBPs,
and at the top of this section we proved that TDPGs with strictly valid
transitions come arbitrarily close to any arbitrary TDPGs. Therefore, we
have proven that TDPGs are equivalent to UBPs, which we state formally as
an extended version of Theorem~\ref{thm:kitmain1}:

\begin{theorem}
Let $f(\beta,\alpha):\R\times\R\rightarrow\R$ be a function such that
$f(\alpha',\beta')\geq f(\alpha,\beta)$ whenever $\alpha'\geq\alpha$ and
$\beta'\geq\beta$, then
\be
\inf_{\text{proto}} f(P_B^*,P_A^*) = \inf_{UBP} f(\beta,\alpha)
= \inf_{TDPG} f(\beta,\alpha),
\ee
\noindent
where the left optimization is carried out over all coin-flipping protocols,
the middle one is carried out over all upper-bounded protocols, and the
right one is carried out over all time dependent point games.
\label{thm:kitmain2}
\end{theorem}

Note that the above construction of a UBP out of a TDPG is not optimal
in terms of resources. In particular, in most cases the number
of qubits needed could be drastically reduced. It is also not in general
true that if we start with a UBP, translate it into a TDPG and then
construct from it a UBP we will end up with the same one. In fact, given
two UBPs with the same underlying protocol but different dual feasible
points, if we converted into TDPGs and then back into UBPs, the resulting
protocols will be radically different!

\section{\label{sec:examples}The illustrated guide to point games}

The purpose of this section is to build an intuition about point games. In
the first part of the section we will classify a few of the simplest valid
transitions. These moves will be part of a basic toolbox of transitions
that will be used throughout the paper.

In the second part of this section we use our basic moves to build some
simple coin-flipping protocols, all described in the language of TDPGs. The
examples include the bias $1/\sqrt{2}-1/2$ protocol of Spekkens and Rudolph
\cite{Spekkens2002} and the author's bias $1/6$ protocol \cite{me2005}.

In this section we will make extensive use of the basis for functions with
finite support introduced in the last section. In particular, we use $[z]$
to denote a one variable function that evaluates to one at a fixed point
$z$, and is zero everywhere else. We also use $[x,y]$ to denote a
two-variable function that is one at $(x,y)$ and zero everywhere else.

\subsection{Basic moves}

We aim to systematically describe all non-trivial one-variable valid
transitions of the following forms: one points to one point, two points to
one point and one point to two points. The latter two respectively
generate all transitions of the form $n$ points to one point and one point
to $n$ points.

Nevertheless, these will not form a complete basis of all valid
transitions. Even two point to two point transitions contain moves that
cannot be generated by the above set. Ultimately, the most concise
description of the set of all valid transitions is as the dual to the cone
of operator monotone functions.

\subsubsection{Point raising}

All possible one point to one point transitions have the form
\be
p[z]\rightarrow p[z'],
\ee
\noindent
where we have already imposed the constraint of probability conservation,
and we assume $p>0$. We now need to impose the constraints $p f(z) \leq p
f(z')$ for all operator monotone functions. Using $f(z)=z$ we have obtain
the necessary condition
\be
z\leq z'.
\ee
It is also sufficient because all operator monotone functions are
monotonically increasing.

When working in a bipartite case we see that $p[x,y]\rightarrow p[x',y]$ is
valid if and only if $x\leq x'$ (and similarly with $p[x,y]\rightarrow
p[x,y']$ and $y\leq y'$). More generally, $p[x,y]\rightarrow p[x',y']$ is
transitively valid if and only if $x\leq x'$ and $y\leq y'$. In simpler
words: we can always move points upwards or rightwards but not downwards or
leftwards. We will call these moves point raising (even when we are moving
rightwards).

Note that the presence of extra unmoving points does not affect any of the
one variable transitions. The transition $p[z]\rightarrow p[z']$ is valid
if and only if $p[z] + \sum_i p_i[z_i] \rightarrow p[z'] + \sum_i p_i[z_i]$
is valid (where $\sum_i p_i[z_i]$ is any other set of points with positive
probability).

\subsubsection{Point merging}

All possible two point to one point transitions have the form
\be
p_1[z_1]+p_2[z_2]\rightarrow (p_1+p_2)[z'],
\ee
\noindent
where we have already imposed the constraint of probability conservation,
and we assume $p_1>0$ and $p_2>0$. We now need to impose the constraints
$p_1 f(z_1) + p_2 f(z_2) \leq (p_1+p_2) f(z')$ for all operator monotone
functions. Using $f(z)=z$ we have obtain the necessary condition
\be
\frac{p_1 z_1 + p_2 z_2}{p_1 + p_2} \leq z'.
\label{eq:pointmerge}
\ee
It is also sufficient because operator monotone functions are concave (a
property that can be checked directly on the extremal functions
$f(z)=\frac{\lambda z}{\lambda+z}$ for $\lambda\in(0,\infty)$).

When equality holds in Eq.~(\ref{eq:pointmerge}) we call the move point
merging. The more general case is simply generated by point merging
followed by point raising). In simpler words: point merging takes two
points and replaces them with a single point carrying their combined
probability and average $z$ value.

For $n$ points merging into one point it is easy to see that we must
conserve probability and average $z$ (or strictly speaking average $z$
cannot decrease). But this exact final configuration can be achieved using
a sequence of pairwise point merges with a possible point raising at the
end.

For the bipartite case, the transition
\be
p_1[x_1,y]+p_2[x_2,y]\rightarrow 
(p_1+p_2) \lp[\frac{p_1 x_1 + p_2 x_2}{p_1 + p_2},y\rp]
\ee
is clearly valid, and similarly with $x$ and $y$ interchanged. However, it
does not follow that the transition $0.5[0,0]+0.5[1,1]\rightarrow
1[0.5,0.5]$ is transitively valid. In fact, the proof of the impossibility
of strong coin flipping is a proof that $0.5[0,0]+0.5[1,1]\rightarrow
1[z,z]$ is transitively invalid if $z<1/\sqrt{2}$.

\subsubsection{Point splitting}

All possible one point to two point transitions have the form
\be
(p_1+p_2)[z] \rightarrow p_1[z_1']+p_2[z_2'],
\ee
\noindent
where we have already imposed the constraint of probability conservation,
and we assume $p_1>0$ and $p_2>0$. We now need to impose the constraints
$(p_1+p_2) f(z) \leq p_1 f(z_1') + p_2 f(z_2')$ for all operator monotone
functions. In particular, for $\lambda\in(0,\infty)$ the inequality is
satisfied for $f(z)=\frac{\lambda z}{\lambda+z} =
\lambda(1-\frac{\lambda}{\lambda+z})$ if and only if it is satisfied for
$f(z)=-\frac{1}{\lambda+z}$. If none of the points are located at zero,
then we can take the limit $\lambda\rightarrow 0$ and the inequality must
still be satisfied. In other words, a necessary condition is
\be
-\frac{p_1+p_2}{z} \leq - \frac{p_1}{z_1'}  - \frac{p_2}{z_2'}.
\label{eq:pointsplit}
\ee
It is also sufficient. Let $w=1/z$ and assume the above inequality holds,
then verifying the original constraint with $f(z)=-\frac{1}{\lambda+z}$ is
equivalent to verifying $(p_1+p_2) g(w)\leq (p_1+p_2) g(\frac{p_1 w_1 + p_2
w_2}{p_1 + p_2}) \leq p_1 g(w_1') + p_2 g(w_2')$ with
$g(w)=-\frac{w}{1+\lambda w}$. But the first inequality holds because
$g(w)$ is monotonically decreasing and the second inequality holds because
$g(w)$ is convex. The special case of $f(z)=z$ follows by considering
the limit $\lambda\rightarrow\infty$ of $f(z)=\frac{\lambda z}{\lambda+z}$.

When equality holds in Eq.~(\ref{eq:pointsplit}) we call the move point
splitting. In simpler words: point splitting takes a point and replaces it
with two points such that the total probability and average $1/z$ is
conserved.  All one point to two points valid transitions can then be
generated by point raising followed by point splitting. Similarly, all one
point to $n$ point valid transitions can be generated by point-raising and
a sequence of one-to-two point splittings.

The above arguments hold provided none of the points is located at zero. If
either $z_1'=0$ or $z_2'=0$ it is not hard to verify that we must also have
$z=0$. If $z=0$ all valid moves can be generated first by splitting the
point into two points located at zero and then using point raising
to move them to the required destination. In fact, we allow this type of
point splitting anywhere: we can always replace a point at $z$ with
probability $p$ with two points at $z$ and probabilities that add to $p$.

\subsubsection{Summary}

\begin{lemma}
The following are valid transitions:
\begin{itemize}
\item Point raising
\be
p[z]\rightarrow p[z']\qquad\qquad\text{(for $z\leq z'$).}
\ee
\item Point merging
\be
p_1[z_1]+p_2[z_2]\rightarrow \lp(p_1+p_2\rp)
\lp[\frac{p_1 z_1+ p_2 z_2}{p_1+p_2}\rp].
\ee
\item Point splitting
\be
\lp(p_1+p_2\rp)\lp[\frac{p_1+p_2}{p_1 w_1'+ p_2 w_2'}
\rp]\rightarrow p_1\lp[\frac{1}{w_1'}\rp]+p_2\lp[\frac{1}{w_2'}\rp].
\ee
\end{itemize}
\end{lemma}

\subsection{Basic protocols}

The simplest of all protocols are the ones when one player flips a coin
and tells the outcome to the other player. If Alice is in charge of
flipping the coin we get the TDPG
\be
\frac{1}{2}[1,0]+\frac{1}{2}[0,1] \quad\rightarrow\quad
\frac{1}{2}[1,1]+\frac{1}{2}[0,1] \quad\rightarrow\quad
1\lp[\frac{1}{2},1\rp],
\ee
\noindent
where the first move is point raising and the second is point
merging. Similarly, if Bob is in charge of flipping the coin we get
\be
\frac{1}{2}[1,0]+\frac{1}{2}[0,1] \quad\rightarrow\quad
\frac{1}{2}[1,0]+\frac{1}{2}[1,1] \quad\rightarrow\quad
1\lp[1,\frac{1}{2}\rp].
\ee
\noindent
These are graphically illustrated in Fig.~\ref{fig:prototriv}. 

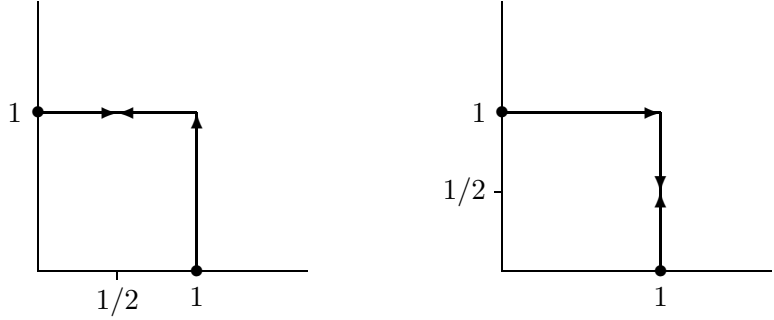
\begin{figure}[tb]
\begin{center}
\unitlength = 60pt
\begin{picture}(2.5,2.3)(-0.5,-0.5)
\put(0,0){\line(1,0){1.7}}
\put(0,0){\line(0,1){1.7}}
\put(0.5,0){\line(0,-1){0.05}}
\put(0.5,-0.1){\makebox(0,0)[t]{$1/2$}}
\put(1,-0.1){\makebox(0,0)[t]{$1$}}
\put(-0.1,1){\makebox(0,0)[r]{$1$}}
\thicklines
\put(0,1){\vector(1,0){0.5}}
\put(1,1){\vector(-1,0){0.5}}
\put(1,0){\vector(0,1){1}}
\put(1,0){\makebox(0,0){$\bullet$}}
\put(0,1){\makebox(0,0){$\bullet$}}
\end{picture}
\qquad
\begin{picture}(2.5,2.3)(-0.5,-0.5)
\put(0,0){\line(1,0){1.7}}
\put(0,0){\line(0,1){1.7}}
\put(0,0.5){\line(-1,0){0.05}}
\put(-0.1,0.5){\makebox(0,0)[r]{$1/2$}}
\put(1,-0.1){\makebox(0,0)[t]{$1$}}
\put(-0.1,1){\makebox(0,0)[r]{$1$}}
\thicklines
\put(1,0){\vector(0,1){0.5}}
\put(1,1){\vector(0,-1){0.5}}
\put(0,1){\vector(1,0){1}}
\put(1,0){\makebox(0,0){$\bullet$}}
\put(0,1){\makebox(0,0){$\bullet$}}
\end{picture}
\caption{The two trivial protocols where Alice (left) or Bob (right) flip a
coin and announce the outcome.}
\label{fig:prototriv}
\end{center}
\end{figure}

Note that in the second protocol Bob sends the first non-trivial message
(equivalently, the final move is vertical). Whereas for UBPs we enforced
the constraint that Alice always sent the first message, for TDPGs we will
let the first message be sent by whomever it is convenient.

To assign some meaning to the abstract operations we can think of point
merging as occurring when a player flips a coin and announces the outcome
to their opponent (recall the reverse time convention, in regular time
point merging occurs first and makes two points out of one). Point raising
does not correspond to any physical operation, but rather to one player
trusting or at least accepting the state provided by the other player.

\subsubsection{The Spekkens and Rudolph protocol}

Fix $x\in(1/2,1)$. Consider
\be
\frac{1}{2}[1,0]+\frac{1}{2}[0,1]
\quad\rightarrow\ &
\frac{2x-1}{2x}\bigg[x,0\bigg] + 
\frac{1-x}{2x}\bigg[\frac{x}{1-x},0\bigg] +
\frac{1}{2}[0,1]&
\\\nonumber
\quad\rightarrow\ &
\frac{2x-1}{2x}\bigg[x,0\bigg] + 
\frac{1-x}{2x}\bigg[\frac{x}{1-x},1\bigg] +
\frac{1}{2}[0,1]&
\\\nonumber
\quad\rightarrow\ &
\frac{2x-1}{2x}\bigg[x,0\bigg] + 
\frac{1}{2x}\bigg[x,1\bigg]&
\ \rightarrow\quad 
1 \bigg[x,\frac{1}{2x} \bigg].
\ee
\noindent
The TDPG is the sequence: split, raise, merge, merge. It is illustrated in
Fig.~\ref{fig:protosnr} for the case $x=1/\sqrt{2}$.  The resulting
protocol satisfies $P_B^*= x$ and $P_A^*= \frac{1}{2x}$, achieving
the tradeoff curve $P_A^* P_B^*=1/2$ from \cite{Spekkens2002}.

\begin{figure}[tb]
\begin{center}
\unitlength = 60pt
\begin{picture}(3.7,2.3)(-0.5,-0.5)
\put(0,0){\line(1,0){3.2}}
\put(0,0){\line(0,1){1.7}}
\put(0.707,0){\line(0,-1){0.05}}
\put(2.414,0){\line(0,-1){0.05}}
\put(0,0.707){\line(-1,0){0.05}}
\put(1,-0.1){\makebox(0,0)[t]{$1$}}
\put(-0.1,1){\makebox(0,0)[r]{$1$}}
\put(0.707,-0.1){\makebox(0,0)[t]{$\frac{1}{\sqrt{2}}$}}
\put(2.414,-0.1){\makebox(0,0)[t]{$1+\sqrt{2}$}}
\put(-0.1,0.707){\makebox(0,0)[r]{$\frac{1}{\sqrt{2}}$}}
\thicklines
\put(0,1){\vector(1,0){0.707}}
\put(2.414,1){\vector(-1,0){1.707}}
\put(1,0){\vector(1,0){1.414}}
\put(1,0){\vector(-1,0){0.293}}
\put(2.414,0){\vector(0,1){1}}
\put(0.707,0){\vector(0,1){0.707}}
\put(0.707,1){\vector(0,-1){0.293}}
\put(1,0){\makebox(0,0){$\bullet$}}
\put(0,1){\makebox(0,0){$\bullet$}}
\end{picture}
\caption{The Spekkens and Rudolph protocol with $x=1/\sqrt{2}$.}
\label{fig:protosnr}
\end{center}
\end{figure}
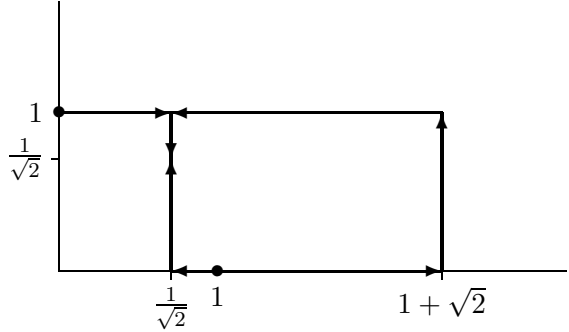

From the point of view of point games, the clever step above is the initial
split which was chosen so that, after the first merge, the remaining points
would be vertically aligned and a second merge could immediately take
place. The initial split corresponds to the cheat detection carried out at
the end of the protocol.

Another interpretation of the compromises made in the above protocol can be
understood as follows: We know that in each move the average value of $x$
and $y$ cannot decrease because $f(z)=z$ is operator monotone (and
$f(x,y)=x+y$ is bi-operator monotone). A perfect zero-bias protocol would
never increase these averages. In a non-perfect protocol every such
increase gets added to the final bias. In particular, the above protocol
has two such ``bad'' steps: the split (which increases average $x$) and the
raise (which increases average $y$). The protocol with
$P_A^*=P_B^*=1/\sqrt{2}$ balances these two effects so that they are equal.

\subsubsection{\label{sec:TDPG16}Quantum public-coin protocols}

The Spekkens and Rudolph protocol can be improved by using the TDPG
depicted by Fig.~\ref{fig:protome} (left). The idea is that we begin by
splitting the initial point on the vertical axis and use the resulting top
point to help raise the rightmost point. The operation on the rightmost
line becomes a point merging which preserves average $y$ instead of the old
point raising which increased average $y$. The cost of doing this, though,
is the split on the vertical axis (which increases average $y$) and the
point raising of the top point (which increases average $x$). Nevertheless,
when the parameters (probabilities and coordinates) are chosen properly the
above pattern results in a improvement.

Closer inspection shows that the added structure of the above protocol
relative to the Spekkens and Rudolph protocol is very similar to the added
structure of the Spekkens and Rudolph relative to the trivial protocol
where Bob announces the coin outcome. In fact, the process can be iterated
as depicted in Fig.~\ref{fig:protome} (right). The process begins by
splitting the two initial points into many points on the axes. Then point
raising is used on the rightmost point so that it is aligned with the
topmost point. The two points are merged and the resulting point ends up
lined up with the second-rightmost point. These two are again merged
producing a point that is lined up with the second-topmost point. All point
are merged in this fashion until a single point remains.

Obviously, the initial splits must be chosen with care so that all the
merges end up properly lined up. We will not describe here the precise
parameters needed to achieve this, though the details can be found in
\cite{me2005}. In fact, the paper describes an even larger family of
coin-flipping protocols which consisted of classical public-coin protocol
with quantum cheat detection. In the language of TDPGs all the protocols in
the family can be characterized as follows: First the point $P_B[1,0]$
splits horizontally into as many points as needed. Similarly the point
$P_A[0,1]$ splits vertically as needed. These steps are the cheat
detection. After that only point raising and merging are allowed, though in
any order or pattern desired. Sadly, the optimal bias that can be achieved
with protocols of this form is $1/6$, and is realized by the pattern from
Fig.~\ref{fig:protome} (right) in the limit of an arbitrarily large number
of merges.

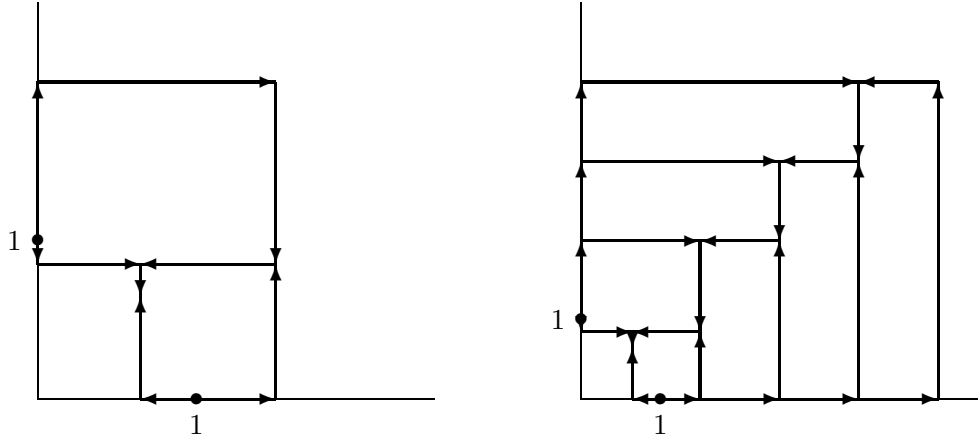
\begin{figure}[tb]
\begin{center}
\unitlength = 60pt
\begin{picture}(3,3)(-0.5,-0.5)
\put(0,0){\line(1,0){2.5}}
\put(0,0){\line(0,1){2.5}}
\put(1,-0.1){\makebox(0,0)[t]{$1$}}
\put(-0.1,1){\makebox(0,0)[r]{$1$}}
\thicklines
\put(0,0.85){\vector(1,0){0.65}}
\put(1.5,0.85){\vector(-1,0){0.85}}
\put(1,0){\vector(1,0){0.5}}
\put(1,0){\vector(-1,0){0.35}}
\put(0,2){\vector(1,0){1.5}}
\put(1.5,0){\vector(0,1){0.85}}
\put(1.5,2){\vector(0,-1){1.15}}
\put(0.65,0){\vector(0,1){0.65}}
\put(0.65,0.85){\vector(0,-1){0.20}}
\put(0,1){\vector(0,1){1}}
\put(0,1){\vector(0,-1){0.15}}
\put(1,0){\makebox(0,0){$\bullet$}}
\put(0,1){\makebox(0,0){$\bullet$}}
\end{picture}
\qquad
\unitlength = 30pt
\begin{picture}(6,6)(-1,-1)
\put(0,0){\line(1,0){5}}
\put(0,0){\line(0,1){5}}
\put(1,-0.2){\makebox(0,0)[t]{$1$}}
\put(-0.2,1){\makebox(0,0)[r]{$1$}}
\thicklines
\put(0,0.85){\vector(1,0){0.65}}
\put(1.5,0.85){\vector(-1,0){0.85}}
\put(1,0){\vector(1,0){0.5}}
\put(1,0){\vector(-1,0){0.35}}
\put(0,2){\vector(1,0){1.5}}
\put(1.5,0){\vector(0,1){0.85}}
\put(1.5,2){\vector(0,-1){1.15}}
\put(0.65,0){\vector(0,1){0.65}}
\put(0.65,0.85){\vector(0,-1){0.20}}
\put(0,1){\vector(0,1){1}}
\put(0,1){\vector(0,-1){0.15}}
\put(1,0){\makebox(0,0){$\bullet$}}
\put(0,1){\makebox(0,0){$\bullet$}}
\put(1.5,0){\vector(1,0){1}}
\put(2.5,0){\vector(0,1){2}}
\put(2.5,2){\vector(-1,0){1}}
\put(0,2){\vector(0,1){1}}
\put(0,3){\vector(1,0){2.5}}
\put(2.5,3){\vector(0,-1){1}}
\put(2.5,0){\vector(1,0){1}}
\put(3.5,0){\vector(0,1){3}}
\put(3.5,3){\vector(-1,0){1}}
\put(0,3){\vector(0,1){1}}
\put(0,4){\vector(1,0){3.5}}
\put(3.5,4){\vector(0,-1){1}}
\put(3.5,0){\vector(1,0){1}}
\put(4.5,0){\vector(0,1){4}}
\put(4.5,4){\vector(-1,0){1}}
\end{picture}
\caption{An improvement to the Spekkens and Rudolph protocol
(left) and further iterations of the improvement (right).
Figures not to scale.}
\label{fig:protome}
\end{center}
\end{figure}

An improved version of the above protocol is presented in
Appendix~\ref{sec:ddb}, where we effectively note that the initial splits
can be done gradually as the protocol progresses (or equivalently, that
cheat-detection can be done gradually). The advantage of this is a
reduction in the number of required qubits to a constant number. The bias,
though, remains fixed at $1/6$.

Below we intend to give an informal description of the bias $1/6$ protocol
in the limit of infinitely many messages, in which our standard ``finite
points with probability'' TDPGs get replaced by probability
densities. Though we will not formalize these TDPGs with probability
densities, they are occasionally useful in studying protocols. In fact, the
main result of the paper will have this form, though in the formal proof we
shall approximate the continuous distribution by a discrete finite set.

Let us imagine that we have carried out the initial point-splitting, and
that we have split into so many points that we effectively have a
continuous probability density on the axes:
\be
\frac{1}{2}\int_{z^*}^\infty p(z) [z,0] dz +
\frac{1}{2}\int_{z^*}^\infty p(z) [0,z] dz,
\ee
\noindent
where $p(z)$ is some probability distribution with $\int_{z^*}^\infty p(z)
dz=1$, and $z^*>0$ is some cutoff below which no points are located.

The continuum limit of the process depicted in Fig.~\ref{fig:protome}
(right) consists of a point moving along the diagonal and collecting the
probability density off of the axes. The point starts at $[\infty,\infty]$
with zero probability and ends at $[z^*,z^*]$ once it has collected all the
probability. What we are trying to determine is for what probability
distributions $p(z)$ is such as thing possible. In other words, for what
probability distributions $p(z)$ is
\be
\frac{1}{2}\int_{z^*}^\infty p(z) [z,0] dz +
\frac{1}{2}\int_{z^*}^\infty p(z) [0,z] dz
\quad\rightarrow\quad
1 [z^*,z^*]
\ee
transitively valid?

Let $Q(z)$ be the probability of the point that is traveling down the
diagonal. Given the point $Q(z)[z,z]$ we can move downwards and
rightwards in a two step process: first we merge with a ``point'' on the
$x$-axis (with effective probability $\frac{p(z)}{2}dz$) to get to $(Q(z) +
p(z)/2 dz)[z,z-dz]$, then we merge with a ``point'' on the $y$-axis (again
with effective probability $\frac{p(z)}{2}dz$) to get to $(Q(z)+p(z)
dz)[z-dz,z-dz]$. Conservation of probability tells us that $Q(z-dz) =
Q(z)+p(z) dz$ or
\be
\frac{dQ(z)}{dz} = -p(z).
\ee
\noindent
But these transitions are point merges and should additionally conserve
average height during the first merge (and average $x$ position during the
second merge). In particular, we get a constraint of the form $(Q(z)) z +
(\frac{p(z)}{2} dz) 0 = (Q(z)+\frac{p(z)}{2}dz) (z-dz)$. Canceling a few
terms we get $0=-Q(z)dz+\frac{p(z)}{2}z  dz$ or
\be
Q(z)=\frac{z p(z)}{2}.
\ee
\noindent
Combining the two constraints we get a differential equation in $Q(z)$
\be
\frac{dQ(z)}{dz} = - \frac{2Q(z)}{z}
\ee
\noindent
solved by $Q(z)=c/z^2$ for some constant $c$. Our original probability
distribution must have the form $p(z)=2c/z^3$, and we can now fix the
constant by the requirement $\int_{z^*}^\infty p(z) dz=1$. We get
$c=(z^*)^2$.

What our arguments above have shown is that for any $z^*>0$ the transition
\be
\frac{1}{2}\int_{z^*}^\infty \frac{2(z^*)^2}{z^3} [z,0] dz +
\frac{1}{2}\int_{z^*}^\infty \frac{2(z^*)^2}{z^3}[0,z] dz
\quad\rightarrow\quad
1 [z^*,z^*]
\label{eq:ddbmerge}
\ee
\noindent
is transitively valid.

But, of course, the true starting state is
$\frac{1}{2}[1,0]+\frac{1}{2}[0,1]$. To reach the initial state above we
must use point splitting independently on each axis. The constraints are
conservation of probability (already imposed) and a non-increasing average
$1/z$:
\be
1 \geq \int_{z^*}^\infty \frac{p(z)}{z} dz
=  \int_{z^*}^\infty \frac{2(z^*)^2}{z^4} dz = \frac{2}{3z^*}
\quad\Rightarrow\quad \frac{2}{3}\leq z^*.
\label{eq:ddbsplit}
\ee
\noindent
In other words, we can achieve bias $1/6$ but no better.

\subsubsection{\label{sec:cheat}Protocols with cheat detection}

As mentioned in the introduction, even bit commitment can be accomplished
using quantum information if we are willing to settle for a cheat-detecting
solution. These protocols with cheat detection may prove to be an important
component of quantum cryptography.

As an example of how Kitaev's formalism can be extended to study
cheat-detecting protocols we will describe in this section a simple
generalization of the above protocol.

One approach to cheat detection is as a payoff maximization
problem. Specifically, in coin flipping we would formulate the problem as
follows: winning the coin flip earns you \$1, loosing the coin flip nets
you \$0, but if you get caught cheating you will lose
\$$\Lambda$ (i.e., the cheating player wins $-\$\Lambda$ 
and we assume $\Lambda\geq 0$). Calculating the maximum expected earnings
for each $\Lambda$ is equivalent to finding the tradeoff curve between the
probability of winning by cheating and the probability of getting caught
cheating.

As most of our equations so far have been designed for positive
semidefinite matrices, it is simplest to begin by modifying our payouts so
that they are all non negative: \$$(\Lambda+1)$ for winning, \$$\Lambda$ for
losing, and \$0 for getting caught cheating. These two formulations are
equivalent in that an optimal payout for one can be found from the optimal
payout for the other by adding/subtracting $\Lambda$.

Accommodating multiple payouts only requires modifying the final projection
operators. For instance, $Z_{A,n}=\Pi_{A,1}$ (which could be though of as a
payout matrix for the non-cheat detecting problem), now needs to be replaced
with a matrix with three eigenvalues: $\Lambda+1$, $\Lambda$ and $0$ so
that the honest states in which Bob wins have the appropriate eigenvalue
\be
Z_{A,n} \lp( \Pi_{A,1}\ket{\psi_{A,n}}\rp) = 
(\Lambda+1) \lp( \Pi_{A,1}\ket{\psi_{A,n}}\rp),
\qquad\qquad
Z_{A,n} \lp( \Pi_{A,0}\ket{\psi_{A,n}}\rp) = 
\Lambda \lp( \Pi_{A,0}\ket{\psi_{A,n}}\rp),
\ee
\noindent
and the orthogonal subspace has eigenvalue zero. In particular, everything
we have done so far is still valid, but the transition of interest becomes
\be
\frac{1}{2}[\Lambda+1,\Lambda]+\frac{1}{2}[\Lambda,\Lambda+1]
\quad\rightarrow\quad
1[\beta,\alpha],
\ee
\noindent
where now $\beta$ and $\alpha$ are upper bounds on the expected payouts of
Bob and Alice respectively.

What would happen if we were now to shift back to the original payouts of
\$1, \$0 and \$$-\Lambda$? We would again return to looking for transitions
of the form $\frac{1}{2}[1,0]+\frac{1}{2}[0,1]\rightarrow 1[\beta,\alpha]$,
however our constraints for valid transitions would need to change. Rather
than looking for transitions that live in the dual to the cone of operator
monotone functions with domain $[0,\infty)$ we would look for transitions
in the dual to the cone of operator monotone functions with domain
$[-\Lambda,\infty)$.

In other words, in a setting with a penalty of $\Lambda$ for cheating, a
transition $p_i(z)\rightarrow p_{i+1}(z)$ is valid if and only if
probability is conserved and
\be
\sum_z p_i(z) \frac{\lambda z}{\lambda+z} \leq 
\sum_z p_{i+1}(z) \frac{\lambda z}{\lambda+z}
\ee
\noindent
for $\lambda\in(\Lambda,\infty)$. The old rule is simply the special case
$\Lambda=0$. 

Note that for $\Lambda>0$ we are simply removing restrictions from valid
transitions. Therefore, not surprisingly, any protocol that was valid with
no cheat detection is still valid in a cheat detecting world. However,
one can often do better in the latter case.

Returning to the coin-flipping protocol we described in the previous
section, Eq.~(\ref{eq:ddbmerge}) is still valid and essentially
optimal. However, the constraint imposed by Eq.~(\ref{eq:ddbsplit}) is no
longer necessary. It is not hard to check that in a cheat detecting
setting, the dominant function that constrains point splitting is
$f(z)=\frac{\Lambda z}{\Lambda+z}$ (or $f(z)=-\frac{1}{\Lambda+z}$ which is
equivalent when probability is conserved). Eq.~(\ref{eq:ddbsplit}) gets
replaced by
\be
-\frac{1}{\Lambda+1} \leq -\int_{z^*}^\infty \frac{p(z)}{\Lambda+z} dz
= -\int_{z^*}^\infty \frac{2(z^*)^2}{z^3(\Lambda+z)} dz
= -\lp(\frac{\Lambda-2z^*}{\Lambda^2} + 
\frac{2(z^*)^2}{\Lambda^3}\log\frac{z^*+\Lambda}{z^*}\rp).
\ee
\noindent
Asymptotically, as $\Lambda\rightarrow\infty$, one finds an optimal
expected winnings of $z^*\sim\frac{1}{2}+\frac{\log\Lambda}{4\Lambda}$. 

\section{\label{sec:Kit2}Kitaev's second coin-flipping formalism (cont.)}

In this section we shall conclude the description, that begun in
Section~\ref{sec:Kit}, of Kitaev's second coin-flipping formalism
\cite{Kit04}. The last step is the transition from points games that are
ordered in time to point games with no explicit time ordering.

\subsection{\label{sec:TIPG}Time Independent Point Games}

The new ingredient for this section is catalyst states. Given a
transitively valid transition such as $P_B[1,0]+P_A[0,1]\rightarrow
1[\beta,\alpha]$ it trivially follows that
\be
P_B[1,0]+P_A[0,1]+\sum_i w_i [x_i,y_i]
\rightarrow 1[\beta,\alpha]+\sum_i w_i [x_i,y_i]
\ee
is also transitively valid, for any ``catalyst'' state $\sum_i w_i
[x_i,y_i]$ with $w_i,x_i,y_i\geq 0$. The question we consider here is
whether the converse is true.

For one-variable transitions the converse is trivially true. The goal of
this section is to prove that the converse is also true for bipartite
transitions (including transitively valid transitions). The proof will
basically show that we can use a small amount of probability to create the
catalyst state, and then run the catalyzed transition in small enough
steps so that by comparison the catalyst state appears large
enough.

Before diving into the proof let us examine some of the surprising
consequences that will follow. The first consequence is that previously all
our probability distributions had a range contained in $[0,\infty)$, but
now we can allow ``probability'' distributions with a range of
$(-\infty,\infty)$. Negative values are simply points were we need to add
in some more probability using a catalyst state. This leads to the
following definition

\begin{definition}
A function with finite support $p:[0,\infty)\rightarrow\R$ is
\textbf{valid} if $\sum_z p(z)=0$ and  $\sum_z \lp(\frac{-1}{\lambda +
z}\rp)p(z)\geq 0$ for all $\lambda>0$.
\end{definition}

The definition implies $p$ is in the dual to the cone of operator monotone
functions. Note that because $\frac{\lambda z}{\lambda + z}= \lambda -
\frac{\lambda^2}{\lambda + z}$, and because of conservation of probability,
checking $\sum_z \lp(\frac{\lambda z}{\lambda + z}\rp)p(z)\geq 0$ for all
$\lambda>0$ is equivalent to checking $\sum_z \lp(\frac{-1}{\lambda +
z}\rp)p(z)\geq 0$ for all $\lambda>0$. The latter condition will be easier
to analyze in later sections, though.

The validity relation is essentially a partial order on functions. Instead
of saying $p\rightarrow q$ is valid we could equally write $p\prec q$.
Similarly, a valid function $p$ could be written as $0\prec p$. We use the
earlier notation because it is easier to say ``a function $p$ is valid''
than ``a function $p$ belongs to the cone dual to the operator monotone
functions.'' 

By construction any valid transition $p\rightarrow q$ can be converted into
the valid function $q-p$. We therefore immediately obtain a number of
standard valid functions (which follow from the proofs in the previous
section):

\begin{lemma}
The following are valid functions:
\begin{itemize}
\item Point raising
\be
-p[z]+p[z']\qquad\qquad\text{(for $z\leq z'$).}
\ee
\item Point merging
\be
-p_1[z_1]-p_2[z_2]+ \lp(p_1+p_2\rp)
\lp[\frac{p_1 z_1+ p_2 z_2}{p_1+p_2}\rp].
\ee
\item Point splitting
\be
-\lp(p_1+p_2\rp)\lp[\frac{p_1+p_2}{p_1 w_1'+ p_2 w_2'}
\rp]+ p_1\lp[\frac{1}{w_1'}\rp]+p_2\lp[\frac{1}{w_2'}\rp].
\ee
\end{itemize}
\label{lemma:rms}
\end{lemma}

The second, and more surprising, consequence of catalyst states is that all
point games can be run using exactly two transitions: one vertical and one
horizontal. The idea is that given any point game we can move all the
horizontal transitions to the beginning and all the vertical transitions to
the end, combining each set into a single vertical or horizontal
transition. Of course, the state after the first transition but before the
second may have some negative probabilities, but as discussed above this
can be fixed with an appropriate catalyst state.  This leads to the
following definition:

\begin{definition}
A function with finite support $p:[0,\infty)\times[0,\infty)\rightarrow\R$ is
\textbf{valid} if either
\begin{itemize}
\item for every $c\in[0,\infty)$ the function $p(z,\underline c)$ is valid,
or
\item for every $c\in[0,\infty)$ the function $p(\underline c,z)$ is valid.
\end{itemize}
where as before $p(z,\underline c)$ is the one-variable function obtained
by fixing the second input. We call the first case a \textbf{valid horizontal}
function and the second case a \textbf{valid vertical} functions.
\end{definition}

Of course, we don't even need to specify whether the horizontal or the
vertical transition occurred first, and therefore we obtain a fully
time independent point game:

\begin{definition}
A \textbf{time independent point game (TIPG)} consists of a pair of
functions with finite support $h,v:[0,\infty)\times[0,\infty)\rightarrow
\R$ such that
\begin{itemize}
\item $h$ is a valid horizontal function.
\item $v$ is a valid vertical function.
\item $h+v =  1[\beta,\alpha] - P_B[1,0] - P_A[0,1]$.
\end{itemize}
We say that $[\beta,\alpha]$ is the final point of the TIPG.
\end{definition}

\subsubsection{Relating TDPGs and TIPGs}

Given a TDPG specified as $p_0,\dots,p_n$ with final point $[\beta,\alpha]$
we shall construct a TIPG with the same final point. Let $H$ (resp $V$) be
the set of indices $1,\dots,n$ such that $p_{i-1}\rightarrow p_i$ is a
valid horizontal (resp. vertical) transition. Define
\be
h = \sum_{i\in H} (p_i-p_{i-1}),
\qquad\qquad
v = \sum_{i\in V} (p_i-p_{i-1}),
\ee
then $h$ is a valid horizontal function, $v$ is a valid vertical function
and
\be
h+v = p_n-p_0 = 1[\beta,\alpha] - P_B[1,0] - P_A[0,1]
\ee
\noindent
as required.

To go the other way we begin with a TIPG specified by $h,v$ with final
point $[\beta,\alpha]$. We define $v^-(x,y) = -\min(v(x,y),0)\geq 0$ as the
magnitude of the negative part of $v$. Then consider
\be
P_B[1,0] + P_A[0,1] + v^- &\rightarrow& P_B[1,0] + P_A[0,1] + v^- +v
\\\nonumber
&\rightarrow& P_B[1,0] + P_A[0,1] + v^- + v + h \ =\  1[\beta,\alpha] + v^-.
\ee
The first is a valid vertical transition and the second is a valid
horizontal transition. Also, the intermediate state is non-negative with
finite support. Therefore, it is a proof that $P_B[1,0] + P_A[0,1] + v^-
\rightarrow 1[\beta,\alpha] + v^-$ is transitively valid. Below we will
show that we can get rid of the catalyst $v^-$ and construct for every
$\epsilon>0$ a sequence such that $P_B[1,0] + P_A[0,1]
\rightarrow 1[\beta+\epsilon,\alpha+\epsilon]$ is transitively valid. This
is the desired TDPG.

What remains to be proven is that we can discard catalyst states. We will
only prove this for the special case of coin flipping, though the general
case is also true. We first require two simple lemmas. The first lemma
shows we can construct arbitrary catalyst states (as long as we allow
extra junk) and the second lemma shows we can clean up the catalyst states
(and extra junk).

\begin{lemma}
Given a function $r:[0,\infty)\times[0,\infty)\rightarrow [0,\infty)$ with
finite support such that $r(0,0)=0$, there exists $c>0$ and a function
$r':[0,\infty)\times[0,\infty)\rightarrow [0,\infty)$ with finite support
such that
\be
c P_B[1,0]+c P_A[0,1]\rightarrow r+r'
\ee
\noindent
is transitively valid, where we additionally assume both $P_A,P_B>0$.
\label{lemma:catbegin}
\end{lemma}

\begin{proof}
First we prove the lemma for the special case when $r$ has support at a
single point. Let $r=q[x,y]$ with $q,x>0$ and $y\geq0$ (the case $y>0$ and
$x\geq0$ will follow by exchanging the axes). If $x\geq 1$ then
\be
\frac{q}{P_B}P_B[1,0]+ \frac{q}{P_B}P_A[0,1]\rightarrow q[x,y] + 
\frac{q}{P_B}P_A[0,1]
\ee
is transitively valid using point raisings, and the lemma is satisfied with
$c=\frac{q}{P_B}$ and $r'=\frac{q P_A}{P_B}[0,1]$. If $x<1$ the sequence
\be
P_B[1,0]+P_A[0,1]\rightarrow P_B[1,y]+P_A[0,1] \rightarrow \frac{x}{2}P_B[x,y]
+ \lp(1-\frac{x}{2}\rp)P_B[2-x,y]+P_A[0,1]
\ee 
is transitively valid by point raising followed by point splitting. Now we
use the fact that we can scale the probability in transitions. That is,
if $\sum_i w_i [x_i,y_i]\rightarrow\sum_j w_j' [x_j',y_j']$
is transitively valid then for $a>0$ so is $\sum_i a w_i
[x_i,y_i]\rightarrow\sum_j a w_j' [x_j',y_j']$. In particular, if we scale
the previous transition by $c=2q/(x P_B)$ we satisfy the lemma with
$r' = c(1-x/2)P_B[2-x,y]+c P_A[0,1]$.

Finally, for the general case where $r=\sum_{i=1}^k q_i [x_i,y_i]$ let
$c_i$ and $r'_i$ be chosen as above so that $c_i P_B[1,0]+c_i
P_A[0,1]\rightarrow q_i [x_i,y_i]+r_i'$ is transitively valid. Then
\be
\sum_{i=1}^k c_i P_B[1,0]+ \sum_{i=1}^k c_i P_A[0,1]&\rightarrow&
\sum_{i=2}^k c_i P_B[1,0]+ 
\sum_{i=2}^k c_i P_A[0,1] + q_1 [x_1,y_1] + r_1'
\nonumber\\
&\rightarrow&\cdots\rightarrow
\sum_{i=j}^k c_i P_B[1,0]+ 
\sum_{i=j}^k c_i P_A[0,1] + \sum_{i=1}^{j-1} q_i [x_i,y_i] + 
\sum_{i=1}^{j-1} r_i'
\nonumber\\
&\rightarrow&\cdots\rightarrow
r + \sum_{i=1}^{k} r_i',
\ee
where by construction each transition is transitively valid and so the
proof is completed by setting $c=\sum_i c_i$ and $r'=\sum_i r_i'$.
\end{proof}

\begin{lemma}
Given $\epsilon>0$ and a function
$r''\rightarrow:[0,\infty)\times[0,\infty)\rightarrow [0,\infty)$ with
finite support and $\sum_{x,y} r''(x,y)=1$, there exists $1>\delta>0$ such
that
\be
(1-\delta)[\beta,\alpha]+\delta\, r'' \rightarrow 
1[\beta+\epsilon,\alpha+\epsilon]
\ee
\noindent
is transitively valid.
\label{lemma:catend}
\end{lemma}

\begin{proof}
Let $x''$ be the largest $x$-coordinate over all points in $r''$ and
similarly let $y''$ be the largest $y$-coordinate over all points in $r''$.
By point raising $r\rightarrow 1[x'',y'']$ is transitively valid, so we
focus on proving that we can find $1>\delta>0$ such that
\be
(1-\delta)[\beta,\alpha]+\delta[x'',y''] \rightarrow 
1[\beta+\epsilon,\alpha+\epsilon]
\ee
is transitively valid. We can also assume (possibly using further point
raisings) that $x''>\beta+\epsilon$ and $y''>\alpha+\epsilon$. Consider the
following sequence
\be
(1-\delta)[\beta,\alpha]+\delta[x'',y''] &\rightarrow&
(1-\delta')[\beta,\alpha]+(\delta'-\delta)[\beta,y'']+ \delta[x'',y'']
\nonumber\\ &\rightarrow&
(1-\delta')[\beta+\epsilon,\alpha]+
\delta'[\beta+\epsilon,y'']
\nonumber\\ &\rightarrow&
1[\beta+\epsilon,\alpha+\epsilon]
\ee
\noindent
corresponding to raise, merge, merge. To make the first merge valid
we need $(\delta'-\delta)\beta+\delta x''= \delta'(\beta+\epsilon)$
which is equivalent to $\delta(x''-\beta)=\delta'\epsilon$.
The second merge requires $\delta'(y''-\alpha)=\epsilon$. Both conditions
can be satisfied by constants such that $1>\delta'>\delta>0$.
\end{proof}

Putting the lemmas together we can prove the main result.

\begin{lemma}
Given $r:[0,\infty)\times[0,\infty)\rightarrow
[0,\infty)$ with finite support such that $r(0,0)=0$ and 
\be
P_B[1,0]+P_A[0,1]+r\rightarrow 1[\beta,\alpha] +r
\ee
\noindent
is transitively valid then for every $\epsilon>0$ 
\be
P_B[1,0]+P_A[0,1]\rightarrow 1[\beta+\epsilon,\alpha+\epsilon]
\ee
is transitively valid as well.
\end{lemma}

\begin{proof}
Fix $\epsilon>0$. Note that by conservation of probability we must have
$P_A+P_B=1$. We also assume $P_A,P_B>0$ as otherwise the proof is trivial.
Therefore, we can use Lemma~\ref{lemma:catbegin} to get $c>0$ and $r'$
so that $c P_B[1,0]+c P_A[0,1]\rightarrow r+r'$ is transitively valid.

If we set $r''=(1/c)(r+r')$ than again by conservation of probability we
must have $\sum_{x,y}r''=1$. We can therefore use Lemma~\ref{lemma:catend}
to get $1>\delta>0$ so that $(1-\delta)[\beta,\alpha]+\delta r''
\rightarrow 1[\beta+\epsilon,\alpha_\epsilon]$ is transitively valid. Now
consider the sequence
\be
P_B[1,0]+P_A[0,1]
&\rightarrow&
(1-\delta)P_B[1,0]+(1-\delta)P_A[0,1]+\frac{\delta}{c}r + +\frac{\delta}{c}r'
\nonumber\\&\rightarrow&
(1-\delta)[\beta,\alpha]+\frac{\delta}{c}r + +\frac{\delta}{c}r'
\nonumber\\&\rightarrow&
1[\beta+\epsilon,\alpha+\epsilon].
\ee
The first transition follows from a scaled version of the transition
obtained from Lemma~\ref{lemma:catbegin}, and the third transition is the
one obtained from Lemma~\ref{lemma:catend}. The middle transition follows
from repeated applications of $a P_B[1,0]+a P_A[0,1]+a r\rightarrow
a[\beta,\alpha] + a r$ for $a\leq \delta/c$. Therefore the whole transition
is transitively valid as required.
\end{proof}

From the conditions of validity, it is easy to verify that neither $h$ nor
$v$ can be positive at the point $(0,0)$. Since their sum must be zero at
this point, individually they must be zero as well. Hence, our catalyst
state is zero at $(0,0)$ and we can apply the above lemma to complete our
argument for the equivalence of TDPGs and TIPGs. We can state the result
formally as a further extension of Theorem~\ref{thm:kitmain2}.

\begin{theorem}
Let $f(\beta,\alpha):\R\times\R\rightarrow\R$ be a function such that
$f(\alpha',\beta')\geq f(\alpha,\beta)$ whenever $\alpha'\geq\alpha$ and
$\beta'\geq\beta$, then
\be
\inf_{\text{proto}} f(P_B^*,P_A^*) = \inf_{UBP} f(\beta,\alpha)
= \inf_{TDPG} f(\beta,\alpha) = \inf_{TIPG} f(\beta,\alpha).
\ee
\label{thm:kitmain3}
\end{theorem}

\section{\label{sec:zero}Towards zero bias}

In this section we will finally describe a family of protocols for
coin-flipping that achieves arbitrarily small bias. The protocols will be
described in Kitaev's second formalism using the tools of the previous
sections.

In the first part of the section we will attempt to give some intuition for
our construction. Along the way we will prove a couple of important
lemmas. In Section~\ref{sec:formal} we will simply present the
corresponding pair of functions $h$ and $v$ and prove that they satisfy the
necessary properties.

For those who have skipped ahead to this section, we review some of the key
concepts that have been defined in previous sections: A valid function
$f(z)$ has finite support and satisfies $\sum_z f(z)=0$ and $\sum_z
\frac{-1}{\lambda+z}f(z)\geq 0$ for all $\lambda>0$. These constraints are
equivalent to those discussed in the introduction. Examples of valid
functions are point raises, point merges and point splits as defined in
Lemma~\ref{lemma:rms}. A valid horizontal function $h(x,y)$ is valid as a
function of $x$ for every $y\geq 0$. Similarly, a valid vertical function
$v(x,y)$ is valid as a function of $y$ for every $x\geq 0$. A TIPG is a
valid horizontal function plus a valid vertical function such that $h+v =
1[\beta,\alpha] -P_B[1,0] -P_A[0,1]$, where $[x_0,y_0]$ denotes a function
that takes value one at $x=x_0$, $y=y_0$ and is zero everywhere else. Such
a TIPG leads to a coin-flipping protocol with $P_A^*\leq\alpha$ and
$P_B^*\leq\beta$.

\subsection{Guiding principles}

We begin our discussion with a TIPG example. We will analyze the TIPG with
bias $1/6$ first introduced in the introduction (Fig.~\ref{fig:introlad16})
and reproduced here in Fig.~\ref{fig:lad16}. The two figures are intended
to denote the same TIPG, though the latter figure has a slightly different
labeling convention. 

The new labeling convention provides some useful intuition for TIPGs:
because of probability conservation, one can associate a probability with
each arrow. The arrows carry the probability from their base to their
head. The probability associated to a point can be computed as the sum of
incoming probabilities, (which equals the sum of outgoing probabilities for
all points except the initial and final points). Furthermore, the
probability carried by arrows not associated with the initial or final
points must go around in loops, such as boxes or figure eights. It is often
easiest in a figure to label each loop with the probability that goes
around it, and this is the idea behind the labeling of
Fig.~\ref{fig:lad16}.

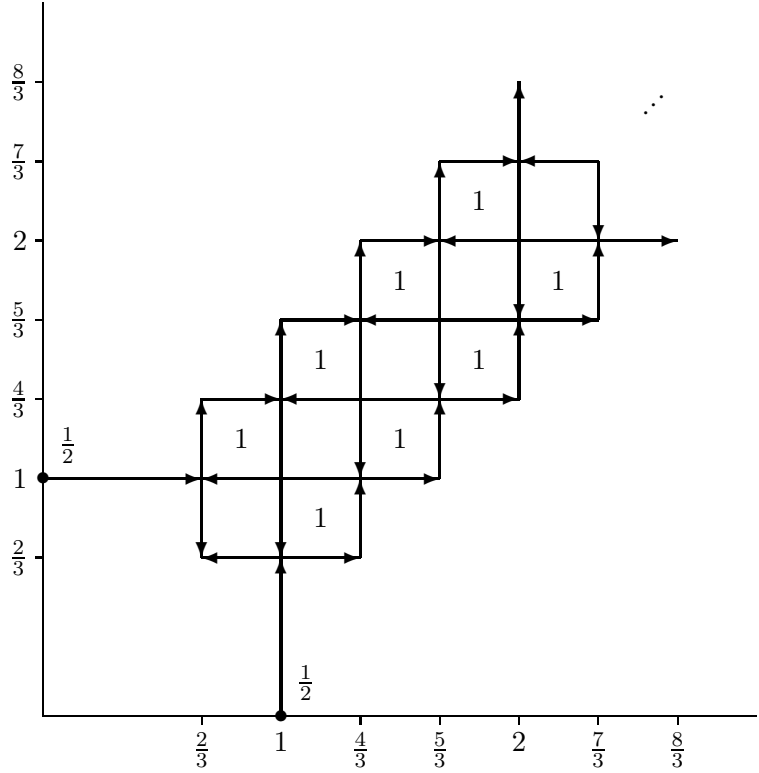
\begin{figure}[tb]
\begin{center}
\unitlength = 30pt
\begin{picture}(11,10)(-1,-1)
\put(0,0){\line(1,0){9}}
\put(0,0){\line(0,1){9}}
\put(2,0){\line(0,-1){0.1}}
\put(4,0){\line(0,-1){0.1}}
\put(5,0){\line(0,-1){0.1}}
\put(6,0){\line(0,-1){0.1}}
\put(7,0){\line(0,-1){0.1}}
\put(8,0){\line(0,-1){0.1}}
\put(0,2){\line(-1,0){0.1}}
\put(0,4){\line(-1,0){0.1}}
\put(0,5){\line(-1,0){0.1}}
\put(0,6){\line(-1,0){0.1}}
\put(0,7){\line(-1,0){0.1}}
\put(0,8){\line(-1,0){0.1}}
\put(2,-0.2){\makebox(0,0)[t]{$\frac{2}{3}$}}
\put(3,-0.2){\makebox(0,0)[t]{$1$}}
\put(4,-0.2){\makebox(0,0)[t]{$\frac{4}{3}$}}
\put(5,-0.2){\makebox(0,0)[t]{$\frac{5}{3}$}}
\put(6,-0.2){\makebox(0,0)[t]{$2$}}
\put(7,-0.2){\makebox(0,0)[t]{$\frac{7}{3}$}}
\put(8,-0.2){\makebox(0,0)[t]{$\frac{8}{3}$}}
\put(-0.2,2){\makebox(0,0)[r]{$\frac{2}{3}$}}
\put(-0.2,3){\makebox(0,0)[r]{$1$}}
\put(-0.2,4){\makebox(0,0)[r]{$\frac{4}{3}$}}
\put(-0.2,5){\makebox(0,0)[r]{$\frac{5}{3}$}}
\put(-0.2,6){\makebox(0,0)[r]{$2$}}
\put(-0.2,7){\makebox(0,0)[r]{$\frac{7}{3}$}}
\put(-0.2,8){\makebox(0,0)[r]{$\frac{8}{3}$}}
\put(3.2,0.2){\makebox(0,0)[bl]{$\frac{1}{2}$}}
\put(0.2,3.2){\makebox(0,0)[bl]{$\frac{1}{2}$}}
\put(3.5,2.5){\makebox(0,0){$1$}}
\put(4.5,3.5){\makebox(0,0){$1$}}
\put(5.5,4.5){\makebox(0,0){$1$}}
\put(6.5,5.5){\makebox(0,0){$1$}}
\put(2.5,3.5){\makebox(0,0){$1$}}
\put(3.5,4.5){\makebox(0,0){$1$}}
\put(4.5,5.5){\makebox(0,0){$1$}}
\put(5.5,6.5){\makebox(0,0){$1$}}
\thicklines
\put(3,0){\vector(0,1){2}}
\put(0,3){\vector(1,0){2}}
\put(3,2){\vector(-1,0){1}}
\put(2,3){\vector(0,-1){1}}
\put(3,2){\vector(1,0){1}}
\put(2,3){\vector(0,1){1}}
\put(4,3){\vector(-1,0){2}}
\put(3,4){\vector(0,-1){2}}
\put(2,4){\vector(1,0){1}}
\put(4,2){\vector(0,1){1}}
\put(4,3){\vector(1,0){1}}
\put(3,4){\vector(0,1){1}}
\put(5,4){\vector(-1,0){2}}
\put(4,5){\vector(0,-1){2}}
\put(3,5){\vector(1,0){1}}
\put(5,3){\vector(0,1){1}}
\put(5,4){\vector(1,0){1}}
\put(4,5){\vector(0,1){1}}
\put(6,5){\vector(-1,0){2}}
\put(5,6){\vector(0,-1){2}}
\put(4,6){\vector(1,0){1}}
\put(6,4){\vector(0,1){1}}
\put(6,5){\vector(1,0){1}}
\put(5,6){\vector(0,1){1}}
\put(7,6){\vector(-1,0){2}}
\put(6,7){\vector(0,-1){2}}
\put(5,7){\vector(1,0){1}}
\put(7,5){\vector(0,1){1}}
\put(7,6){\vector(1,0){1}}
\put(6,7){\vector(0,1){1}}
\put(7,7){\vector(-1,0){1}}
\put(7,7){\vector(0,-1){1}}
\put(7.6,7.6){\makebox(0,0){$\cdot$}}
\put(7.7,7.7){\makebox(0,0){$\cdot$}}
\put(7.8,7.8){\makebox(0,0){$\cdot$}}
\put(3,0){\makebox(0,0){$\bullet$}}
\put(0,3){\makebox(0,0){$\bullet$}}
\end{picture}
\caption{A TIPG with bias $1/6$.}
\label{fig:lad16}
\end{center}
\end{figure}

We can also use our algebraic notation to express the TIPG. The horizontal
arrows encode the function
\be
h &=&
+\frac{1}{2}\bigg[\frac{2}{3},\frac{2}{3}\bigg]
-\frac{3}{2}\bigg[1,\frac{2}{3}\bigg]
+1\bigg[\frac{4}{3},\frac{2}{3}\bigg]
\\\nonumber&&
-\frac{1}{2}\bigg[0,1\bigg]
+\frac{3}{2}\bigg[\frac{2}{3},1\bigg]
-2\bigg[\frac{4}{3},1\bigg]
+1\bigg[\frac{5}{3},1\bigg]
\\\nonumber&&
+\sum_{k=4}^\infty\lp(
-1\bigg[\frac{k-2}{3},\frac{k}{3}\bigg]
+2\bigg[\frac{k-1}{3},\frac{k}{3}\bigg]
-2\bigg[\frac{k+1}{3},\frac{k}{3}\bigg]
+1\bigg[\frac{k+2}{3},\frac{k}{3}\bigg]
\rp).
\ee

The last line of $h$ denotes a pattern that we shall call a ladder. More
generally, we shall refer to any regular pattern that heads to infinity on
the diagonal as a ladder. The problem with ladders is that they involve an
infinite number of points, and we previously required that $h$ have support
only on a finite set of points. In Section~\ref{sec:trunc} we will study
how to properly truncate these ladders, but for the moment we shall ignore
this issue.

By symmetry we can define $v(x,y)=h(y,x)$. By construction all terms in the
sum $h+v$ cancel except for the three required
\be
h+v = 1\bigg[\frac{2}{3},\frac{2}{3}\bigg]-
\frac{1}{2}\bigg[1,0\bigg]-\frac{1}{2}\bigg[0,1\bigg].
\ee

Let us verify that $h$ is a valid horizontal functions (which by symmetry
proves that $v$ is a valid vertical function). It is easy to see that for
every $y$ we have $\sum_x h(x,y)=0$. The other constraint that needs to be
checked is $\sum_x \frac{-1}{\lambda+x} h(x,y)\geq 0$ for all $\lambda>0$
and all $y\geq 0$. We begin with the case $y=k/3\geq 4/3$.
\be
\sum_{x} \frac{-1}{\lambda+x} h\lp(x,\,\frac{k}{3}\rp) &=&
\frac{1}{\lambda + \frac{k-2}{3}}
-\frac{2}{\lambda + \frac{k-1}{3}}
+\frac{2}{\lambda + \frac{k+1}{3}}
-\frac{1}{\lambda + \frac{k+2}{3}}
\nonumber\\
&=&
\frac{\frac{1}{3}}{(\lambda + \frac{k-2}{3})(\lambda + \frac{k-1}{3})}
-\frac{\frac{2}{3}}{(\lambda + \frac{k-1}{3})(\lambda + \frac{k+1}{3})}
+\frac{\frac{1}{3}}{(\lambda + \frac{k+1}{3})(\lambda + \frac{k+2}{3})}
\nonumber\\
&=&
\frac{\frac{1}{3}}{(\lambda + \frac{k-2}{3})(\lambda + \frac{k-1}{3})
(\lambda + \frac{k-1}{3})}
-\frac{\frac{1}{3}}{(\lambda + \frac{k-1}{3})(\lambda + \frac{k+1}{3})
(\lambda + \frac{k+2}{3})}
\nonumber\\
&=&
\frac{(\frac{1}{3})(\frac{4}{3})}
{(\lambda + \frac{k-2}{3})(\lambda + \frac{k-1}{3}) 
(\lambda + \frac{k-1}{3})(\lambda + \frac{k+2}{3})}
\geq 0
\label{eq:sampcons}
\ee
\noindent
for $\lambda> 0$, where the successive simplifications involves splitting
the middle terms and using relations of the form
\be
\frac{1}{\lambda+x_1}-\frac{1}{\lambda+x_2}=
\frac{x_2-x_1}{(\lambda+x_1)(\lambda+x_2)}.
\ee
\noindent
The idea is that we can interpret successive lines as follows: The first
line is the standard sum over points with a numerator corresponding to
probabilities. The second line is a sum over arrows with a numerator
corresponding to probabilities times distance traveled (which we can think
of as momentum). Finally, the third line is a sum over pairs of arrows with
zero net momentum.

The constraint at $y=1$ is roughly the same, except that the leftmost arrow
travels twice the distance and carries half the probability. More
specifically, we can write
\be
\label{eq:ladnsplit}
&&-\frac{1}{2}\bigg[0\bigg]
+\frac{3}{2}\bigg[\frac{2}{3}\bigg]
-2\bigg[\frac{4}{3}\bigg]
+1\bigg[\frac{5}{3}\bigg]
\\\nonumber
&&\qquad\qquad\qquad\qquad=
\lp(-\frac{1}{2}\bigg[0\bigg]+1\bigg[\frac{1}{3}\bigg]
- \frac{1}{2}\bigg[\frac{2}{3}\bigg]\rp)
+\lp(-1\bigg[\frac{1}{3}\bigg]+ 2\bigg[\frac{2}{3}\bigg]
-2\bigg[\frac{4}{3}\bigg]+1\bigg[\frac{5}{3}\bigg]\rp).
\ee
The right term is exactly what would be there if the ladder had been
extended to $y=1$. It is valid by Eq.~(\ref{eq:sampcons}). The left term is
the difference between the long arrow carrying probability $1/2$ and the
short arrow carrying probability $1$. It is valid because it is a point
merge. Therefore, the original expression is a sum of two valid terms and
itself is valid, as can also be checked by direct computation.

The constraint for $y=\frac{2}{3}$ can also be directly checked, though in
fact it is just a point splitting, and hence valid because
$\frac{3}{2}/1=\frac{1}{2}/\frac{2}{3}+1/\frac{4}{3}$.

Except for the fact that the ladder involves an infinite number of points,
we have completed the proof that the resulting TIPG corresponds to a
protocol with bias $1/6$. The infinite number of points, though, is a
serious problem from the point of view of our constructive description of
Kitaev's formalism: for instance, the canonical catalyst state used to
convert the TIPG into a TDPG carries an infinite amount of probability.

To proceed we therefore must truncate the ladder at some large distance
$\Gamma$. The truncation will add small extra terms to the bias, which
will go to zero as $\Gamma\rightarrow\infty$. We can think of the
different values for $\Gamma$ as a family of protocols which converges to
a bias of $1/6$. 

While the formal truncation is done in Section~\ref{sec:trunc} we will try
to paint an intuitive picture here as to why truncation is possible. Let us
imagine that we naively cut the ladder diagonally at some point in such a
way that the end looks like the top of Fig.~\ref{fig:lad16}. The ladder is
still valid (the top rung is just a point merge), however we are left with
an excess of probability in the edges and a deficit of probability in the
center. To correct the situation we need to add in a term of the form
$2[\frac{\Gamma}{3},\frac{\Gamma}{3}]
-1[\frac{\Gamma+1}{3},\frac{\Gamma-1}{3}]
-1[\frac{\Gamma-1}{3},\frac{\Gamma+1}{3}]$, which is a coin-flipping
problem!

Admittedly it is a coin-flipping problem with twice the probability and two
thirds of the distance between points but that does not make much of a
difference. The problem is located far away from the axes, though, so it
really is a coin flipping with cheat detection problem as described by
Section~\ref{sec:cheat}. Even better, the cheat detection is proportional
to $\Gamma$ which can be made arbitrarily large at no cost to us (us being
the designers of the protocols, of course there is a practical cost
involved in implementing protocols with large $\Gamma$).

As the amount of cheat detection becomes infinite, the rules of point games
become very simple: probability is conserved and average $x$ and $y$ cannot
decrease. As the problem we are carrying off to infinity has zero net
probability, and zero net average $x$ and $y$, it should be resolvable at
infinity. Sadly, even at infinity zero-bias coin flipping is impossible
(only arbitrarily small bias is allowed) so after resolving the problem at
infinity, we still need to bring back an error term down through the
ladder.

In practice, we still need to truncate the ladder at a finite distance,
resolve the coin-flipping problem at that height, and carry the error terms
back down through the ladder. There is a fairly automatic way of taking care
of all of this, but it involves more complicated ladders. The next section
will present the most important result used in building such complicated
ladders.

\subsubsection{Obtaining non-negative numerators}

Whenever we want to verify that a function $p(x)$ is valid, we need to
examine expressions of the form
\be
\sum_i \lp(\frac{-1}{\lambda+x_i}\rp) p(x_i) = 
\frac{f(-\lambda)}{\prod_i(\lambda+x_i)},
\ee
\noindent
where $x_1,\dots,x_n\geq 0$ are the finite support of $p$, and
$f(-\lambda)$ is some polynomial whose coefficients depends on the non-zero
values of $p(x)$. The reasons for making $f$ a function of $-\lambda$
rather than $\lambda$ will become clear below.

The function $p(x)$ is valid only if the above expression is non-negative
for $\lambda>0$, which in turn is true if and only if $f(-\lambda)$ is
non-negative for $\lambda>0$. The problem is that combining the terms to
find $f(-\lambda)$ is often tedious, and verifying its non-negativity can be
fairly difficult.

On the other hand, constructing a non-negative polynomial is generally easy
(for instance we can specify it as a product of its zeros). Therefore, it
is often easier to start with $f(-\lambda)$, and use it to compute an
appropriate distribution $p(x)$ over some previously selected points
$x_1,\dots,x_n$. That is the approach that we will be developing in this
section. The next two lemmas will help us prove that the desired expression
is $p(x_i) = -f(x_i)/\prod_{\substack{j\neq i}}(x_j-x_i)$, which also
satisfies probability conservation so long as $f(-\lambda)$ has degree no
greater than $n-2$.

\begin{lemma}
Let $n\geq 2$ and $x_1,\dots,x_n\in\R$ be distinct. Then
\be
\sum_{i=1}^n\prod_{\substack{j=1\cr j\neq i}}^n \frac{1}{(x_j-x_i)}=0.
\ee
\end{lemma}
\begin{proof}
We proceed by induction. For $n=2$ we trivially have
\be
\frac{1}{(x_2-x_1)}+\frac{1}{(x_1-x_2)}=0.
\ee
For $n>2$ we use the identities for $1<i<n$
\be
\frac{1}{(x_1-x_i)(x_n-x_i)}=\frac{1}{(x_n-x_1)}
\lp(\frac{1}{(x_1-x_i)}-\frac{1}{(x_n-x_i)}\rp)
\ee
to expand
\be
\sum_{i=1}^n\prod_{\substack{j=1\cr j\neq i}}^n \frac{1}{(x_j-x_i)}=
\frac{1}{(x_n-x_1)}\lp(
\sum_{i=1}^{n-1}\prod_{\substack{j=1\cr j\neq i}}^{n-1} \frac{1}{(x_j-x_i)}-
\sum_{i=2}^n\prod_{\substack{j=2\cr j\neq i}}^n \frac{1}{(x_j-x_i)}\rp)
\ee
and by induction both terms inside the parenthesis are zero.
\end{proof}

\begin{lemma}
Let $n\geq 2$ and $x_1,\dots,x_n\in\R$ be distinct. Let $f(x)$ be a
polynomial of degree $k\leq n-2$. Then
\be
\sum_{i=1}^n \frac{f(x_i)}{\prod_{\substack{j\neq i}}(x_j-x_i)}=0.
\ee
\end{lemma}
\begin{proof}
We proceed by induction on $k$. For $k=0$ the result follows from the
previous lemma. If $k>0$ we can write $f(x)=c\prod_{j=1}^{k}(x_j-x)+g(x)$
for some scalar $c\in\R$ and a polynomial $g(x)$ of degree less than $k$.
Then
\be
\sum_{i=1}^n \frac{f(x_i)}{\prod_{\substack{j\neq i}}(x_j-x_i)}=c
\sum_{i=k+1}^n\prod_{\substack{j=k+1\cr j\neq i}}^n \frac{1}{(x_j-x_i)}+
\sum_{i=1}^n \frac{g(x_i)}{\prod_{\substack{j\neq i}}(x_j-x_i)}
\ee
\noindent
and both terms are zero by induction.
\end{proof}

\begin{lemma}
Let $x_1,\dots,x_n$ be distinct non-negative numbers and let $f(-\lambda)$
be a polynomial in $\lambda$ of degree $k\leq n-2$ which is non-negative
for $\lambda>0$. Then
\be
p = \sum_i \lp(\frac{-f(x_i)}{\prod_{\substack{j\neq i}}(x_j-x_i)} \rp)
[x_i]
\ee
is a valid function.
\label{lemma:posnumer}
\end{lemma}
\begin{proof}
Using the previous lemma with an appended point $x_{n+1}=-\lambda$, we get
\be
\sum_{i=1}^n\frac{-1}{\lambda+x_i}
\lp( \frac{f(x_i)}{\prod_{\substack{j\neq i}}(x_j-x_i)} \rp) +
\frac{f(-\lambda)}{\prod_{i}(\lambda+x_i)}=0
\ee
which proves the constraints for $\lambda>0$. In fact, the above relation
holds so long as $f$ has degree $k\leq (n+1)-2$. However, we must reduce
the allowed degree by one more to get probability conservation
\be
\sum_i p(x_i) = \lim_{\lambda\rightarrow\infty} 
\sum_{i=1}^n\frac{\lambda}{\lambda+x_i}
\lp( \frac{-f(x_i)}{\prod_{\substack{j\neq i}}(x_j-x_i)} \rp)
= \lim_{\lambda\rightarrow\infty} 
\frac{-\lambda f(-\lambda)}{\prod_{i}(\lambda+x_i)}
\ee
\noindent
which converges to zero if the degree is $k\leq n-2$.
\end{proof}

\subsubsection{\label{sec:trunc}Truncating the ladder}

A single rung of our ladder has the form
\be
a\bigg[\frac{k-2}{3},\frac{k}{3}\bigg]
+b\bigg[\frac{k-1}{3},\frac{k}{3}\bigg]
+c\bigg[\frac{k+1}{3},\frac{k}{3}\bigg]
+d\bigg[\frac{k+2}{3},\frac{k}{3}\bigg]
\ee
for some constants $a,b,c,d\in\R$. Following the discussion in the last
section we want to set
\be
a=\frac{-f(\frac{k-2}{3})}{\lp(\frac{1}{3}\rp)
\lp(\frac{3}{3}\rp)\lp(\frac{4}{3}\rp)}=
-\frac{9f(\frac{k-2}{3})}{4}
\ee
\noindent
and similarly
\be
b = +\frac{9f(\frac{k-1}{3})}{2},\qquad
c = -\frac{9f(\frac{k+1}{3})}{2},\qquad
d = +\frac{9f(\frac{k+2}{3})}{4}.
\ee
The ladder for the bias $1/6$ protocol can be described this way with
$f(-\lambda)= 4/9$, which is clearly positive for $\lambda\geq 0$. But, of
course, we can allow $f$ to be other quadratic functions. More importantly,
we can allow different quadratic functions at different heights of the
ladder. That would lead to a function $f(x,y)$ with the constraint that for
every $y$ (on which the ladder is non-zero) $f(-\lambda,y)$ is a quadratic
polynomial in $\lambda$ that is non-negative for $\lambda>0$.

But there is a catch. We want to keep the symmetry of the problem so that we
can choose $v(x,y)=h(y,x)$ and still get $h+v$ to cancel on the ladder. In
other words, we want to ensure $h(x,y)=-h(y,x)$. This leads to the
conditions
\be
a_{x=\frac{k-2}{3},y=\frac{k}{3}} = -d_{x=\frac{k}{3},y=\frac{k-2}{3}}
&\qquad\Longrightarrow\qquad&
-\frac{9f(\frac{k-2}{3},\frac{k}{3})}{4} = 
-\frac{9f(\frac{k}{3},\frac{k-2}{3})}{4} 
\\
b_{x=\frac{k-1}{3},y=\frac{k}{3}} = -c_{x=\frac{k}{3},y=\frac{k-1}{3}}
&\qquad\Longrightarrow\qquad&
\frac{9f(\frac{k-1}{3},\frac{k}{3})}{2} = 
\frac{9f(\frac{k}{3},\frac{k-1}{3})}{2} 
\ee
\noindent
which are both satisfied if we enforce $f(x,y)=f(y,x)$.

Now we can choose our function $f$ to stop the ladder at a certain height
$y=\Gamma/3$ by setting
\be
f(x,y) = 
C\lp(\frac{\Gamma+1}{3}-x\rp)\lp(\frac{\Gamma+2}{3}-x\rp)
\lp(\frac{\Gamma+1}{3}-y\rp)\lp(\frac{\Gamma+2}{3}-y\rp)
\ee
for some large integer $\Gamma$ and positive constant $C$ to be determined
below.  The ladder part of $h$ becomes
\be
h_{lad}&=&\sum_{k=3}^\Lambda\bigg(
-\frac{9f(\frac{k-2}{3},\frac{k}{3})}{4}
\bigg[\frac{k-2}{3},\frac{k}{3}\bigg]
+\frac{9f(\frac{k-1}{3},\frac{k}{3})}{2}
\bigg[\frac{k-1}{3},\frac{k}{3}\bigg]
\\\nonumber&&\qquad\qquad\qquad\qquad\qquad\qquad
-\frac{9f(\frac{k+1}{3},\frac{k}{3})}{2}
\bigg[\frac{k+1}{3},\frac{k}{3}\bigg]
+\frac{9f(\frac{k+2}{3},\frac{k}{3})}{4}
\bigg[\frac{k+2}{3},\frac{k}{3}\bigg]
\bigg).
\ee
We have stopped the ladder sum at height $\Lambda/3$. Of course, we could
also have simply stopped the original ladder with $f=4/9$ at a particular
height, but we would have lost the antisymmetry of $h$. The fact that
$f(\frac{\Gamma+2}{3},\frac{\Gamma}{3})=
f(\frac{\Gamma+1}{3},\frac{\Gamma}{3})=
f(\frac{\Gamma+1}{3},\frac{\Gamma-1}{3})=0$ are all zero ensures that we
can stop the above pattern and still retain $h(x,y)=-h(y,x)$.

The next step is to verify that $h_{lad}$ is horizontally valid, but that
follows from the fact that $f(-\lambda,y)\geq0$ for $\lambda>0$ and $y\leq
\Gamma/3$, and that it is quadratic in $\lambda$. 

Finally, let us examine the bottom of the ladder. If $\Gamma$ is very
large, and $x$ and $y$ are small compared to $\Gamma$, then $f\simeq C
\Gamma^4/3^4$. If we further choose $C\simeq 36/\Gamma^4$ we end up
approximating the original constant $f=4/9$ ladder. 

The rest of this section will work out the details of merging this
truncated ladder with the structure that needs to lie at the bottom. The
discussion contains no critical new ideas and can be skipped on a first
reading.

Putting everything together we end up with a structure of the form
\be
h &=& h_{lad} 
\\\nonumber
&&- \frac{1}{2}\bigg[0,1\bigg] + 1\bigg[\frac{1}{3},1\bigg]
- \frac{1}{2}\bigg[\frac{2}{3},1\bigg]
\\\nonumber
&&+\frac{1}{2}\bigg[\frac{2+\delta}{3},\frac{2}{3}\bigg]
+\lp(\frac{1}{2}-h_{lad}\lp(\frac{2}{3},1\rp)\rp)
\bigg[1,\frac{2}{3}\bigg]
-h_{lad}\lp(\frac{2}{3},\frac{4}{3}\rp)\bigg[\frac{4}{3},\frac{2}{3}\bigg]
\\\nonumber
&&-\frac{1}{2}\bigg[\frac{2}{3},\frac{2+\delta}{3}\bigg]
+\frac{1}{2}\bigg[\frac{2+\delta}{3},\frac{2+\delta}{3}\bigg]
\ee
\noindent
for some small $\delta>0$ to be determined in a moment.

Note that $h_{lad}$ runs up to $y=1$, which produces a point at $x=1/3$,
$y=1$. We must exactly cancel the amplitude in this point with the term on
the second line of $h$. We therefore choose
\be
C= \frac{4}{9}\lp(\frac{3}{\Gamma}\rp) \lp(\frac{3}{\Gamma+1}\rp)
 \lp(\frac{3}{\Gamma-2}\rp)  \lp(\frac{3}{\Gamma-1}\rp)
\ee
\noindent
so that $h_{lad}(\frac{1}{3},1)=-1$. Note that this $C$ has the right
behavior as $\Gamma\rightarrow \infty$. The validity of $h$ at $y=1$ then
follows because it is a sum of two valid terms (one coming from $h_{lad}$),
just as it was in the original ladder in Eq.~(\ref{eq:ladnsplit}).

The difficult line is $y=2/3$ where we have coefficients that are constrained
by the symmetry $h(x,y)=-h(y,x)$. The coefficients are
\be
\frac{1}{2},\qquad\qquad
\frac{1}{2}-h_{lad}\lp(\frac{2}{3},1\rp)= \frac{1}{2}-2
\lp(\frac{\Gamma-1}{\Gamma+1}\rp)
\qquad\text{and}\qquad
-h_{lad}\lp(\frac{2}{3},\frac{4}{3}\rp)=\frac{\Gamma-3}{\Gamma+1}.
\ee
Conservation of probability follows trivially, and the line will be valid
if it satisfies the point splitting constraint $\sum p_i/x_i=0$
\be
\frac{3}{2(2+\delta)}+\frac{1}{2}-2
\lp(\frac{\Gamma-1}{\Gamma+1}\rp)
+\frac{3}{4}\lp(\frac{\Gamma-3}{\Gamma+1}\rp)=0
\ee
\noindent
which implicitly defines $\delta=8/(3\Gamma-1)$.

Finally, without the last line we would have a protocol of the form
\be
\frac{1}{2}[1,0]+\frac{1}{2}[0,1]\rightarrow
\frac{1}{2}\bigg[\frac{2+\delta}{3},\frac{2}{3}\bigg]+
\frac{1}{2}\bigg[\frac{2}{3},\frac{2+\delta}{3}\bigg].
\ee
The last line uses point raising to merge the two final points at
$[\frac{2+\delta}{3},\frac{2+\delta}{3}]$, giving us a protocol with
$P_A^*=P_B^*=\frac{2+\delta}{3}$ where $\delta\rightarrow 0$ as
$\Gamma\rightarrow \infty$.

\subsubsection{\label{sec:alternate}Building better ladders}

Having completed our analysis of the original bias $1/6$ ladder, our task
now is to apply the knowledge gained to the building of better ladders.

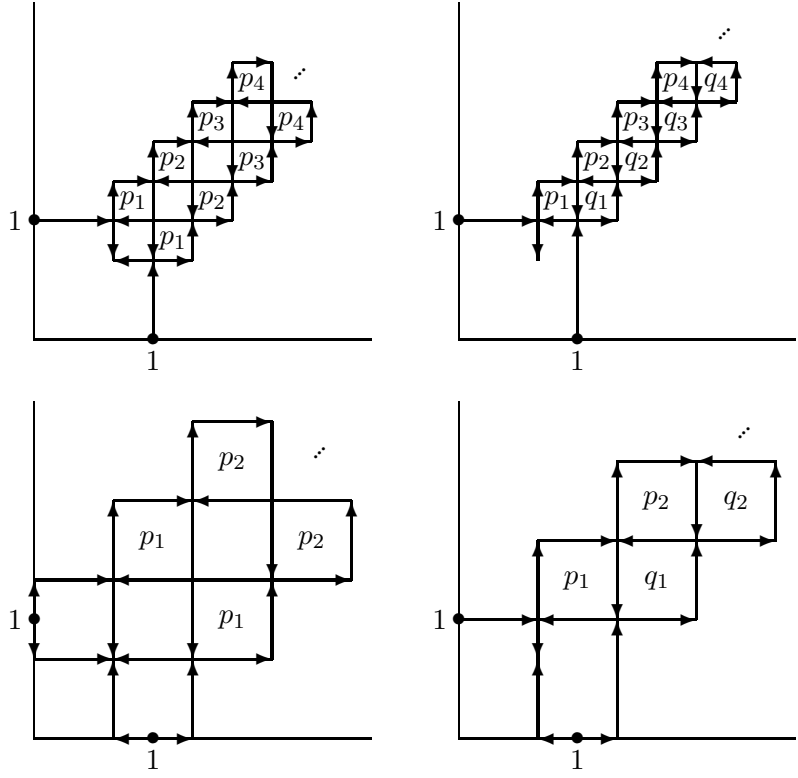
\begin{figure}[tbp]
\begin{center}
\unitlength = 15pt
\begin{picture}(9,9.5)(-0.5,-1)
\put(0,0){\line(1,0){8.5}}
\put(0,0){\line(0,1){8.5}}
\put(3,-0.3){\makebox(0,0)[t]{$1$}}
\put(-0.3,3){\makebox(0,0)[r]{$1$}}
\put(3.5,2.5){\makebox(0,0){$p_1$}}
\put(4.5,3.5){\makebox(0,0){$p_2$}}
\put(5.5,4.5){\makebox(0,0){$p_3$}}
\put(6.5,5.5){\makebox(0,0){$p_4$}}
\put(2.5,3.5){\makebox(0,0){$p_1$}}
\put(3.5,4.5){\makebox(0,0){$p_2$}}
\put(4.5,5.5){\makebox(0,0){$p_3$}}
\put(5.5,6.5){\makebox(0,0){$p_4$}}
\thicklines
\put(3,0){\vector(0,1){2}}
\put(0,3){\vector(1,0){2}}
\put(3,2){\vector(-1,0){1}}
\put(2,3){\vector(0,-1){1}}
\put(3,2){\vector(1,0){1}}
\put(2,3){\vector(0,1){1}}
\put(4,3){\vector(-1,0){2}}
\put(3,4){\vector(0,-1){2}}
\put(2,4){\vector(1,0){1}}
\put(4,2){\vector(0,1){1}}
\put(4,3){\vector(1,0){1}}
\put(3,4){\vector(0,1){1}}
\put(5,4){\vector(-1,0){2}}
\put(4,5){\vector(0,-1){2}}
\put(3,5){\vector(1,0){1}}
\put(5,3){\vector(0,1){1}}
\put(5,4){\vector(1,0){1}}
\put(4,5){\vector(0,1){1}}
\put(6,5){\vector(-1,0){2}}
\put(5,6){\vector(0,-1){2}}
\put(4,6){\vector(1,0){1}}
\put(6,4){\vector(0,1){1}}
\put(6,5){\vector(1,0){1}}
\put(5,6){\vector(0,1){1}}
\put(7,6){\vector(-1,0){2}}
\put(6,7){\vector(0,-1){2}}
\put(5,7){\vector(1,0){1}}
\put(7,5){\vector(0,1){1}}
\put(6.6,6.6){\makebox(0,0){$\cdot$}}
\put(6.7,6.7){\makebox(0,0){$\cdot$}}
\put(6.8,6.8){\makebox(0,0){$\cdot$}}
\put(3,0){\makebox(0,0){$\bullet$}}
\put(0,3){\makebox(0,0){$\bullet$}}
\end{picture}
\qquad
\begin{picture}(9,9.5)(-0.5,-1)
\put(0,0){\line(1,0){8.5}}
\put(0,0){\line(0,1){8.5}}
\put(3,-0.3){\makebox(0,0)[t]{$1$}}
\put(-0.3,3){\makebox(0,0)[r]{$1$}}
\put(3.5,3.5){\makebox(0,0){$q_1$}}
\put(4.5,4.5){\makebox(0,0){$q_2$}}
\put(5.5,5.5){\makebox(0,0){$q_3$}}
\put(6.5,6.5){\makebox(0,0){$q_4$}}
\put(2.5,3.5){\makebox(0,0){$p_1$}}
\put(3.5,4.5){\makebox(0,0){$p_2$}}
\put(4.5,5.5){\makebox(0,0){$p_3$}}
\put(5.5,6.5){\makebox(0,0){$p_4$}}
\thicklines
\put(3,0){\vector(0,1){3}}
\put(0,3){\vector(1,0){2}}
\put(2,3){\vector(0,-1){1}}
\put(2,3){\vector(0,1){1}}
\put(3,3){\vector(-1,0){1}}
\put(3,4){\vector(0,-1){1}}
\put(2,4){\vector(1,0){1}}
\put(3,3){\vector(1,0){1}}
\put(3,4){\vector(0,1){1}}
\put(4,4){\vector(-1,0){1}}
\put(4,5){\vector(0,-1){1}}
\put(3,5){\vector(1,0){1}}
\put(4,3){\vector(0,1){1}}
\put(4,4){\vector(1,0){1}}
\put(4,5){\vector(0,1){1}}
\put(5,5){\vector(-1,0){1}}
\put(5,6){\vector(0,-1){1}}
\put(4,6){\vector(1,0){1}}
\put(5,4){\vector(0,1){1}}
\put(5,5){\vector(1,0){1}}
\put(5,6){\vector(0,1){1}}
\put(6,6){\vector(-1,0){1}}
\put(6,7){\vector(0,-1){1}}
\put(5,7){\vector(1,0){1}}
\put(6,5){\vector(0,1){1}}
\put(6,6){\vector(1,0){1}}
\put(7,6){\vector(0,1){1}}
\put(7,7){\vector(-1,0){1}}
\put(6.6,7.6){\makebox(0,0){$\cdot$}}
\put(6.7,7.7){\makebox(0,0){$\cdot$}}
\put(6.8,7.8){\makebox(0,0){$\cdot$}}
\put(3,0){\makebox(0,0){$\bullet$}}
\put(0,3){\makebox(0,0){$\bullet$}}
\end{picture}
\\
\begin{picture}(9,10)(-0.5,-1)
\put(0,0){\line(1,0){8.5}}
\put(0,0){\line(0,1){8.5}}
\put(3,-0.3){\makebox(0,0)[t]{$1$}}
\put(-0.3,3){\makebox(0,0)[r]{$1$}}
\put(5,3){\makebox(0,0){$p_1$}}
\put(3,5){\makebox(0,0){$p_1$}}
\put(7,5){\makebox(0,0){$p_2$}}
\put(5,7){\makebox(0,0){$p_2$}}
\thicklines
\put(3,0){\vector(1,0){1}}
\put(3,0){\vector(-1,0){1}}
\put(0,3){\vector(0,1){1}}
\put(0,3){\vector(0,-1){1}}
\put(0,2){\vector(1,0){2}}
\put(2,0){\vector(0,1){2}}
\put(0,4){\vector(1,0){2}}
\put(4,0){\vector(0,1){2}}
\put(2,4){\vector(0,-1){2}}
\put(4,2){\vector(-1,0){2}}
\put(2,4){\vector(0,1){2}}
\put(4,2){\vector(1,0){2}}
\put(2,6){\vector(1,0){2}}
\put(6,2){\vector(0,1){2}}
\put(6,4){\vector(-1,0){4}}
\put(4,6){\vector(0,-1){4}}
\put(4,6){\vector(0,1){2}}
\put(6,4){\vector(1,0){2}}
\put(4,8){\vector(1,0){2}}
\put(8,4){\vector(0,1){2}}
\put(8,6){\vector(-1,0){4}}
\put(6,8){\vector(0,-1){4}}
\put(7.1,7.1){\makebox(0,0){$\cdot$}}
\put(7.2,7.2){\makebox(0,0){$\cdot$}}
\put(7.3,7.3){\makebox(0,0){$\cdot$}}
\put(3,0){\makebox(0,0){$\bullet$}}
\put(0,3){\makebox(0,0){$\bullet$}}
\end{picture}
\qquad
\begin{picture}(9,10)(-0.5,-1)
\put(0,0){\line(1,0){8.5}}
\put(0,0){\line(0,1){8.5}}
\put(3,-0.3){\makebox(0,0)[t]{$1$}}
\put(-0.3,3){\makebox(0,0)[r]{$1$}}
\put(3,4){\makebox(0,0){$p_1$}}
\put(5,6){\makebox(0,0){$p_2$}}
\put(5,4){\makebox(0,0){$q_1$}}
\put(7,6){\makebox(0,0){$q_2$}}
\thicklines
\put(3,0){\vector(1,0){1}}
\put(3,0){\vector(-1,0){1}}
\put(2,0){\vector(0,1){2}}
\put(2,3){\vector(0,-1){1}}
\put(0,3){\vector(1,0){2}}
\put(4,0){\vector(0,1){3}}
\put(4,3){\vector(-1,0){2}}
\put(4,3){\vector(1,0){2}}
\put(2,5){\vector(1,0){2}}
\put(6,5){\vector(-1,0){2}}
\put(6,5){\vector(1,0){2}}
\put(2,3){\vector(0,1){2}}
\put(4,5){\vector(0,-1){2}}
\put(6,3){\vector(0,1){2}}
\put(4,5){\vector(0,1){2}}
\put(6,7){\vector(0,-1){2}}
\put(8,5){\vector(0,1){2}}
\put(4,7){\vector(1,0){2}}
\put(8,7){\vector(-1,0){2}}
\put(7.1,7.6){\makebox(0,0){$\cdot$}}
\put(7.2,7.7){\makebox(0,0){$\cdot$}}
\put(7.3,7.8){\makebox(0,0){$\cdot$}}
\put(3,0){\makebox(0,0){$\bullet$}}
\put(0,3){\makebox(0,0){$\bullet$}}
\end{picture}
\caption{A few simple ladder protocols: Symmetric (left) and asymmetric 
(right), with initial point split (bottom) and without (top).}
\label{fig:simplads}
\end{center}
\end{figure}

We begin by studying a few simple variants of the $1/6$ ladder which are
depicted in Fig.~\ref{fig:simplads}. Whereas the ladders discussed thus far
have been symmetric (by reflection across the diagonal), all the TDPGs
discussed in Section~\ref{sec:examples} were asymmetric. Because the space
of TIPGs is a cone, an asymmetric TIPG can always be made symmetric by
taking a combination of itself and its reflection. The advantage of working
with symmetric TIPGs is that the validity of $h$ implies the validity of
$v$ so there are less constraints to check. The disadvantage is that the
expressions are generally more complicated (as we shall see below).
Nevertheless, in this paper we will use symmetric ladders to express the
main result and only study asymmetric ladders for comparison.

There is also the possibility of starting the ladders with a point split on
the axes. Numerical optimizations of the ladders depicted in
Fig.~\ref{fig:simplads} using a variable ladder spacing indicate that the
optimal TIPGs with no initial point split can achieve $P_A^*=P_B^*\approx
0.64$ whereas those with an initial point split can achieve
$P_A^*=P_B^*\approx 0.57$. From this perspective, constructing TIPGs with
an initial point split may be better. On the other hand, TIPGs with no
initial split tend to have simpler analytical expressions.

Unfortunately, one can also analytically prove that none of the forms
depicted in Fig.~\ref{fig:simplads} can achieve arbitrarily small bias. We
will not cover the proof here and instead directly proceed to studying more
complicated ladders that have more than four points across a horizontal
section.

A horizontal rung of an asymmetric ladder with $2k$ points across and
constant lattice spacing $\epsilon$ has the form
\be
\sum_{i=1}^{2k} \frac{-f(x+i\epsilon)}{\prod_{j\neq i} 
(j \epsilon - i \epsilon)}
\bigg[x+i\epsilon\bigg]=
\sum_{i=1}^{2k} \frac{(-1)^i f(x+i\epsilon)}{\epsilon^{2k-1}(i-1)!(2k-i)!}
\bigg[x+i\epsilon\bigg].
\ee
A symmetric ladder is similar, except that the center point is always
missing. Therefore a rung with $2k$ points can be written as
\be
\sum_{\substack{i=-k\cr i\neq 0}}^{k} \frac{-f(x+i\epsilon)}{\prod_{j\neq
i,j\neq 0} (j \epsilon - i \epsilon)} \bigg[x+i\epsilon\bigg]=
\sum_{\substack{i=-k\cr i\neq 0}}^{k} \frac{(-1)^{k+i} (i) f(x+i\epsilon)}
{\epsilon^{2k-1}(k+i)!(k-i)!}
\bigg[x+i\epsilon\bigg].
\ee
A complete symmetric ladder has the form
\be
h_{lad}= \sum_{j=j_0}^{\Gamma}
\sum_{\substack{i=-k\cr i\neq 0}}^{k}
\frac{(-1)^{k+i}(i) f((j+i)\epsilon,j\epsilon)}
{\epsilon^{2k-1}(k+i)!(k-i)!}
\bigg[(j+i)\epsilon,j\epsilon\bigg].
\ee
\noindent
The ladder has been truncated at $y/\epsilon=\Gamma$ which can be done if
we pick
\be
f(x,y) = g(x,y) \lp(\prod_{i=1}^{k} \lp(\Gamma\epsilon+i\epsilon-x\rp)\rp)
\lp(\prod_{i=1}^{k} \lp(\Gamma\epsilon+i\epsilon-y\rp)\rp)
\ee
\noindent
and then we are still free to choose a symmetric polynomial $g(x,y)=g(y,x)$
so long as $g(-\lambda,y)$ is non-negative for $\lambda>0$ and $y>0$, and
has degree at most $k-2$ in $\lambda$.

Because we only have $k-2$ zeros to play with in $g$ we can't fully
truncate the bottom of the ladder as we did the top. But that is not a
problem. After all, our goal is to attach the bottom of the ladder to our
coin-flipping problem. 

It is still useful to truncate as much of the bottom of the ladder as we
can in order to have fewer points to deal with at the bottom. We
can partially truncate at some height $y=j_0\epsilon$ by setting
\be
g(x,y) = C (-1)^k \lp(\prod_{i=1}^{k-2} \lp(j_0\epsilon-i\epsilon-x\rp)\rp)
\lp(\prod_{i=1}^{k-2} \lp(j_0\epsilon-i\epsilon-y\rp)\rp).
\ee
\noindent
The overall sign is chosen so that $g(-\lambda,y)$ is positive for
$\lambda>0$ and $j_0 \epsilon\leq y \leq \Gamma\epsilon$.

\begin{figure}[tbp]
\begin{center}
\unitlength = 20pt
\begin{picture}(17,17)(-3.5,-3.5)
\dashline{0.25}(-3,1)(13,1)
\dashline{0.25}(-3,0)(13,0)
\dashline{0.25}(1,-3)(1,13)
\dashline{0.25}(0,-3)(0,13)
\dashline{0.25}(9,-3)(9,13)
\dashline{0.25}(10,-3)(10,13)
\dashline{0.25}(11,-3)(11,13)
\dashline{0.25}(12,-3)(12,13)
\dashline{0.25}(-3,9)(13,9)
\dashline{0.25}(-3,10)(13,10)
\dashline{0.25}(-3,11)(13,11)
\dashline{0.25}(-3,12)(13,12)
\thicklines
\put(3,2){\vector(1,0){1}}
\put(2,3){\vector(0,1){1}}
\put(4,3){\vector(-1,0){2}}
\put(3,4){\vector(0,-1){2}}
\put(2,4){\vector(1,0){1}}
\put(4,2){\vector(0,1){1}}
\put(4,3){\vector(1,0){1}}
\put(3,4){\vector(0,1){1}}
\put(5,4){\vector(-1,0){2}}
\put(4,5){\vector(0,-1){2}}
\put(3,5){\vector(1,0){1}}
\put(5,3){\vector(0,1){1}}
\put(5,4){\vector(1,0){1}}
\put(4,5){\vector(0,1){1}}
\put(6,5){\vector(-1,0){2}}
\put(5,6){\vector(0,-1){2}}
\put(4,6){\vector(1,0){1}}
\put(6,4){\vector(0,1){1}}
\put(6,5){\vector(1,0){1}}
\put(5,6){\vector(0,1){1}}
\put(7,6){\vector(-1,0){2}}
\put(6,7){\vector(0,-1){2}}
\put(5,7){\vector(1,0){1}}
\put(7,5){\vector(0,1){1}}
\put(7,6){\vector(1,0){1}}
\put(6,7){\vector(0,1){1}}
\put(8,7){\vector(-1,0){2}}
\put(7,8){\vector(0,-1){2}}
\put(6,8){\vector(1,0){1}}
\put(8,6){\vector(0,1){1}}
\put(5,2){\vector(-1,0){1}}
\put(5,3){\vector(0,-1){1}}
\put(6,3){\vector(-1,0){1}}
\put(6,4){\vector(0,-1){1}}
\put(7,4){\vector(-1,0){1}}
\put(7,5){\vector(0,-1){1}}
\put(8,5){\vector(-1,0){1}}
\put(8,6){\vector(0,-1){1}}
\put(2,5){\vector(0,-1){1}}
\put(3,5){\vector(-1,0){1}}
\put(3,6){\vector(0,-1){1}}
\put(4,6){\vector(-1,0){1}}
\put(4,7){\vector(0,-1){1}}
\put(5,7){\vector(-1,0){1}}
\put(5,8){\vector(0,-1){1}}
\put(6,8){\vector(-1,0){1}}
\put(5,2){\vector(1,0){1}}
\put(6,2){\vector(0,1){1}}
\put(6,3){\vector(1,0){1}}
\put(7,3){\vector(0,1){1}}
\put(7,4){\vector(1,0){1}}
\put(8,4){\vector(0,1){1}}
\put(2,5){\vector(0,1){1}}
\put(2,6){\vector(1,0){1}}
\put(3,6){\vector(0,1){1}}
\put(3,7){\vector(1,0){1}}
\put(4,7){\vector(0,1){1}}
\put(4,8){\vector(1,0){1}}
\put(3,-1){\vector(0,1){3}}
\put(2,3){\vector(0,-1){4}}
\put(2,-2){\vector(0,1){1}}
\put(2,-1){\vector(1,0){1}}
\put(2,-1){\vector(-1,0){0.5}}
\put(-1,3){\vector(1,0){3}}
\put(3,2){\vector(-1,0){4}}
\put(-2,2){\vector(1,0){1}}
\put(-1,2){\vector(0,1){1}}
\put(-1,2){\vector(0,-1){0.5}}
\put(2,-2){\makebox(0,0){$\bullet$}}
\put(-2,2){\makebox(0,0){$\bullet$}}
\put(1.5,-1){\makebox(0,0){$\bullet$}}
\put(-1,1.5){\makebox(0,0){$\bullet$}}
\end{picture}
\caption{A ladder with $2k=8$ points across terminated at top and partially
terminated at bottom (with an added point split). The dashed lines indicate
the zeros of $f(x,y)$.}
\label{fig:step}
\end{center}
\end{figure}

With this truncation $h_{lad}$ has exactly three points to the left of the
of the first truncation band $(j_0-k+2)\epsilon\leq x\leq
(j_0-1)\epsilon$. The three points are located at
$[(j_0-k+1)\epsilon,(j_0+1)\epsilon]$, $[(j_0-k+1)\epsilon,j_0\epsilon]$
and $[(j_0-k)\epsilon,j_0\epsilon]$. If we add in an extra point split at
$y=j_0-k+1$, and use $v(x,y)=h(y,x)$ we end up with the situation depicted
in Fig.~\ref{fig:step}.

With some further tuning we can end up with
\be
h+v = -\frac{1}{2}[1,0]- \frac{1}{2}[0,1]
+\frac{1}{2}[1-\epsilon'',\epsilon']+ \frac{1}{2}[\epsilon',1-\epsilon'']
\ee
for some $0<\epsilon''<\epsilon'$, which is effectively a step along the
diagonal connecting $[1,0]$ and $[0,1]$ towards the center
$[\frac{1}{2},\frac{1}{2}]$. The slope $\epsilon'/\epsilon''$ of the step
should converge to one as $k\rightarrow\infty$, and therefore one can use a
sequence of such steps to merge the two points in the center and obtain a
coin-flipping protocol with arbitrarily small bias. We shall not explore
this construction further, though, and instead will focus our efforts on a
procedure that works is a single big step.

\subsubsection{Mixing ladders with points on the axes}

In this section we will present a family of protocols that converges to
zero bias. The corresponding TIPGs will mix ladders with probability
located on the axes. These TIPGs will be generalizations of the protocol
from Section~\ref{sec:TDPG16}, which is the simplest example of the
constructions used in this section. The discussion in this section will be
informal and the formal proofs will be deferred to the next section.

The complete protocol can be thought of as a three step process
\be
\label{eq:threestep}
\frac{1}{2}[1,0]+\frac{1}{2}[0,1]&\rightarrow& 
\frac{1}{2}\lp(\sum_{j=z^*/\epsilon}^\Gamma p(j\epsilon) [j\epsilon,0]\rp)
+ \frac{1}{2}\lp(\sum_{j=z^*/\epsilon}^\Gamma p(j\epsilon) [0,j\epsilon]\rp)
\\\nonumber
&\rightarrow& \frac{1}{2}[z^*,z^*-k\epsilon]+\frac{1}{2}[z^*-k\epsilon,z^*]
\rightarrow 1[z^*,z^*]
\ee
\noindent
and depends on the usual parameters: integer $k\geq 1$, small $\epsilon>0$,
large integer $\Gamma>0$ and a final point $1/2<z^*<1$. It also involves a
function $p(z)$ to be chosen below. There are a few obvious constraints
such as $k\epsilon<z^*$ and $z^*/\epsilon\in\Z$ which will all be resolved
in the limits $\epsilon\rightarrow 0$ and $\Gamma\rightarrow\infty$ for
fixed $k$.  Our goal will be to prove that the above process is valid for
some $z^*$, and then to find the minimum valid $P_A^*=P_B^*=z^*$ for a
given $k$.

The first transition of the above process is intended to be a series of
point splits along the axes, the second transition is the difficult step
involving ladders, and the third transition is trivially valid by
point-raising. The first transition is also valid given the following 
simple constraints
\be
1= \sum_{j=z^*/\epsilon}^\Gamma p(j\epsilon)\qquad\text{and}\qquad
1\geq \sum_{j=z^*/\epsilon}^\Gamma p(j\epsilon)/z
\label{eq:psdisc}
\ee
which correspond to probability conservation and the point-splitting
constraint (i.e., non-increasing average $1/z$). Note that, as opposed to
Section~\ref{sec:TDPG16}, we are using here a discrete function $p(z)$ with
finite support.

The second transition is the interesting step and will consist of a process
involving a ladder that slowly collects the amplitude on the axes and
deposits it near $x=z^*$ and $y=z^*$. The complete transition will be
described as usual by a valid horizontal function $h(x,y)$ and a valid vertical
function $v(x,y)=h(y,x)$ such that
\beq
h+v = - \frac{1}{2}\lp(\sum_{j=z^*/\epsilon}^\Gamma p(j\epsilon)
[j\epsilon,0]\rp) - \frac{1}{2}\lp(\sum_{j=z^*/\epsilon}^\Gamma
p(j\epsilon) [0,j\epsilon]\rp) +
\frac{1}{2}[z^*,z^*-k\epsilon]+\frac{1}{2}[z^*-k\epsilon,z^*].
\eeq

The new element is that when constructing the coefficients
$-f(x_i,y_i)/\prod_{j\neq i}(x_j-x_i)$ for the ladder, one of the
coordinates that appears in the product in the denominator is $x_j=0$. Why
is this different? So far we have exploited the fact that the product of
distances for a point was the same (up to sign) whether we computed its
vertical or horizontal neighbors. But now the expression includes an extra
factor of $(0-x_i)$ or $(0-y_i)$, which means the two computations will be
different. In other words, using a symmetric $f(x,y)=f(y,x)$ will not yield
an antisymmetric $h(x,y)=-h(y,x)$. The solution is to pull out an extra
factor of $1/y$ out of $f(x,y)$, which will once again allow us to use
symmetric functions $f(x,y)$, and will not affect the computation of
horizontal validity (since $1/y$ is positive and order zero as a polynomial
in $x$). We therefore set
\be
h= \sum_{j=z^*/\epsilon}^{\Gamma}\lp(
- \frac{p(j\epsilon)}{2} [0,j\epsilon]
+\sum_{\substack{i=-k\cr i\neq 0}}^{k}
\frac{-f((j+i)\epsilon,j\epsilon)}
{(j\epsilon)\prod_{x_\ell\neq x_i} (x_\ell-x_i)}
\bigg[(j+i)\epsilon,j\epsilon\bigg]\rp),
\ee
\noindent
where $\prod_{x_\ell\neq x_i} (x_\ell-x_i)$ still includes a factor of
$(0-x_i)=-(j+i)\epsilon$.  

In order to use Lemma~\ref{lemma:posnumer}, though, we must 
write the coefficient of $[0,j\epsilon]$ is the standard form, which
imposes relation between $p(z)$ and $f(x,y)$
\be
\frac{p(j\epsilon)}{2}= \frac{f(0,j\epsilon)}
{(j\epsilon)\prod_{x_\ell\neq 0} (x_\ell-0)}=
\frac{f(0,j\epsilon)}
{\epsilon^{2k+1} \prod_{\ell=j-k}^{j+k}\ell}.
\ee
We now pick $f$ as usual to fully truncate the ladder at the top and to
truncate as much as possible of the ladder at the bottom
\be
f(x,y) &=& C (-1)^{k-1}\lp(\prod_{i=1}^{k-1} \lp(z^*-i\epsilon-x\rp)\rp)
\lp(\prod_{i=1}^{k-1} \lp(z^*-i\epsilon-y\rp)\rp)
\\\nonumber && \qquad\qquad \times
\lp(\prod_{i=1}^{k} \lp(\Gamma\epsilon+i\epsilon-x\rp)\rp)
\lp(\prod_{i=1}^{k} \lp(\Gamma\epsilon+i\epsilon-y\rp)\rp).
\ee
We have the usual symmetry properties $f(x,y)=f(y,x)$. Furthermore,
$f(-\lambda,y)$ is positive for $\lambda>0$ and $z^*\leq y \leq
\Gamma\epsilon$ (provided $\epsilon$ is small enough so that
$k\epsilon<z^*$). Finally, $f(-\lambda,y)$ has degree $2k-1$ in $\lambda$, 
which is allowed because we have $2k+1$ points across, once we include the
point on the axis.

Because the truncation at the bottom of the ladder uses $k-1$ zeros, $h$
will only have a single point (excluding those on the axis) to the left of
the first truncation band $z^*-(k-1)\epsilon\leq x\leq z^*-\epsilon$. The
point will be located at $[z^*-k\epsilon,z^*]$ and by probability
conservation must have exactly the same amount of probability as was
originally located on the axes.

All that remains undone is to figure out what values of $z^*$ are
allowed. For the moment we will only compute this in the limit of
$\epsilon\rightarrow0$ and $\Gamma\rightarrow\infty$ in which 
\be
f(0,z)= C' \lp(\prod_{i=1}^{k-1} \lp(z^*-i\epsilon-z\rp)\rp)
\lp(\prod_{i=1}^{k} \lp(\Gamma\epsilon+i\epsilon-z\rp)\rp)
\rightarrow C'' (z-z^*)^{k-1}
\ee
\noindent
for some $k$-dependent constants $C'$ and $C''$ where we used
$(\Gamma-z)/\Gamma\rightarrow 1$. We then have
\be
p(z) = \frac{2f(0,j\epsilon)}
{z\prod_{l\neq 0} (x_l-0)}
\rightarrow C''' \frac{(z-z^*)^{k-1}}{z^{2k+1}}
\ee 
\noindent
and the two constants $z^*$ and $C'''$ are fixed by the constraints of 
Eq.~(\ref{eq:psdisc}), now in integral form
\be
1=\int_{z^*}^\infty p(z)dz, \qquad\qquad
1=\int_{z^*}^\infty \frac{p(z)}{z}dz,
\label{eq:betabeg}
\ee
\noindent
where we we have imposed equality in the second constraint to obtain the
smallest $z^*$ for a given $k$. Using $w=z^*/z$, we can transform the
following integral into a representation of the Beta function to get
\be
\int_{z^*}^\infty \frac{(z-z^*)^j}{z^\ell} dz
&=& \lp(z^*\rp)^{j-\ell+1} \int_0^1 w^{\ell-j-2}(1-w)^j d w
= \lp(z^*\rp)^{j-\ell+1} B(\ell-j-1,j+1) 
\nonumber\\
&=& \lp(z^*\rp)^{j-\ell+1} \frac{\Gamma(\ell-j-1)\Gamma(j+1))}{\Gamma(\ell)}
= \lp(z^*\rp)^{j-\ell+1} \frac{(\ell-j-2)!(j)!}{(l-1)!}.
\ee
\noindent
We then set our two constraint integrals equal to each other to cancel $C'''$
and solve for $z^*$
\be
\lp(z^*\rp)^{-(k+1)} \frac{(k)!(k-1)!}{(2k)!}=
\lp(z^*\rp)^{-(k+2)} \frac{(k+1)!(k-1)!}{(2k+1)!}
\qquad\Longrightarrow\qquad z^* = \frac{k+1}{2k+1}.
\label{eq:betaend}
\ee
In other words, we have constructed a coin-flipping protocol with
$P_A^*=P_B^*=(k+1)/(2k+1)$. The construction is valid for all $k\geq 1$.
The next section will formalize the protocol. All the main ideas will be
the same, though some of the functions will be redefined by constant
factors.

\subsection{\label{sec:formal}Formal proof}

\begin{definition}
Fix an integer $k>0$, a small $\epsilon>0$ satisfying $k\epsilon<1/2$, a
large integer $\Gamma>4k$ and a parameter $z^*\in(\frac{1}{2},1)$ such that
$z^*/\epsilon$ is an integer. Let $\Upsilon=\{k,\epsilon,\Gamma,z^*\}$ and
define
\be
g_{\Upsilon}(z) &=& 
\lp(\prod_{j=1}^{k-1} \lp(\frac{z^*-j\epsilon-z}{z^*-j\epsilon}\rp)\rp)
\lp(\prod_{j=1}^{k} \lp(\frac{(\Gamma+j)\epsilon-z}{(\Gamma+j)\epsilon}\rp)
\rp),
\\
p_{\Upsilon}(z) &=& (-1)^{k-1}g_{\Upsilon}(z)
\prod_{j=-k}^{k}\lp(\frac{1}{z+j\epsilon}\rp),
\\
C_{\Upsilon} &=& 1\bigg/\sum_{j=z^*/\epsilon}^\Gamma 
p_{\Upsilon}(j\epsilon),
\\
D_{\Upsilon}(i) &=& \epsilon^{2k-1}
\prod_{\substack{\ell=-k\cr\ell\neq i}}^{k} \lp(\ell-i\rp),
\\
2h_{\Upsilon} &=& -1[1,0] + C_{\Upsilon}\sum_{j=z^*/\epsilon}^\Gamma 
p_{\Upsilon}(j\epsilon)[j\epsilon,0]
\label{eq:h}
\\\nonumber
&&-1[z^*-k\epsilon,z^*]+1[z^*,z^*]
\\\nonumber
&&+C_{\Upsilon} \sum_{j=z^*/\epsilon}^\Gamma\lp(
-p_{\Upsilon}(j\epsilon)[0,j\epsilon]+
\sum_{\substack{i=-k\cr i\neq 0}}^{k}
\frac{(-1)^k g_{\Upsilon}((j+i)\epsilon)g_{\Upsilon}(j\epsilon)}
{(j\epsilon)((j+i)\epsilon)D_{\Upsilon}(i)}
[(j+i)\epsilon,j\epsilon]
\rp)
\ee
and $v_{\Upsilon}(x,y)=h_{\Upsilon}(y,x)$.
\label{def:main}
\end{definition}

\begin{lemma}
Given $\Upsilon$ as in Definition~\ref{def:main}, if
\be
1\geq C_\Upsilon
\sum_{j=z^*/\epsilon}^\Gamma 
\frac{p_{\Upsilon}(j\epsilon)}{j\epsilon}
\ee
then $h_{\Upsilon}$ is a valid horizontal function and $v_{\Upsilon}$
is a valid vertical function.
\end{lemma}

\begin{proof}
By symmetry, it suffices to prove that $h_{\Upsilon}$ is a valid horizontal
function. It is also sufficient to check the conditions independently on
each of the three lines of Eq.~(\ref{eq:h}). Given the constraint in the
definition of the lemma, the first line is a valid point split and the
second line is a valid point raise (see Lemma~\ref{lemma:rms}).

Also $p_\upsilon(z)\geq 0$ for $z^*/\epsilon\leq z\leq\Gamma$ so
$C_\Upsilon\geq 0$ and it can be canceled as we focus on the term in
parenthesis in the third line for each integer $j\in[z^*/\epsilon,\Gamma]$.
If we define
\be
f_j(x) = \frac{(-1)^{k-1} 
g_{\Upsilon}(x)g_{\Upsilon}(j\epsilon)}
{(j\epsilon)}
\ee
then we can write the line at $y=j\epsilon$ as
\be
\sum_i \lp(\frac{-f_j(x_i)}{\prod_{\substack{j\neq i}}(x_j-x_i)} \rp)
[x_i,y],
\ee
\noindent
where $x_i\in\{0\}\cup
\{(j-k)\epsilon,\dots,(j-1)\epsilon,(j+1)\epsilon,\dots,(j+k)\epsilon\}$
and we used $g_{\Upsilon}(0)=1$ among other relations.

Now we can apply Lemma~\ref{lemma:posnumer} which tells us that the above
function is valid so long as $f_j(-\lambda)$ is a polynomial in $\lambda$
of order no greater than $2k-1$ and positive for all $\lambda>0$. The first
condition is trivial and the second follows because
$g_{\Upsilon}(-\lambda)>0$ for $\lambda>0$ and
$(-1)^{k-1}g_{\Upsilon}(j\epsilon)\geq 0$ for $z^*\leq
j\epsilon\leq\Gamma\epsilon$.
\end{proof}

\begin{lemma}
Given $\Upsilon$ as in Definition~\ref{def:main} we have
\be
h_{\Upsilon}+v_{\Upsilon} = 
-\frac{1}{2}[1,0]-\frac{1}{2}[0,1]+1[z^*,z^*].
\ee
\end{lemma}

\begin{proof}
The only nontrivial cancellation is on the points $[z^*-k\epsilon,z^*]$ and
$[z^*,z^*-k\epsilon]$ which have by symmetry the same coefficient (with the
same sign). Because each line of $h_{\Upsilon}$ conserves probability,
$h_{\Upsilon}+v_{\Upsilon}$ must have a net zero probability, and so the
coefficients of these points must also be zero.
\end{proof}

\begin{corollary}
Given $\Upsilon$ as in Definition~\ref{def:main}, if
\be
1\geq C_\Upsilon
\sum_{j=z^*/\epsilon}^\Gamma 
\frac{p_{\Upsilon}(j\epsilon)}{j\epsilon}
\ee
then $h_{\Upsilon}$ and $v_{\Upsilon}$ is a TIPG with final point
$[z^*,z^*].$
\end{corollary}

\begin{lemma}
Given $\Upsilon$ as in Definition~\ref{def:main}, there exists a family of
solutions to 
\be
1\geq C_\Upsilon
\sum_{j=z^*/\epsilon}^\Gamma 
\frac{p_{\Upsilon}(j\epsilon)}{j\epsilon}
\ee
\noindent
such that $\epsilon\rightarrow0$, $\Gamma\rightarrow\infty$ and
\be
z^*\rightarrow \frac{k+1}{2k+1}.
\ee
\end{lemma}

\begin{proof}
For a given $\epsilon$ and $\Gamma$, the best $z^*$ is constrained by
$z^*/\epsilon\in\Z$ and 
\be
\sum_{j=z^*/\epsilon}^\Gamma 
p_{\Upsilon}(j\epsilon)
\geq
\sum_{j=z^*/\epsilon}^\Gamma 
\frac{p_{\Upsilon}(j\epsilon)}{j\epsilon},
\ee
\noindent
where we have expanded the definition of $C_\Upsilon$. We then use
\be
\lp(\frac{z-(z^*-k\epsilon)}{z^*-k\epsilon}\rp)^{k-1}
\lp(\frac{1}{z-k\epsilon}\rp)^{2k+1} \geq &p_{\Upsilon}(z)&
\geq 
\lp(\frac{z-z^*}{z^*}\rp)^{k-1}\lp(\frac{\Gamma-z}{\Gamma}\rp)^k
\lp(\frac{1}{z+k\epsilon}\rp)^{2k+1}
\ee
\noindent
to note that if we choose $z^*$ in accordance with the strict inequality
\be
\sum_{j=z^*/\epsilon}^\infty
\lp(\frac{j\epsilon-z^*}{z^*}\rp)^{k-1}
\lp(\frac{1}{j\epsilon+k\epsilon}\rp)^{2k+1}
>
\sum_{j=z^*/\epsilon}^\infty
\lp(\frac{j\epsilon-(z^*-k\epsilon)}{z^*-k\epsilon}\rp)^{k-1}
\lp(\frac{1}{j\epsilon-k\epsilon}\rp)^{2k+2} 
\ee
then we can always find a large enough integer $\Lambda$ so that the
original inequality is satisfied (that is because the new right-hand side is
greater than or equal to the original right-hand side, whereas the original
left-hand side will converge as $\Lambda\rightarrow\infty$ to an
expression greater than or equal to the new left-hand side).

Similarly, if we choose $z^*$ in accordance with the strict inequality
\be
\int_{z^*}^\infty
\lp(\frac{z-z^*}{z^*}\rp)^{k-1}
\lp(\frac{1}{z}\rp)^{2k+1} dz
&>&
\int_{z^*}^\infty
\lp(\frac{z-z^*}{z^*}\rp)^{k-1}
\lp(\frac{1}{z}\rp)^{2k+2} dz
\ee
we can find an appropriate $\epsilon>0$ so that the original constraints
are satisfied. The argument for solving the above inequality is the same as
the one used for Eq.~(\ref{eq:betabeg}) through Eq.~(\ref{eq:betaend}). The
constraint becomes
\be
z^*> \frac{k+1}{2k+1}
\ee
and therefore for any such $z^*$ we can find appropriate $\epsilon$ and
$\Gamma$ that satisfy the original constraints.
\end{proof}

\begin{corollary}
For every integer $k>0$ there is a family of coin-flipping protocols 
that converges to
\be
P_A^*=P_B^*=\frac{k+1}{2k+1}.
\ee
\end{corollary}

\begin{corollary}
There exists protocols for quantum weak coin flipping with
arbitrarily small bias.
\end{corollary}

\section{\label{sec:end}Conclusions}

We have constructively proven the existence of protocols for quantum weak
coin flipping with arbitrarily small bias. In the end, it appears that
quantum information has fulfilled at least a small part of its promise in
the area of two-party secure computation.

We have also tried to provide a primer on how to use Kitaev's formalism to
build interesting protocols. Hopefully the present result will be the
first of many to be obtained by viewing quantum games as dual to the cone
of bi-operator monotone functions.

\subsection{Open problems}

\begin{enumerate}

\item Improvements and extensions to the proof:

\begin{enumerate}
\item We did not explicitly complete the proof that if no coin-flipping
protocols with arbitrarily small bias exists, then the bound can be proven
using a bi-operator monotone function. Completing the proof would show that
the cone of TIPGs and the cone of bi-operator monotone functions are dual.

\item We have only constructed protocols for the case
$P_A=P_B=1/2$. Protocols for other cases can be constructed using serial
composition of this protocol. Nevertheless, it should also be
straightforward to explicitly describe a TIPG for such cases.

\item The alternative construction from Section~\ref{sec:alternate} needs
to be fleshed out. It may lead to some elegant TIPGs.

\item Appendix~\ref{sec:f2m} needs to be simplified/made more elegant.

\end{enumerate}

\item Improvements and extensions to the protocol:

\begin{enumerate}

\item Can we find a simple unitary description (such as the one in
Appendix~\ref{sec:ddb}) of a family of protocols that achieves arbitrarily
small bias?

\item Can we optimize the resources (messages, storage qubits, complexity
of unitaries?) needed to implement such a protocol. 

\begin{enumerate}

\item What are the asymptotic costs of achieving arbitrarily small bias?

\item What are the practical costs of achieving small bias? How hard would 
it be to make a protocol with bias $0.001$?

\end{enumerate}

\item What can be said about multiparty weak coin flipping?

\end{enumerate}

\item Beyond weak coin flipping 

\begin{enumerate}

\item Weak coin flipping with arbitrarily small bias leads trivially to a
new protocol for strong coin flipping with bias (arbitrarily close to)
$1/4$. Is this the best that can be done?

\item More generally, what is the optimal protocol for strong coin flipping
with cheat detection? While the formalism in this paper can be adapted to
strong coin flipping, it probably cannot be used unmodified. The difference
is that in strong coin flipping we need to simultaneously bound four
quantities (Alice and Bob's probabilities each of obtaining zero and
one). Even classically one can construct protocols that achieve
$P_{A0}^*P_{B0}^*=1/2$ (so long as one is willing to tolerate
$P_{A1}=P_{B1}=1$).

\item Strong coin flipping with cheat detection can be used to bound bit
commitment with cheat detection. Do they share an optimal protocol? Can it
be used as a building block for all secure two-party computations with
cheat detection?

\item Kitaev's first formalism can also be used in the study of specific 
oracle problems. In this context the second formalism could be used to
study all oracles simultaneously, which may prove useful in identifying
optimal oracles in some sense (for instance for proving separations between
quantum and classical computation).

\item What else?

\end{enumerate}

\end{enumerate}


\section*{Acknowledgments}
I would like to thank Alexei Kitaev for teaching me his new formalism,
without which this result would not have been possible. I would also like
to thank Dave Feinberg and Debbie Leung for their help and useful
discussions. Research at Perimeter Institute for Theoretical Physics is
supported in part by the Government of Canada through NSERC and by the
Province of Ontario through MRI.




\appendix

\newpage
\section{\label{sec:ddb}Dip-Dip-Boom and the bias 1/6 protocol}

We present in this section a reformulation of the bias $1/6$ protocol from
\cite{me2005}. The new version of the protocol is simpler and uses
measurements throughout the protocol in order to keep the total storage
space small: one qutrit for each of the players and one qubit for sending
messages. In fact, the message qubit can be discarded and reinitialized
after each message, so that only the qutrits need to be kept coherent for
the length of the protocol.

The new protocol can be described as a quantum version of the ancient game
of Dip-Dip-Boom, whose classical version is played as follows: two players
sequentially say either ``Dip'' or ``Boom''. In the first case the game
proceeds whereas in the second case the game immediately ends and the
person who said ``Boom'' is declared the winner. There are no bonus points
for longer games and a player can begin and immediately win a game by
saying ``Boom''. Why this game is played at MIT, and how it was accepted as
a major component of the author's doctoral thesis, are questions beyond the
scope of this paper.

What we shall do is build a coin-flipping protocol out of the above game.
First, we introduce honest and cheating players. Honest players will have
to output ``Dip'' vs ``Boom'' according to some previously fixed
probability distribution, whereas cheating players are still free to say
whatever they want at each round. A game is now specified by a set of
numbers $p_1, p_2, p_3, \dots$ so that $p_i$ is the probability that the
$i$th word is ``Boom''. Note that $p_1,p_3,\dots$ apply to the first player
whom we call Alice and $p_2,p_4,\dots$ apply to the second player whom we
call Bob. For convenience we shall fix an $n\geq1$ and assume that the game
ends after $n$ messages by setting $p_n = 1$.

We now define the quantities $P_A(i)$, $P_B(i)$ and $P_U(i)$ which are
respectively the probabilities that after $i$ messages Alice has won the
game, Bob has won the game, or the game remains undecided. These quantities
can be inductively calculated by
\be
P_A(i) &=& 
\begin{cases}
P_A(i-1) & \text{for $i$ even,}\cr
P_A(i-1) + p_i P_U(i-1) & \text{for $i$ odd,}
\end{cases}
\\
P_B(i) &=& 
\begin{cases}
P_B(i-1) + p_i P_U(i-1) & \text{for $i$ even,}\cr
P_B(i-1) & \text{for $i$ odd,}
\end{cases}
\\
P_U(i) &=& (1-p_i) P_U(i-1),
\ee
with initial conditions of $P_U(0)=1$ and $P_A(0)=P_B(0)=0$.

The quantum version of the protocol is simply a coherent version of the
above, with an additional cheat detection step. Alice and Bob will each
hold a qutrit
\be
\HA = \HB = \vspan\{\ket{A}, \ket{B}, \ket{U}\}
\ee
which encodes the state of the game: ``Alice has won'', ``Bob has won'' and
``undecided'' respectively. The message space is just a qubit
\be
\HM = \vspan\{\ket{\text{DIP}},\ket{\text{BOOM}}\}
\ee
comprising the two possible messages.

Amplitude will be moved coherently using the unitary operator $\Rot$
defined by
\be
\Rot(\ket{\alpha},\ket{\beta},\epsilon)
= \mypmatrix{\ket{\alpha} & \ket{\beta}}
\mypmatrix{\sqrt{1-\epsilon} & -\sqrt{\epsilon}\cr
\sqrt{\epsilon} & \sqrt{1-\epsilon}}
\mypmatrix{\bra{\alpha} \cr \bra{\beta}}
+
\bigg(I - \ket{\alpha}\bra{\alpha} - \ket{\beta}\bra{\beta}  \bigg),
\ee
\noindent
which simply effectuates a rotation in the $\ket{\alpha}$, $\ket{\beta}$
plane. For instance, consider the classical step where Alice says ``Boom''
with probability $p$ given that her state is $\ket{U}$. In the quantum protocol
this is described by the rotation $\Rot(\ket{U}\otimes\ket{\text{DIP}},
\ket{A}\otimes\ket{\text{BOOM}}, p)$. Assuming that previously
$\ket{A}\otimes\ket{\text{BOOM}}$ had zero amplitude, the operation will
move into this state an amplitude of $\sqrt{p}$ times the prior amplitude
of $\ket{U}\otimes\ket{\text{DIP}}$. We are now ready to state the quantum
protocol:

\begin{protocol}[Weak coin flipping with bias 1/6, simplified]\ \\
Fix $n\geq 1$ and numbers $p_1,\dots,p_n\in[0,1]$. Define $\HA$, $\HB$,
$\HM$, $P_A(i)$, $P_B(i)$ and $P_U(i)$ as above. 
The protocol has the following steps

\begin{enumerate}

\item Initialization: Alice prepares $\ket{U}\otimes\ket{\text{DIP}}$ 
in $\HA\otimes\HM$. Bob prepares $\ket{U}$ in $\HB$.

\item For $i = 1$ to $n$ execute the following steps:\\
(where we use $X$ to denote the player that would choose the $i$th message
in the classical protocol, and $Y$ is the other player. That is
$X=A$ and $Y=B$ if $i$ is odd or $X=B$ and $Y=A$ if $i$ is even).
\label{item:loop}

\begin{enumerate}

\item $X$ applies the operator
\be
R_i \equiv \Rot\lp(\ket{U}\otimes\ket{\text{DIP}},\,
\ket{X}\otimes\ket{\text{BOOM}},\, p_i\rp).
\ee

\item $X$ sends the qubit $H_M$ which is received by $Y$.
\label{item:mes}

\item $Y$ applies the operator
\be
\tilde R_i = \Rot\lp(\ket{U}\otimes\ket{\text{BOOM}},\,
\ket{X}\otimes\ket{\text{DIP}},\, \frac{p_i\,P_U(i-1)}{P_X(i)}\rp).
\ee

\item $Y$ measures the message qubit $\HM$ in the computational basis.\\
If the outcome is $\ket{\text{BOOM}}$ then $Y$ aborts and outputs $Y$.
\label{item:meas}

\end{enumerate}

\item Alice and Bob each measure their qutrit in the computational basis.
If the outcome is $U$ they declare themselves the winner, otherwise they
output the measurement outcome as the winner.

\end{enumerate}
\end{protocol}

\noindent
It is not hard to see that when both players are honest the state of the
system at the end of the $i$th iteration of Step~\ref{item:loop} can be
written as
\be
\lp(\sqrt{P_A(i)} \ket{A}\otimes\ket{A} 
+ \sqrt{P_B(i)} \ket{B}\otimes\ket{B} 
+ \sqrt{P_U(i)} \ket{U}\otimes\ket{U}\rp) \otimes \ket{\text{DIP}}.
\ee
To obtain a standard coin-flipping protocol we must therefore restrict the
choices of $p_1,\dots,p_n$ to values such that $P_A(n)=P_B(n)=1/2$. For
simplicity, we will also assume that for $i<n$, the other probabilities
$p_i$ are neither zero nor one. In particular, this requires the previously
discussed condition $p_n=1$.

Step~\ref{item:meas} is the new element of the quantum protocol, and serves
as a cheat detecting step. Without this measurement the protocol is exactly
equivalent to the original classical protocol.

It is important to note that when both players are honest, neither will
ever abort in Step~\ref{item:meas}. However, when one player is cheating
then the honest player will abort at Step~\ref{item:meas} with some
non-zero probability. Roughly speaking, the measurement compares the ratio
of amplitude in ``Boom'' (tensored with $\ket{U}$) in the present message,
to the total amplitude of ``Boom'' from previous messages. The classical
strategy of declaring with probability one ``Boom'' on the first message is
thwarted because this ratio will effectively be zero for all subsequent
messages.  The optimal quantum cheating strategy involves a small amount
of cheating on each message, but that gives an honest player an opportunity
to declare victory too.

In the next section we will sketch the computation of $P_A^*$ and $P_B^*$
for the above protocol and find that they match the expressions for the
optimal protocols from \cite{me2005}. This is the first step in proving the
equivalence of the two protocols. A complete proof of equivalence would
take us too far afield, though one possible route follows from the
discussion in Section~\ref{sec:TDPG16}.

In practice, the main difference between the protocols is that in the one
presented above the cheat detection is done gradually as the protocol
progresses rather than in one big measurement at the end. This simplifies
the description of the protocol, and reduces the resources required for a
physical implementation. Alas, it does not make it any more cheat
resistant.  In the limit $n\rightarrow\infty$ and with suitably chosen
values for $p_1,\dots,p_n$ (almost any choice so long as they go smoothly
to zero as $n\rightarrow\infty$) the above protocol achieves a bias of
$1/6$.

\subsection*{Analysis}

In this section we will make use of Kitaev's first formalism to compute
$P_B^*$ ($P_A^*$ can be obtained by a similar argument). In other words we
will find a feasible point of the dual SDP. Of course, this only proves an
upper bound on $P_B^*$, but the resulting expressions are in fact optimal,
and a matching lower bound can also be constructed.

The final result for this section will be $P_B^*$ as a function of the
variables $p_1,\dots,p_n$. We will not attempt to find the optimal values
for these parameters as that task is already done in \cite{me2005}.
We will also not attempt to describe the result as a TDPG, though this is a
simple task using the dual feasible point constructed below.

Henceforth we will consider the case of honest Alice and cheating Bob. The
analysis will be done from the perspective of Alice's qubits which must
satisfy the following SDP:
\be
\rho_1 &=& \Tr_{\HM}[R_1\lp(\ket{U}\bra{U}\otimes\ket{\text{DIP}}
\bra{\text{DIP}} \rp) R_1^\dagger],\\
\Tr_{\HM}[\rho_i] &=& \rho_{i-1} 
\qquad\qquad\qquad\qquad\qquad\qquad\qquad\qquad\qquad\ 
\text{for $i$ even,}\\
\rho_i &=& \Tr_{\HM}\lp[ R_i \Pi_{\text{DIP}} \tilde R_{i-1}
\rho_{i-1} \tilde R_{i-1}^\dagger \Pi_{\text{DIP}}  R_i^\dagger \rp]
 \qquad\qquad\text{for $i>1$ odd,}\\
\rho_\text{f} &=& \begin{cases}
\rho_n & \text{$n$ odd,}\\
\Tr_{\HM}\lp[\tilde R_n\rho_n \tilde R_n\dagger \rp] & \text{$n$ even,}
\end{cases}\\
P_{win} &=& \bra{B} \rho_{f} \ket{B},
\ee
\noindent
where $\rho_i$ is the state of Alice's qubits immediately after the $i$th
message (i.e., after Step~\ref{item:mes}). The state $\rho_i$ for even $i$
is an operator on $\HA\otimes\HM$ and is unknown but constrained by
$\rho_{i-1}$. The state $\rho_i$ for odd $i$ is an operator on $\HA$ and
can be computed from $\rho_{i-1}$ by applying the operators that Alice
would use.

Note that $\Pi_{\text{DIP}}\equiv I\otimes\ket{\text{DIP}}\bra{\text{DIP}}$
is a projector corresponding to a successful measurement in
Step~\ref{item:meas}. The normalization of $\rho_i$ will therefore be the
probability that Alice has not aborted yet. The final state $\rho_f$ on
$\HA$ will have a similar normalization and $P_{win}$ will be the
probability that Alice outputs a victory for Bob. In particular, Bob's
maximum probability of winning by cheating, $P_B^*$, is given by the
maximization of $P_{win}$ over positive semidefinite matrices satisfying
the above constraints.

The dual to the above SDP can be written in the following form:
\be
P_B^* &\leq& \bra{U}Z_0\ket{U},\\
Z_{i-1} \otimes\ket{\text{DIP}}\bra{\text{DIP}} &=&
\Pi_{\text{DIP}} R_{i}^\dagger \lp(Z_{i}\otimes I_{\HM}\rp) 
R_{i} \Pi_{\text{DIP}} \qquad\qquad\ \, \text{for $i$ odd,}
\label{eq:zodd2}\\
Z_{i-1} \otimes I_{\HM} &\geq& \tilde R_{i}^\dagger
\lp( Z_{i} \otimes\ket{\text{DIP}}\bra{\text{DIP}} \rp) \tilde R_{i}
\qquad\qquad\quad\ 
\text{for $i$ even,}
\label{eq:zeven2}\\
Z_n &\geq& \ket{B}\bra{B},
\label{eq:zlast}
\ee
\noindent
where the variables $Z_0,\dots,Z_n$ are semidefinite operators on
$\HA$. Any assignment of these variables consistent with the above
constraints provides an upper bound on $P_B^*$. The infimum of
$\bra{U}Z_0\ket{U}$ will in fact be equal to $P_B^*$.

Now comes the crucial bit of guesswork, where we choose a diagonalizing
basis for the operators $Z_i$. Because any assignment to $Z_0,\dots,Z_n$
consistent with the constraints provides an upper bound on $P_B^*$,
choosing a bad basis will at worse give us a non-tight upper bound. Based
on experience from past protocols, we choose to restrict ourselves to
studying matrices that are diagonal in the computational basis, and we
write
\be
Z_i = \mypmatrix{a_i & 0 & 0 \cr 0 & b_i & 0 \cr 0 & 0 & u_i}. 
\ee
We can now express Eq.~(\ref{eq:zodd2}) as
\be
\mypmatrix{a_{i-1} & 0 & 0 \cr 0 & b_{i-1} & 0 \cr 0 & 0 & u_{i-1}}
= \mypmatrix{a_i & 0 & 0 \cr 0 & b_i & 0 \cr 0 & 0 & 
(1-p_i) u_{i} + p_i a_i}
\ee
\noindent
valid for $i$ odd. We can also write Eq.~(\ref{eq:zeven2}) as
\be
\mypmatrix{
a_{i-1} &       0 &       0 &       0 &       0 &       0 \cr
      0 & a_{i-1} &       0 &       0 &       0 &       0 \cr
      0 &       0 & b_{i-1} &       0 &       0 &       0 \cr
      0 &       0 &       0 & b_{i-1} &       0 &       0 \cr
      0 &       0 &       0 &       0 & u_{i-1} &       0 \cr
      0 &       0 &       0 &       0 &       0 & u_{i-1}}
\geq
\mypmatrix{
    a_i &       0 &       0 &       0 &       0 &       0 \cr
      0 &       0 &       0 &       0 &       0 &       0 \cr
      0 &       0 &     (1-\tilde p_i) b_i &       0 &       0 &
\sqrt{\tilde p_i (1-\tilde p_i)} b_i \cr
      0 &       0 &       0 &       0 &       0 &       0 \cr
      0 &       0 &       0 &       0 &     u_i &       0 \cr
      0 &       0 & \sqrt{\tilde p_i (1-\tilde p_i)} b_i &       
0 &       0 &       \tilde p_i b_i}
\ee
\noindent
valid for $i$ even, where we introduced $\tilde p_i = p_i P_U(i-1)/P_B(i)$
as a notation for the rotation parameter of $\tilde R_i$. The above
inequality can be simplified to $a_{i-1}\geq a_i$, $u_{i-1}\geq u_i$, plus
the positivity of the $2\times2$ matrix
\be
\mypmatrix{ b_{i-1} - (1-\tilde p_i) b_i & 
-\sqrt{\tilde p_i (1-\tilde p_i)} b_i \cr
-\sqrt{\tilde p_i (1-\tilde p_i)} b_i &
u_{i-1} - \tilde p_i b_i
}\geq 0 
\label{eq:2x2cond}
\ee
\noindent
which is satisfied if its determinant and its diagonal elements are
non-negative. 

It is not hard to see that we can choose $a_i=0$ for all $i$. We are after
all trying to minimize $\bra{U}Z_0\ket{U} = u_0$. We now get the relation
$u_{i-1}=(1-p_i) u_i$ for $i$ odd, and we know that the best that we can
hope for is $u_{i-1}=u_i$ for $i$ even. If this were true we would have
\be
u_i = u_0 \prod_{\substack{j=1\cr\text{$j$ odd}}}^i \frac{1}{1-p_i}.
\ee
\noindent
Let us be optimistic, and assume the above holds and then check whether the
remaining constraints can be satisfied. Surprisingly, we shall find that the
answer is yes.

The intuition for the $b_i$ variables is that $b_n=1$ and $b_i$ increases
as $i$ decreases. In fact, the larger $b_0$ is, the better the bound we
will find, and the infimum is attained for $b_0=\infty$. With a little
care, we can directly handle this infinity and get $b_1=b_0=\infty$ and
$b_2 = u_1/\tilde p_2$ which satisfies Eq.~(\ref{eq:2x2cond}) for $i=2$.

The remaining conditions for $b_i$ will no longer involve infinities and
are obtained by minimizing the determinant Eq.~(\ref{eq:2x2cond}), which
leads to
\be
\frac{1}{b_i} = \frac{\tilde p_i}{u_{i-1}} + \frac{1-\tilde p_i}{b_{i-1}}
\qquad\qquad\text{for $i$ even,}
\ee
\noindent
which also guarantees that the constraint on the diagonal elements of
Eq.~(\ref{eq:2x2cond}) is satisfied. By induction we can write
\be
\frac{1}{b_i} &=& 
\sum_{\substack{j=2\cr\text{$j$ even}}}^i \frac{\tilde p_j}{u_{j-1}}
\prod_{\substack{k=j+2\cr\text{$k$ even}}}^i (1-\tilde p_k)
=\sum_{\substack{j=2\cr\text{$j$ even}}}^i \frac{\tilde p_j}{u_{j-1}}
\prod_{\substack{k=j+2\cr\text{$k$ even}}}^i \frac{P_B(k-1)}{P_B(k)}
=\sum_{\substack{j=2\cr\text{$j$ even}}}^i \frac{\tilde p_j}{u_{j-1}}
\frac{P_B(j+1)}{P_B(i)}\\
&=& \sum_{\substack{j=2\cr\text{$j$ even}}}^i 
\frac{p_j P_U(j-1)}{u_{j-1} P_B(i)}
= \frac{1}{u_0 P_B(i)} \sum_{\substack{j=2\cr\text{$j$ even}}}^i 
p_j P_U(j-1) \prod_{\substack{k=1\cr\text{$k$ odd}}}^{j-1} (1-p_k),
\ee

\noindent
where in the second equality we used $P_B(k)- p_k P_U(k-1) = P_B(k-1)$ for
$k$ even, and in the third equality we used $P_B(k+1)/P_B(k)=1$ for $k$
even.

What we have done is solve for all the dual variables in terms of $u_0$. 
We have also satisfied all constraints except for those from
Eq.~(\ref{eq:zlast}), which tells to set $b_n=1$ and allows us to solve for
$u_0$ to obtain our upper bound:
\be
P_B^* \leq u_0 = 2 
\sum_{\substack{j=2\cr\text{$j$ even}}}^n p_j
\lp(\prod_{k=1}^{j-1} (1-p_k)\rp)
\lp(\prod_{\substack{k=1\cr\text{$k$ odd}}}^{j-1} (1-p_k)\rp),
\ee
\noindent
where we used $P_B(n)=1/2$.

The above expression is equivalent to the result found in
\cite{me2005}. For a simple example, take the Spekkens and Rudolph
\cite{Spekkens2002} protocol with $P_A^*=P_B^*=1/\sqrt{2}$, which can be 
described in the above formalism by setting $p_1 = 1-1/\sqrt{2}$,
$p_2=1/\sqrt{2}$ and $p_3=1$. The above bound then becomes $u_0 = 2 p_2
(1-p_1)^2 = 1/\sqrt{2}$ as expected.

\newpage

\section{\label{sec:strong}Proof of strong duality}

We say strong duality holds when the maximum of the primal SDP and the
infimum of the dual SDP are equal. Though strong duality does not hold in
general, it does hold in most cases. A number of lemmas which provide
sufficient conditions for strong duality can be found in convex
optimization books such as \cite{convex} or \cite{convex2}.
In this section, however, we aim to give a direct proof of strong duality
for the particular case of the coin-flipping SDP. This appendix follows the
notation of Section~\ref{sec:kit1}

More specifically, the goal for this section is to prove the existence of
arbitrarily good upper bound certificates. Mathematically, we aim to show
that $\inf \bra{\psi_{A,0}} Z_{A,0} \ket{\psi_{A,0}} = P_B^*$, where the
infimum is taken over all dual feasible points.

Let us begin by defining $R_i$ as the set of density matrices $\rho_{A,i}$
on $\HA$ that are attainable by a cheating Bob after $i$ messages. We will
focus only on even $i$. These sets satisfy the properties:
\begin{itemize}
\item $R_0 = \lp\{\ket{\psi_{A,0}}\bra{\psi_{A,0}}\rp\}$.
\item $R_{i+2} = \lp\{\Tr_\HM[ U_{A,i+1} \tilde\rho_{A,i} U_{A,i+1}^\dagger]
: \tilde \rho_{A,i}\geq 0 \text{ and }
\Tr_\HM \tilde \rho_{A,i}\in R_i\rp\}$.
\item $R_i$ is convex for all $i$.
\end{itemize}
\noindent
The last property follows from the previous two by induction.

Though in general we cannot find a dual feasible point such that
$\bra{\psi_{A,0}} Z_{A,0} \ket{\psi_{A,0}} = P_B^*$, we can get arbitrarily
close. Specifically, we can prove that for every $\epsilon>0$ we can pick
the dual variables one by one, starting with $Z_{n-2}$ and working our way
backwards towards $Z_0$, such that
\begin{itemize}
\item $Z_{A,i}\otimes I_\HM \geq
U_{A,i+1}^\dagger \lp(Z_{A,i+2}\otimes I_\HM\rp) U_{A,i+1}$
\item $\max_{\rho_{A,i}\in R_i} \Tr[Z_{A,i} \rho_{A,i}] \leq 
\max_{ \rho_{A,i+2}\in R_{i+2}} \Tr[Z_{A,i+2} \rho_{A,i+2}] + 2 \epsilon$
\end{itemize}
\noindent
for even $i$, where as usual $Z_{A,n} = \Pi_{A,1}$. The first condition
guarantees that the constructed solution is indeed a dual feasible point
(the operators for $i$ odd can be found from the equality $Z_{A,i} =
Z_{A,i+1}$). The second condition gives us
\be
\bra{\psi_{A,0}} Z_{A,0} \ket{\psi_{A,0}} = 
\max_{\rho_{A,0}\in R_0} \Tr[Z_{A,0} \rho_{A,0}] \leq
\max_{\rho_{A,n}\in R_n} \Tr[Z_{A,n} \rho_{A,n}] + n\epsilon
= P_B^* + n\epsilon
\ee
\noindent
which is the desired result since $\epsilon>0$ is arbitrary and we already
know by weak duality that $\bra{\psi_{A,0}} Z_{A,0} \ket{\psi_{A,0}}\geq
P_B^*$.

Let us assume that $\epsilon>0$ has been given and that
$Z_{A,i+2},\dots,Z_{A,n}$ have been constructed according to the above
criteria. We shall find a $Z_{A,i}$ satisfying the criteria as well.

Let $\Gamma = U_{A,i+1}^\dagger \lp(Z_{A,i+2}\otimes I_\HM\rp) U_{A,i+1}$,
define $\Pos(\HA)$ to be the set of positive semidefinite operators on $\HA$
and define the function $f:\Pos(\HA)\rightarrow \R$ by 
\be
f(\rho) = \max_{\substack{\tilde \rho\in\Pos(\HA\otimes\HM) 
\cr \Tr_\HM [\tilde \rho]=\rho}} \Tr[\Gamma \tilde \rho],
\ee 
\noindent
where we are maximizing over positive semidefinite operators $\tilde \rho$
on on $\HA\otimes\HM$ whose partial trace is $\rho$.  The function $f$ has
a number of simple to verify properties
\begin{itemize}
\item $\max_{\rho\in R_n} f(\rho) = 
\max_{ \rho_{A,i+2}\in R_{i+2}} \Tr[Z_{A,i+2} \rho_{A,i+2}]\equiv \gamma$.
\item $f$ is continuous.
\item $f$ is concave (equivalently $f(\rho)+f(\rho')\leq f(\rho+\rho')$).
\end{itemize}
\noindent
where we used the first equality to define $\gamma$.

We now aim to construct a convex set out of $f$ by using the space of
points below its graph. More specifically let $V$ be the vector space of
Hermitian operators on $\HA$. We want to think of $V$ as a real Hilbert
space of dimension $(\dim\HA)^2$ with inner product
$\braket{H}{H'}=\Tr[H^\dagger H']$. Let $W = V\times \R$, which can be
parametrized by ordered pairs $(H,a)$ where $H$ is a Hermitian operator on
$\HA$ and $a\in\R$. Now define the set $X\subset W$ by
\be
X = \Big\{(H,a) : H\in\Pos(\HA), \Tr H \leq 2 \, 
\text{ and } -1\leq a \leq f(H)
\Big\}.
\ee
\noindent
The numbers $2$ and $-1$ are arbitrary, and are mainly there to make $X$
compact. It is easy to check that $X$ is a convex and has non-empty
interior.

Now we define a second compact convex set $Y$ which will be disjoint from
$X$ but will sit above it (see Fig.~\ref{fig:strongdual}). We use the fact
that $f$ is continuous to find a $\delta>0$ such that
\be
\max_{\rho\in\Pos(\HA), \dist(\rho,R_i)\leq \delta} f(\rho) \leq
\max_{\rho\in R_i} f(\rho) +
\epsilon \equiv \gamma + \epsilon,
\ee
\noindent
where $\dist(\rho,R_i)$ is the $\ell_2$ distance. The idea is that we want
$Y$ to sit atop the set $R_i$ but we also want $Y$ to have non-empty
interior. Therefore we expand $R_i$ to a small (closed) neighborhood of
$R_i$ still consisting of positive semidefinite matrices $R_i^\delta =
\{\rho\in\Pos(\HA):\dist(\rho,R_i)\leq \delta\}$.  Now we define
$Y\subset W$ as
\be
Y = \Big\{(H,a): H\in R_i^\delta 
\text{ and } \gamma+2\epsilon\leq a \leq \gamma+2\epsilon+1
\Big\}.
\ee
\noindent
The set $Y$ is convex because $R_i$ was convex and the distance function is
convex. $Y$ is also compact with non-empty interior and disjoint from $X$.

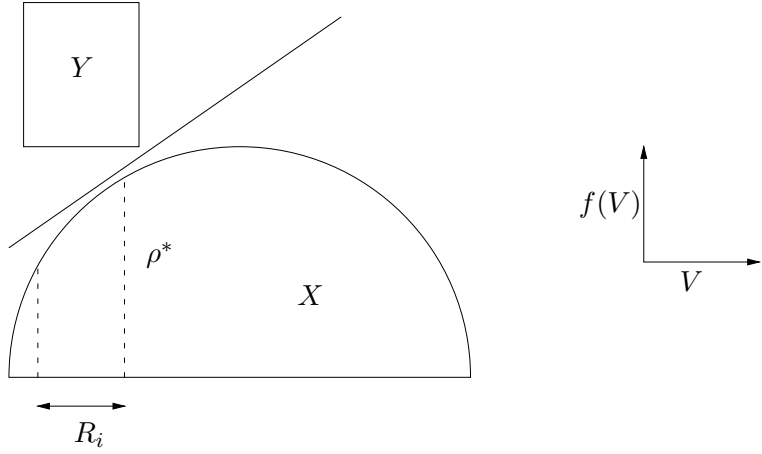
\begin{figure}[tb]
\begin{center}
\setlength{\unitlength}{0.0005in}
\begin{picture}(7824,4602)(0,-10)
\put(2412.000,675.000){\arc{4800.000}{3.1416}{6.2832}}
\dashline{60.000}(312,675)(312,1800)
\dashline{60.000}(1212,675)(1212,2700)
\blacken\path(432.000,405.000)(312.000,375.000)(432.000,345.000)(432.000,405.000)
\path(312,375)(1212,375)
\blacken\path(1092.000,345.000)(1212.000,375.000)(1092.000,405.000)(1092.000,345.000)
\path(6612,1875)(7812,1875)
\blacken\path(7692.000,1845.000)(7812.000,1875.000)(7692.000,1905.000)(7692.000,1845.000)
\path(6612,1875)(6612,3075)
\blacken\path(6642.000,2955.000)(6612.000,3075.000)(6582.000,2955.000)(6642.000,2955.000)
\path(162,3075)(1362,3075)(1362,4575)
	(162,4575)(162,3075)
\path(12,675)(4812,675)
\path(12,2025)(3462,4425)
\put(687,0){$R_i$}
\put(1437,1875){$\rho^*$}
\put(6987,1575){$V$}
\put(5937,2400){$f(V)$}
\put(650,3800){$Y$}
\put(3012,1425){$X$}
\end{picture}
\caption{The sets $X$ and $Y$. The horizontal axis is a cross section
through the density operators given by $\rho =
x\ket{0}\bra{0}+(1-x)\ket{1}\bra{1}$. The curve is given by $f(x) =
(\sqrt{x}+\sqrt{1-x})^2/2$ corresponding to $\Gamma$ being a projector onto
a bell state. The optimal state $\rho^*$ maximizes $f$ on the feasible set
$R_i$. The ideal hyperplane is the tangent to $f$ at
$\rho^*$. Unfortunately, it is quite common that along some cross sections
$R_i$ corresponds to a single boundary point where the slope of $f$ is
infinite. The non-zero width of $Y$ guarantees a finite slope for our
(non-optimal) hyperplane even in this case.}
\label{fig:strongdual}
\end{center}
\end{figure}

Now we use the separating hyperplane theorem (see for instance
\cite{convex} \S 2.5) which says that given two disjoint compact
convex sets there exists a hyperplane such that each set is on one side
(a proof sketch is let $x\in X$ and $y\in Y$ attain the minimum distance
between $X$ and $Y$, then define the hyperplane by the equation 
$(x-y)\cdot (w -(x+y)/2)) = 0$ for $w\in W$).

Because $R_i^\delta$ is of non-empty interior in $V$, and $X$ and $Y$
respectively provide lower and upper bounds on the hyperplane is this
region, the hyperplane cannot be vertical and its normal vector can be
written in the form $(M,1)$ where $M$ is Hermitian. The hyperplane itself
can be parametrized as $(M,1)\cdot(H,a) = c$ for $(H,a)\in W$ and some
parameter $c\in\R$. As the hyperplane separates the sets $X$ and $Y$ we can
write
\begin{itemize}
\item $(H,a)\in X \ \Rightarrow\   c - \Tr M H \geq a,$
\item $(H,a)\in Y \ \Rightarrow\   c - \Tr M H \leq a.$
\end{itemize}
\noindent
We are now ready to choose $Z_{A,i} = c I - M$. The two properties above
tell us that
\begin{itemize}
\item For all density operators $\rho$ on $\HA$ we have 
$\Tr[Z_{A,i}\rho]\geq f(\rho)$.\\
$\Longrightarrow$  For all density operators $\tilde \rho$ on 
$\HA\otimes\HM$ we have
$\Tr\lp[(Z_{A,i}\otimes I) \tilde\rho\rp]
\geq \Tr[\Gamma \tilde \rho]$.\\
$\Longrightarrow Z_{A,i}\otimes I\geq U_{A,i+1}^\dagger
\lp(Z_{A,i+2}\otimes I_\HM\rp) U_{A,i+1}$.
\item For all $\rho\in R_i$ we have $\Tr[Z_{A,i}\rho]\leq \gamma+2\epsilon$.\\
$\Longrightarrow \max_{\rho_{A,i}\in R_i} \Tr[Z_{A,i} \rho_{A,i}] \leq 
\max_{ \rho_{A,i+2}\in R_{i+2}} \Tr[Z_{A,i+2} \rho_{A,i+2}] + 2 \epsilon$.
\end{itemize}

Therefore our chosen $Z_{A,i}$ satisfies the two required properties and
we have completed the proof of strong duality for coin flipping.

\newpage

\section{\label{sec:f2m}From functions to matrices}

TDPGs are specified in terms of pairs of functions constituting valid
transitions. In order to compile them back into the language of quantum
mechanics and semidefinite programming, we need to have a way of extracting
matrices (states and unitaries) out of these transitions. That is our goal
below. As a corollary we shall prove Lemmas~\ref{lemma:f2mmain}
and~\ref{lemma:stdformmain}. This section uses the notation from
Section~\ref{sec:TDPG}.

We begin our discussion by considering the following alternate condition on
transitions:

\begin{definition}
Let $p(z)$ and $q(z)$ be two functions $[0,\infty)\rightarrow
[0,\infty)$ with finite support. We say $p(z)\rightarrow q(z)$ is
\textbf{expressible by matrices} if there exists positive semidefinite
operators $X$ and $Y$, and an (unnormalized) vector $\ket{\psi}$ such that
$X\leq Y$ and $p(z)=\Prob(X,\ket{\psi})$ and $q(z)=\Prob(Y,\ket{\psi})$.
\end{definition}

\noindent
It is not hard to verify that all transitions expressible by matrices are
also valid, justifying our notation. It is also true that all transitions
constructed from UBPs are expressible by matrices (the only non-trivial
step is to expand the Hilbert space so as to purify the density
operator). What we will end up showing by the end of this section is
essentially the converse: that all strictly valid transitions are
expressible by matrices. This is the content of Lemma~\ref{lemma:f2mmain}.

Given that valid transitions and transitions expressible by matrices are
essentially equivalent concepts, we could have used the latter in our
definition of TDPGs. However, that would diminish one of the main
accomplishments of TDPGs: moving the problem of coin flipping outside the
traditional realm of quantum-mechanics/matrices/SDPs.

We note that there do exist valid (but not strictly valid) transitions that
are not expressible by matrices unless we allow infinite
eigenvalues. However, all operators in this paper are assumed to have finite
dimension and finite eigenvalues.

The restriction to strictly valid transitions can be lifted if we work with
functions defined on a compact domain rather than our usual domain of
$[0,\infty)$. For this section it will be useful to work with such compact
domains, and we extend to them our definitions of valid and expressible by
matrices:

\begin{definition}
Fix a compact interval $[a,b]$. Given two functions
$p(z),q(z):[a,b]\rightarrow[0,\infty)$ with finite support we say
$p\rightarrow q$
\begin{itemize}
\item is \textbf{valid on \boldmath $[a,b]$} if $\sum_z p(z)=\sum_z q(z)$
and $\sum_z \frac{\lambda z}{\lambda+z}(q(z) - p(z))\geq 0$ for all
$-\lambda\in\R\setminus I$ .
\item is \textbf{expressible by matrices in \boldmath $[a,b]$} if there
exists matrices $X$, $Y$ with spectrum in $[a,b]$, and a vector
$\ket{\psi}$ such that $X\leq Y$ and $p(z)=\Prob(X,\ket{\psi})$ and
$q(z)=\Prob(Y,\ket{\psi})$.
\end{itemize}
\end{definition}

\noindent
These definitions become useful with the following lemma:

\begin{lemma}
Given two functions $p(z),q(z):[0,\infty)\rightarrow[0,\infty)$ with finite
support such that $p\rightarrow q$ is strictly valid, there exists
$\Lambda>0$, larger than the maximum of the supports of $p$ and $q$, such
that $p\rightarrow q$ is valid on $[0,\Lambda]$.
\label{lemma:inf2lamb}
\end{lemma}

\begin{proof}
Because $p\rightarrow q$ is strictly valid $\sum_z z (q(z) - p(z))>0$,
which implies
\be
\lim_{\lambda\rightarrow-\infty} \sum_z \frac{\lambda z}{\lambda+z} 
(q(z) - p(z))>0,
\ee
\noindent
where the expression inside the limit is only defined for $|\lambda|$
larger than the maximum of the support of $p$ and $q$. The expression
inside the limit is continuous as a function of $\lambda$, so there must
exists a finite $\Lambda>0$ such that for $\lambda\leq-\Lambda$ we also
have satisfy the inequality $\sum_z \frac{\lambda z}{\lambda+z} (q(z) -
p(z))>0$.
\end{proof}

\noindent
We can now restrict our attention to probability distributions and operator
monotone functions with domain $[0,\Lambda]$ and matrices with eigenvalues
in $[0,\Lambda]$, for some large $\Lambda>0$. 

One approach to proving that valid on $[0,\Lambda]$ implies expressible by
matrices in $[0,\Lambda]$ is as follows: define $K$ to be the set of
functions with finite support of the form $g(z)\equiv
q(z)-p(z):[0,\Lambda]\rightarrow \R$ where $p\rightarrow q$ is expressible
by matrices with eigenvalues in $[0,\Lambda]$. The set $K$ is a convex
cone, and we can define its dual cone $K^*$ in the space of functions with
arbitrary support $f(z):[0,\Lambda]\rightarrow\R$. The inner product
between the two spaces is defined by $\braket{f}{g}=\sum_z f(z) g(z)$ which
is well defined because $g$ has finite support. It is easy to check that
$K^*$ is exactly the set of operator monotone functions with support
$[0,\Lambda]$. The dual $K^{**}$ of $K^*$ is the set of functions with
finite support of the form $g(z)\equiv q(z)-p(z):[0,\Lambda]\rightarrow \R$
where $p\rightarrow q$ is valid. Proving $K^{**}=\closure(K)$ covers
most of what we want to prove.  The two difficulties with this approach is
that it requires some fairly advanced analysis to properly place $K$ and
$K^*$ into a pair of locally convex dual topological vector spaces (see for
instance \cite{convex2} \S IV.4), and that in the end we prove something
slightly weaker than needed: mainly we find a transition expressible by
matrices $p'\rightarrow q'$ so that $q'(z)-p'(z)=q(z)-p(z)$ rather than the
stronger conditions $p'(z)=p(z)$ and $q'(z)=q(z)$.

Instead we will use a more constructive approach to completing the proof:
Given $p\rightarrow q$ valid on $[a,b]$, we will construct a
perturbation $p'$ of $p$ so that $p'\rightarrow q$ is still valid on
$[a,b]$ and such that proving that this new transition is expressible
by matrices in $[a,b]$ will also prove that the original transition
is expressible by matrices in $[a,b]$. Alternatively, we may find a
perturbation $q'$ of $q$ and reduce the problem to studying $p\rightarrow
q'$. A sequence of such transformations can be used until we end up with a
valid transition that is also trivially expressible by matrices.

To formalize the perturbations we need to introduce some notation.  Given
$p:[a,b]\rightarrow[0,\infty)$ with finite support we define a
\textit{canonical representation} $p=\Prob(X,\ket{\psi})$ by choosing $X$
diagonal with a non-degenerate set of eigenvalues equal to $S(p)$, the
support of $p$, and then choosing $\ket{\psi}=\sum_{z\in S(p)}
\sqrt{p(z)}\ket{z}$. In particular, the dimension of $X$ is the size of the
support of $p$, denoted by $|S(p)|$. Similarly we can construct a
canonical representation for $q:[a,b]\rightarrow[0,\infty)$ as
$q=\Prob(Y,\ket{\xi})$ with the spectrum of $Y$ is equal to $S(q)$. Note,
however, that even if $p\rightarrow q$ is valid on $[a,b]$ the matrices $X$
and $Y$ so constructed have no a priori relation, and in general are of
different dimensions.

All perturbations will be of the following form: let $c>0$ be a constant
and $\ket{\phi}$ a non-zero vector. Set $X'=X+c\ket{\phi}\bra{\phi}$ and
(assuming $X'$ has eigenvalues in $[a,b]$) set
$p'=\Prob(X',\ket{\psi})$. Then $p\rightarrow p'$ is trivially
constructable by matrices in $[a,b]$.  Similarly, we can set
$q'=\Prob(Y',\ket{\xi})$ for $Y'=Y+c\ket{\phi}\bra{\phi}$ where now we want
$c<0$. Ensuring that the perturbations can be chosen so that $p'\rightarrow
q$ or $p\rightarrow q'$ are valid will take up most of the rest of the
section.

Our first simple result bounds the dimension of the space needed when
studying general transitions that are expressible by matrices. It also
standardizes the spectra of the matrices involved.

\begin{lemma}
Let $p\rightarrow q$ be expressible by matrices in $[a,b]$, then we can
find matrices $X\leq Y$ and a vector $\ket{\phi}$ such that $p=
\Prob(X,\ket{\psi})$ and $q=\Prob(Y,\ket{\psi})$ and additionally:
\begin{enumerate}
\item The spectrum of $X$ is equal to $\{a\}\cup S(p)$,
with all eigenvalues (excluding $a$) occurring once.
\item The spectrum of $Y$ is equal to $\{b\}\cup S(q)$,
with all eigenvalues (excluding $b$) occurring once.
\item The dimension of $X$ and $Y$ is no greater than 
$|S(p)|+|S(q)|-1$.
\end{enumerate}
\label{lemma:stdmattrans}
\end{lemma}

\begin{proof}
The fact that $p\rightarrow q$ is expressible by matrices in $[a,b]$,
guarantees the existence of matrices $X\leq Y$ with spectrum in $[a,b]$
and a vector $\ket{\phi}$ such that $p= \Prob(X,\ket{\psi})$ and
$q=\Prob(Y,\ket{\psi})$. What we need to prove that we can modify the
given matrices to satisfy the additional properties.

We begin by working with $X$ and $\ket{\psi}$. We can write
$\ket{\psi}=\sum_z \sqrt{p(z)} \ket{z;X}$ where $z$ ranges over the support
of $p$ and $\ket{z;X}$ is a normalized eigenvector of $X$ with eigenvalue
$z$. Let $\Pi=\sum_z\ket{z;X}\bra{z;X}$ be the projector onto the spaced
spanned by these eigenvalues. If we define $X'=\Pi X \Pi + a (I-\Pi)$ we
obtain a new matrix with eigenvalues in $[a,b]$ that is diagonal in the
same basis as $X$, but may have some eigenvalues changed to $a$. Therefore,
$X' \leq X \leq Y$. Furthermore, $X'$ has the desired spectrum and
$p=\Prob(X',\ket{\psi})$.

We can do something similar with $Y$. We can again write $\ket{\psi}=\sum_z
\sqrt{q(z)} \ket{z;Y}$ where $z$ ranges over the support of $q$ and 
$\ket{z;Y}$ is a normalized eigenvector of $Y$ with eigenvalue $z$. Note
that even if $p$ and $q$ both have support on $z$ we may have
$\ket{z;X}\neq\ket{z;Y}$. We now set $\Pi=\sum_z\ket{z;Y}\bra{z;Y}$ and
define $Y'=\Pi Y\Pi + b (I-\Pi)$. The new matrix satisfies $X\leq Y\leq
Y'$, $q=\Prob(Y',\ket{\psi})$, and has the desired spectrum.

Everything is correct except for the dimension. Now let $\Pi$ be the
projector onto the space spanned by both $\{\ket{z;X}\}$ and
$\{\ket{z;Y}\}$. The dimension of this space is no greater than
$|S(p)|+|S(q)|-1$ (the minus one occurring because $\ket{\psi}$ is in the
span of both sets of vectors). Clearly $\Pi\ket{\psi}=\ket{\psi}$ and both
$X'$ and $Y'$ are block diagonal with respect to $\Pi$ and
$I-\Pi$. Therefore, the required objects are $X'$, $Y'$ and $\ket{\psi'}$
restricted to the support of $\Pi$.
\end{proof}

\begin{corollary}
Expressible by matrices in $[a,b]$ is a transitive relation.
\end{corollary}

\begin{proof}
Let $p\rightarrow r$ and $r\rightarrow q$ be expressible by matrices in
$[a,b]$. We use $X_1$, $Y_1$ and $\ket{\psi_1}$ for the first transition
and $X_2$, $Y_2$ and $\ket{\psi_2}$ for the second transition, all chosen
in accordance with the conditions of Lemma~\ref{lemma:stdmattrans}. The
matrices $Y_1$ and $X_2$ have the same spectrum, except that the second one
has extra $a$ eigenvalues and the first one has extra $b$ eigenvalues. We
can append eigenvalues to both of them using a direct sum so that
\be
\mypmatrix{Y_1&0\cr0&a I}\qquad\text{and}\qquad
\mypmatrix{X_2&0\cr0&b I}
\ee
have the same spectrum and dimension, where the blocks may be different
sized. We can also map $\ket{\psi_1}$ and $\ket{\psi_2}$ into this
enlarged space by using a direct sum with a zero vector.

Because the two unitaries have the same spectrum, there exists a unitary
that maps the second into the first by conjugation. We can further ask 
that the unitary satisfy $U\ket{\psi_2}=\ket{\psi_1}$ because up to a phases
they assign the same coefficient to each eigenvector. Then
\be
\mypmatrix{X_1&0\cr0&a I}\leq
\mypmatrix{Y_1&0\cr0&a I}=
U\mypmatrix{X_2&0\cr0&b I}U^\dagger\leq
U\mypmatrix{Y_2&0\cr0&b I}U^\dagger.
\ee
We can construct $p$ out of the first matrix and the extended
$\ket{\psi_1}\bra{\psi_1}$ and we can construct $q$ out of the last matrix
and $\ket{\psi_1}\bra{\psi_1}$ as well. Therefore $p\rightarrow q$ is
expressible by matrices in $[a,b]$.
\end{proof}

Before studying the perturbations, we need one final simplification. Though
the domain $[0,\Lambda]$ is good, the domain $[-1,1]$ is even better
because we can write $\frac{\lambda z}{\lambda+z}=\frac{z}{1+\gamma
z}$ and $-\lambda\in\R\setminus[-1,1]$ is equivalent to $\gamma\in(-1,1)$,
which is a connected set. Note that the value $\gamma=0$ corresponds to the
function $f(z)=z$, and by continuity in $\gamma$ we can always
include/exclude it among our conditions. The following lemma follows by a
simple rescaling argument.

\begin{lemma}
If every valid transition on $[-1,1]$ is expressible by matrices in $[-1,1]$
then every valid transition on $[0,\Lambda]$ is expressible by matrices in
$[0,\Lambda]$.
\label{lemma:lamb21}
\end{lemma}

It will also be convenient to work with functions $p$ and $q$ with support
in $(-1,1)$, though we still keep a domain of $[-1,1]$.

\begin{lemma}
If every valid transition on $[-1,1]$ involving functions with support on
$(-1,1)$ is expressible by matrices in $[-1,1]$ then every valid transition
on $[-1,1]$ is expressible by matrices in $[-1,1]$.
\end{lemma}

\begin{proof}
Fix $p,q:[-1,1]\rightarrow[0,\infty)$ with finite support so that
$p\rightarrow q$ is valid in $[-1,1]$. For any $0<c<1$
we define functions $[-1,1]\rightarrow[0,\infty)$ by
\be
p_c(z) = \begin{cases}
p(\frac{z}{c}) & z\in [-c,c],\cr
0 & \text{otherwise,}
\end{cases}
\qquad\qquad
q_c(z) = \begin{cases}
q(\frac{z}{c}) & z\in [-c,c],\cr
0 & \text{otherwise,}
\end{cases}
\ee
which have support in $(-1,1)$ and furthermore $p_c\rightarrow q_c$ is
valid in $[-1,1]$. Therefore, the transition is expressible by matrices in
$[-1,1]$ and and we can choose $X_c$, $Y_c$ and $\ket{\psi_c}$ in
accordance with the constraints of
Lemma~\ref{lemma:stdmattrans}. Furthermore, by changes of basis we can
ensure that $X_c$ and $\ket{\psi_c}$ converge as $c\rightarrow 1$ to the
canonical representation of $p$ (appended with $-1$ eigenvalues). Any
limit point of $Y_c$ as $c\rightarrow 1$ will complete the proof.
\end{proof}

For the rest of this section we will be concerned with the interval
$[-1,1]$. When clear from context we will use the terms valid and
expressible by matrices to refer to valid on $[-1,1]$ and expressible by
matrices in $[-1,1]$.

We now begin studying perturbations on $p$ and $q$ which will have
constraints arising from rational functions in $\gamma$ of the form $\sum_z
\frac{z}{1+\gamma z}(q(z)-p(z))$. The main difficulty will be near the
zeros of the function, which we want to approximate by simple expressions
of the form $a|\gamma-\gamma_0|^{k}$. The following is a standard result.

\begin{lemma}
Let $N(\gamma)$ and $D(\gamma)$ be two non-negative polynomials for
$\gamma\in[-1,1]$. If $N$ has a zero of order $k$ at
$\gamma=\gamma_0\in[-1,1]$ and $D(\gamma_0)\neq 0$ then for any $k'\geq k\geq
k''>0$ there exists $\epsilon,a,b>0$ such that
\be
a|\gamma-\gamma_0|^{k'} \leq \frac{N(\gamma)}{D(\gamma)} \leq
b|\gamma-\gamma_0|^{k''}\qquad\text{for $\gamma\in[-1,1]$ satisfying
$|\gamma-\gamma_0|<\epsilon$}.
\ee
\label{lemma:poly}
\end{lemma}

\begin{proof}
Choose $\epsilon>0$ so that $N(\gamma)$ and $D(\gamma)$ are non-zero for
$0<|\gamma-\gamma_0|\leq \epsilon$. Let $a$ be the infimum of
$\frac{N(\gamma)}{|\gamma-\gamma_0|^{k'} D(\gamma)}$ and $b$ be the
maximum of $\frac{N(\gamma)}{|\gamma-\gamma_0|^{k''} D(\gamma)}$ for
$\gamma\in[-1,1]$ satisfying $0<|\gamma-\gamma_0|\leq \epsilon$. Note that
the cases $\gamma_0=-1$ and $\gamma_0=1$ are special in that $k$ can be
odd, and the inequalities will only hold inside the region $[-1,1]$.
\end{proof}

For the next set of lemmas let $\HH$ be a real $n$-dimensional Hilbert
space, let $X$ be an operator on $\HH$, let $\ket{\psi}$ and $\ket{\phi}$
be non-zero vectors in $\HH$ and let $c\neq 0$ be a real constant. Also let
$f_\gamma(x) = \frac{x}{1+\gamma x}$ for $x\in[-1,1]$ and $\gamma\in(-1,1)$.

\begin{lemma}
Let $X$, $\ket{\phi}$, $c$ and $f_\gamma(x)$ be as above. If both $X$ and
$X+c\ket{\phi}\bra{\phi}$ have eigenvalues in $[-1,1]$, then
\be
f_\gamma(X+c\ket{\phi}\bra{\phi}) - f_\gamma(X) =
\lp(\frac{c}{1+\gamma c \bra{\phi}
\lp(I+\gamma X\rp)^{-1}\ket{\phi}}\rp)
\lp(I+\gamma X\rp)^{-1}\ket{\phi}\bra{\phi}
\lp(I+\gamma X\rp)^{-1}.
\label{eq:fgbound}
\ee
\end{lemma}

\begin{proof}
The proof for $\gamma=0$ is trivial. Otherwise $f_\gamma(x) = \frac{1}{\gamma}
\lp(1-\frac{1}{1+\gamma x}\rp)$. Let $Z=I+\gamma X$ and $c'=\gamma c$.
Because $Z$ is positive definite and
$Z+c'\ket{\phi}\bra{\phi}$ is also positive definite, then
\be
\big(Z + c'\ket{\phi}\bra{\phi}\big)^{-1}&=&
\lp( Z^{1/2} \lp( I + c' Z^{-1/2}\ket{\phi}\bra{\phi}Z^{-1/2}\rp)
Z^{1/2}\rp)^{-1}
\nonumber\\
&=& Z^{-1/2} \lp( I + c' Z^{-1/2}\ket{\phi}\bra{\phi}Z^{-1/2}\rp)^{-1}
Z^{-1/2}
\nonumber\\
&=& Z^{-1/2} \lp( I + \lp( -1 + \frac{1}{1+c'\bra{\phi}Z^{-1}\ket{\phi}}\rp)
\frac{Z^{-1/2}\ket{\phi}\bra{\phi}Z^{-1/2}}{\bra{\phi}Z^{-1}\ket{\phi}}\rp)
Z^{-1/2}
\nonumber\\
&=& Z^{-1/2} \lp( I - \lp(\frac{c'}{1+c'\bra{\phi}Z^{-1}\ket{\phi}}\rp)
Z^{-1/2}\ket{\phi}\bra{\phi}Z^{-1/2}\rp)
Z^{-1/2}
\nonumber\\
&=& Z^{-1} - \lp(\frac{c'}{1+c'\bra{\phi}Z^{-1}\ket{\phi}}\rp)
Z^{-1}\ket{\phi}\bra{\phi}Z^{-1},
\ee
\noindent
where the third equality follows because $I + c'
Z^{-1/2}\ket{\phi}\bra{\phi}Z^{-1/2}$ is a matrix with only two
eigenvalues. The main result follows because its LHS equals 
$-\frac{1}{\gamma}((Z+c'\ket{\phi}\bra{\phi})^{-1}-Z^{-1})$.
\end{proof}

We are ready to start imposing constraints on the transitions $p\rightarrow
p'$ for $p=\Prob(X,\ket{\psi})$ and $p'=\Prob(X',\ket{\psi})$, where
$X'=X+c\ket{\phi}\bra{\phi}$. In particular, we want to place an upper
bound on $\lp|\sum_z \frac{z}{1+\gamma z} (p'(z)-
p(z))\rp|=\lp|\bra{\psi}\Big(f_\gamma(X+c\ket{\phi}\bra{\phi}) -
f_\gamma(X)\Big)\ket{\psi}\rp|$ in the form of a polynomial with a zero of
order $2k$.  The next lemma will show that there are many choices for
$\ket{\phi}$ that satisfy the bound. In fact, there is an $n-k$ subspace of
such vectors.

\begin{lemma}
Let $\HH$, $n$, $X$, $\ket{\psi}$ and $f_\gamma(x)$ be as above, with the
eigenvalues of $X$ restricted to $(-1,1)$. Given a polynomial
$b(\gamma-\gamma_0)^{2k}$ with constants $b>0$, integer $k>0$ and
$\gamma_0\in[-1,1]$, then we can construct an $(n-k)$-dimensional subspace
$\HHP\subset\HH$ such that for every $\ket{\phi}\in\HHP$ there exists $c>0$
and $\epsilon>0$ satisfying
\beq
\lp|\bra{\psi}\Big(f_\gamma(X+c\ket{\phi}\bra{\phi}) -
f_\gamma(X)\Big)\ket{\psi}\rp|
\leq b(\gamma-\gamma_0)^{2k}
\qquad\text{for $\gamma\in(-1,1)$ satisfying $|\gamma-\gamma_0|<\epsilon$,}
\eeq
where additionally $|c|$ must be small enough so that
$X+c\ket{\phi}\bra{\phi}$ also has eigenvalues in $(-1,1)$.
The same conditions can also be satisfied while demanding that in every
case $c<0$.
\label{lemma:vecadd}
\end{lemma}

\begin{proof}
Because we are only considering matrices $X$ with eigenvalues in $(-1,1)$,
the matrix $|X|$ has a maximum eigenvalue $\lambda_{max}<1$ and if for a
given $\ket{\phi}$ we restrict $|c|<(1-\lambda_{max})/(2\braket{\phi}{\phi})$
then Eq.~(\ref{eq:fgbound}) implies
\be
\lp|\bra{\psi}\Big(f_\gamma(X+c\ket{\phi}\bra{\phi}) -
f_\gamma(X)\Big)\ket{\psi}\rp|
\leq
2|c|\lp| \bra{\psi}\lp(I+\gamma X\rp)^{-1} \ket{\phi}\rp|^2.
\label{eq:fgbound1}
\ee
The expression $\bra{\psi}(I+\gamma X)^{-1} \ket{\phi}$ is a rational
function in $\gamma$ (recall we are working in a real vector space) with no
poles in $[-1,1]$. If we can choose $\ket{\phi}$ such that it has a zero of
order at least $k$ at $\gamma_0$ then by Lemma~\ref{lemma:poly} we can
choose small enough $c>0$ (or large enough $c<0$) and $\epsilon>0$ such
that
\be
2|c|\lp| \bra{\psi}\lp(I+\gamma X\rp)^{-1} \ket{\phi}\rp|^2 
\leq b(\gamma-\gamma_0)^{2k}
\qquad\text{for $\gamma\in(-1,1)$ satisfying $|\gamma-\gamma_0|<\epsilon$,}
\ee
thereby proving the constraint. What remains to be shown is that there
exists an $n-k$ dimensional space $\HHP$ such that $\ket{\phi}\in\HHP$
implies that $\bra{\psi}(I+\gamma X)^{-1} \ket{\phi}$ has a zero of order
at least $k$ at $\gamma_0$.

If we choose $\delta>0$ such that $\delta|I+\gamma_0 X|^{-1}<I$ then we can
write
\be
\bra{\psi}(I+\gamma X)^{-1} &=& 
\bra{\psi}(I+\gamma_0 X + (\gamma-\gamma_0) X)^{-1} 
\\\nonumber
&=& \sum_{j=1}^{\infty}
(\gamma-\gamma_0)^j \bra{\psi} (I+\gamma_0 X)^{-1} 
\lp(\frac{X}{I+\gamma_0 X} \rp)^j
\qquad\text{for $|\gamma-\gamma_0|<\delta$,}
\ee
\noindent
where all matrices are diagonal in the eigenbasis of $X$, which justifies
the unordered products. We have shown that the requirement that
$\bra{\psi}(I+\gamma X)^{-1} \ket{\phi}$ have a zero of order at least $k$
at $\gamma_0$ is equivalent to $k$ linear constraints on $\ket{\phi}$, and
therefore there exists a subspace of dimension $n-k$ that simultaneously
satisfies all of them. The lemma follows so long as we ensure to pick the
constants $\epsilon$ small enough so that $\epsilon\leq\delta$.
\end{proof}

The following fact will also be useful: given $c>0$ and $\epsilon>0$
satisfying the inequality in the above lemma (for a given $\ket{\phi}$),
then it is also satisfied by any $c'$ and $\epsilon'$ such that $0<c'<c$
and $0<\epsilon'<\epsilon$, the former property arising because $f_\gamma$
is operator monotone and hence $\bra{\psi}f_\gamma(X+c\ket{\phi}\bra{\phi})
\ket{\psi}$ is monotone as a function of $c$. In the next lemma we will use
this to deal with multiple simultaneous constraints.

\begin{lemma}
Let $p$ and $q$ be functions with support in $(-1,1)$ such that
$p\rightarrow q$ is valid. Let $S(p)$ and $S(q)$ be respectively the
supports of $p$ and $q$, and let $m$ be the number of zeros (including
multiplicities) for $\gamma\in[-1,1]$ of the rational function $\sum_z
\frac{z}{1+\gamma z} (q(z)- p(z))$.
\begin{itemize}
\item If $m+2< 2|S(p)|$ there exists a vector $\ket{\phi}\neq0$ and
constant $c>0$ such that $p'\rightarrow q$ is valid,
where $p'=\Prob(X',\ket{\psi})$, $X'= X+c\ket{\phi}\bra{\phi}$ and 
$p$ is canonically represented by $\Prob(X,\ket{\psi})$.
\item If $m+2< 2|S(q)|$ there exists a vector $\ket{\phi}\neq0$ and
constant $c<0$ such that $p\rightarrow q'$ is valid,
where $q'=\Prob(Y',\ket{\xi})$, $Y'= Y+c\ket{\phi}\bra{\phi}$ and 
$q$ is canonically represented by $\Prob(Y,\ket{\xi})$.
\end{itemize}
\label{lemma:ppq}
\end{lemma}

\begin{proof}
We will prove the first case, the second case being nearly identical.
The key idea is that $q(z)-p'(z)= (q(z)-p(z))-(p'(z)-p(z))$, and therefore
$p'\rightarrow q$ will be valid if
\be
\bra{\psi}\Big(f_\gamma(X+c\ket{\phi}\bra{\phi}) -
f_\gamma(X)\Big)\ket{\psi}
\leq \sum_z \frac{z}{1+\gamma z} (q(z)- p(z))
\qquad\text{for all $\gamma\in[-1,1]$.}
\label{eq:ppq}
\ee
By construction the right-hand side has no poles and exactly $m$ zeros
(including multiplicities) for $\gamma\in[-1,1]$. For each zero we can
find, by Lemma~\ref{lemma:poly}, a lower bound for the right-hand side of
the form $a(\gamma-\gamma_0)^{2k}$ valid in some neighborhood of the zero.
Since all multiplicities for zeros in $(-1,1)$ are even we can chooses $2k$
to equal the multiplicity of the respective zero. Zeros occurring at the
end points $-1$ and $1$ can have odd multiplicities, in which case we can
choose $2k$ to be at most one plus the multiplicity of the zero. By
Lemma~\ref{lemma:vecadd} for each neighborhood there is an $S(p)-k$
subspace of vectors $\ket{\phi}$ that satisfies the constraint, possibly in
a smaller neighborhood, for some $c>0$. The sum of the constants $2k$ over
each of the zeros is therefore at most $m+2$ (and this can only occur if
there are zeros of odd order at both $\gamma_0=-1$ and $\gamma_0=1$).
Therefore, the number of linear constraints needed to specify all the
subspaces is at most $\lfloor\frac{m+2}{2}\rfloor<|S(p)|$. Therefore, the
intersection of all these subspaces is non-trivial and we can find a
non-zero $\ket{\phi}$ and a $c>0$ (chosen as the smallest of the given $c$
constants) such that the inequality of Eq.~(\ref{eq:ppq}) is satisfied in
the union of all the neighborhoods.

We also need to ensure that $c>0$ is chosen small enough so that
$X'=X+c\ket{\phi}\bra{\phi}$ has eigenvalues in $(-1,1)$, but because $X$
has eigenvalues in $(-1,1)$ that is always possible. What remains to be
checked is that the inequality is satisfied outside the neighborhoods
surrounding the zeros. But the complement of these neighborhoods in
$[-1,1]$ is a compact set on which the inequality becomes a strict
inequality. Therefore, for any vector $\ket{\phi}$ we can find a $c>0$ such
that the inequality holds. The lemma is proven by choosing $c>0$ small
enough to satisfy all the preceding conditions.
\end{proof}

To make use of the above lemma we note that the degree of the
numerator of the rational function of $\gamma$ given by
\be
\sum_z \frac{z}{1+\gamma z} (q(z)- p(z))
\ee
is at most $|S(p)|+|S(q)|-1$ because we are summing at most $|S(p)|+|S(q)|$
terms. In fact, the degree of the numerator is at most $|S(p)|+|S(q)|-2$
because the coefficient of $\gamma^{|S(p)|+|S(q)|-1}$ is $(\prod_z
z)\sum_z(q(z)-p(z))=0$.

In the following discussion \textit{the number of zeros of $p\rightarrow
q$} refers to the number of zeros (including multiplicities) for
$\gamma\in[0,1]$ of the numerator of the rational function $\sum_z
\frac{z}{1+\gamma z} (q(z)- p(z))$. The argument from the last paragraph
shows that the number of zeros of $p\rightarrow q$ is at most
$|S(p)|+|S(q)|-2$.

As a consequence, if $|S(p)|>|S(q)|$, then $2|S(p)|>|S(p)|+|S(q)|\geq m+2$
where $m$ is the number of zeros of $p\rightarrow q$. Therefore, we can
always find a $\ket{\phi}$ as in the above lemma to perturb $p$. Similarly,
if $|S(p)|<|S(q)|$ there exists a $\ket{\phi}$ that can be used as a
perturbation on $q$. Both perturbations can be found if $|S(p)|=|S(q)|$ and
the number of zeros of $p\rightarrow q$ is less than its maximal value of
$2|S(p)|-2$. The special case of $|S(p)|=|S(q)|$ with maximal zeros will
be dealt with separately later in the section

Let us now discuss what happens once we have found a valid $\ket{\phi}$ and
$c$, and start increasing the value of $c$. For concreteness, we take the
first case of the preceding lemma so that $p'\rightarrow q$ is valid,
$p'=\Prob(X',\ket{\psi})$, $p=\Prob(X,\ket{\psi})$ is a canonical
representation so that $X$ has dimension $|S(p)|$, and
$X'=X+c\ket{\phi}\bra{\phi}$.

As we increase $c$ there are two constraints that can force us to stop:
either $X'(c)\equiv X+c\ket{\phi}\bra{\phi}$ gets an eigenvalues larger
than $1$, or $p'\rightarrow q$ is no longer valid. Let $c_0$ be the largest
allowed value of $c$ and let $p'_0 =\Prob(X'(c_0),\ket{\psi})$. Note that
$p\rightarrow p_0'$ is still expressible by matrices and $p_0'\rightarrow
q$ is still valid because the respective constraints are defined using
non-strict inequalities.

If the stopping condition for increasing $c$ is that the eigenvalues get
too large, then $X'(c_0)\leq I$ but $X'(c_0)$ has $1$ as an eigenvalue.
Because $q$ has support in $(-1,1)$, the value of $\sum_z
\frac{z}{1+\gamma z} q(z)$ at $\gamma=-1$ is finite, so as $c\rightarrow
c_0$ the amplitude of $\ket{\psi}$ on the largest eigenvector of $X'(c)$
must be going to zero, and must be zero at $c=c_0$. The first consequence
of this is that $p'_0(z)$ ultimately does not have support on $z=1$, and
therefore its support is still contained in $(-1,1)$. The second
consequence is that the size of the support of $p'_0$ is smaller
than the dimension of the space on which $X$ is defined, which equals the
support of $p$. In other words $|S(p'_0)|<|S(p)|$.

On the other hand lets assume that for $c>c_0$ the transition to $q$ is no
longer valid. In such a case the rational function $\sum_z \frac{z}{1+\gamma
z} (q(z)- p'_0(z))$ must have an extra zero in $[-1,1]$ that $\sum_z
\frac{z}{1+\gamma z} (q(z)- p(z))$ did not have, or one of the existing
zeros must have a higher degree. That is, the number of zeros in
$p_0'\rightarrow q$ is greater than the number of zeros in $p\rightarrow
q$. Note that in this case, by construction $|S(p'_0)|\leq|S(p)|$, whereas
in the previous case (where we proved $|S(p'_0)|<|S(p)|$) the number of
zeros cannot decrease.

Let $p\rightarrow q$ be valid with $|S(p)|>|S(p)|$. We can repeatedly apply
the above perturbation process to obtain a sequence of valid transitions
$p\equiv p_0\rightarrow p_1\rightarrow\cdots\rightarrow p_l\rightarrow q$,
such that, as functions of $i$, the number of zeros of $p_i\rightarrow q$
is monotonically increasing along the chain and $|S(p_i)|$ is monotonically
decreasing along the chain. Furthermore, in each transition either the
number of zeros increases or $|S(p_i)|$ decreases. The number of zeros is
upper bounded by $|S(p)|+|S(q)|-2$, and therefore after a finite number of
steps we must get to a $p_\ell$ such that either $|S(p_\ell)|<|S(q)|$ or
$|S(p_\ell)|=|S(q)|$ and the number of zeros is maximal.  In the former case,
we can continue the chain by perturbing $q$. Repeatedly working on both
sides we can end up with a chain
\be
p \equiv p_0\rightarrow p_1\rightarrow\cdots\rightarrow p_\ell\equiv p'
\rightarrow q'\equiv q_{\ell'}\rightarrow\cdots\rightarrow q_1
\rightarrow q_0\equiv q
\ee
\noindent
such that all transitions are valid, and all but $p'
\rightarrow q'$ are known to be expressible by matrices. Furthermore,
$|S(p')|=|S(q')|$ and $p'\rightarrow q'$ has $2|S(p')|-2$ zeros.  If we
can prove that $p'\rightarrow q'$ is expressible by matrices then by
transitivity we will have also proven that $p\rightarrow q$ is expressible
by matrices.

To complete the result of this section, we now need to study the remaining
special case of $|S(p')|=|S(q')|$ with maximal zeros. From the proof of
Lemma~\ref{lemma:ppq} we can see that in fact the only case when we can't
find a perturbation is when there are zeros of odd order at both
$\gamma_0=-1$ and $\gamma_0=1$. In all other cases we can continue the
above chain until $p'=q'$.

To deal with zeros of odd order at $\gamma_0=-1$ and $\gamma_0=1$ we need
to enlarge the vector space on which we are working. Rather than using a
canonical representation $p=\Prob(X,\ket{\psi})$ we append to $X$ an extra
eigenvalue $-1$, so that we end up with a new matrix $\bar X$ having
dimension $|S(p)+1|$ and eigenvalues in $[-1,1)$. Similarly, we can request
a representation $q=\Prob(\bar Y,\ket{\xi})$ so that $\bar Y$ has dimension
$|S(q)|+1$ including a single eigenvalue $1$. Note that it is pointless to
append eigenvalues of $1$ to $X$ (or $-1$ to $Y$) as we can't add
(resp. subtract) in any vectors $\ket{\phi}\bra{\phi}$ to that subspace
without ending up with eigenvalues outside $[-1,1]$.

The following discussion will concern the first case, where the matrix
$\bar X$ has a single eigenvector $\ket{-1}$ with eigenvalue $-1$. It will
be useful to define $\HHB$ as the space containing $\bar X$, and
$\HH\subset\HHB$ as the subspace orthogonal to $\ket{-1}$. We want to keep
$n$ as the size of the support of $p$, so $\HH$ will have dimension $n$ and
$\HHB$ will have dimension $n+1$. We still want $p=\Prob(\bar
X,\ket{\psi})$ to have support in $(-1,1)$, though, so
$\braket{-1}{\psi}=0$ or equivalently $\ket{\psi}\in\HH$. We will consider
perturbations of the form $\bar X+c\ket{\bar \phi}\bra{\bar \phi}$ where
$\ket{\bar\phi}=\ket{\phi}+a\ket{-1}$.

Our immediate goal is to extend Lemma~\ref{lemma:vecadd} to deal with
matrices $X$ with eigenvalues in $[-1,1)$ for $c>0$. For
$\gamma_0\in[-1,1)$ it will be a straightforward extension: We want to find
an $(n-k)$-dimensional subspace $\HHP\subset\HH$ of vectors $\ket{\phi}$ so
that for any $a\in\R$ the vector $\ket{\bar\phi}=\ket{\phi}+a\ket{-1}$
satisfies the inequality
\beq
\lp|\bra{\psi}\Big(f_\gamma(\bar X+c\ket{\bar \phi}\bra{\bar \phi}) -
f_\gamma(\bar X)\Big)\ket{\psi}\rp|
\leq b(\gamma-\gamma_0)^{2k}
\qquad\text{for $\gamma\in(-1,1)$ satisfying $|\gamma-\gamma_0|<\epsilon$}
\eeq
for some $c>0$ and $\epsilon>0$.

As the proof of the above is nearly identical to the proof of
Lemma~\ref{lemma:vecadd}, we will only discuss the differences. Our main
tool is Eq.~(\ref{eq:fgbound}) which also applies to our extend matrix $\bar
X$. For convenience we rewrite it here
\be
f_\gamma(\bar X+c\ket{\bar \phi}\bra{\bar \phi}) - f_\gamma(\bar X) =
\lp(\frac{c}{1+\gamma c \bra{\bar \phi}
\lp(I+\gamma \bar X\rp)^{-1}\ket{\bar \phi}}\rp)
\lp(I+\gamma \bar X\rp)^{-1}\ket{\bar \phi}\bra{\bar \phi}
\lp(I+\gamma \bar X\rp)^{-1}.
\label{eq:fgboundbis}
\ee
\noindent
Given $\ket{\bar \phi}=\ket{\phi}+a\ket{-1}$ we want to choose
$c>0$ small enough so that
\be
\lp|\bra{\psi}\Big(f_\gamma(\bar X+c\ket{\bar \phi}\bra{\bar \phi}) -
f_\gamma(\bar X)\Big)\ket{\psi}\rp|
\leq
2|c|\lp| \bra{\psi}\lp(I+\gamma \bar X\rp)^{-1} \ket{\bar \phi}\rp|^2
\ee
\noindent
for $\gamma\in(-1,1)$ satisfying $|\gamma-\gamma_0|<\epsilon$. To
accomplish this let $\lambda_{max}$ be the largest eigenvalue of $\bar X$
and restrict $c<(1-\lambda_{max}')
/(2\braket{\bar\phi}{\bar\phi})$. We then have $\gamma c \bra{\bar \phi}
\lp(I+\gamma \bar X\rp)^{-1}\ket{\bar
\phi}>-1/2$ as required. We then note that $\bra{\psi}\lp(I+\gamma \bar
X\rp)^{-1} \ket{\bar \phi} = \bra{\psi}\lp(I+\gamma X\rp)^{-1}
\ket{\phi}$ where $X$ is the restriction of $\bar X$ to $\HH$.
We end up with an equation identical to Eq.~(\ref{eq:fgbound1}) in the
proof of Lemma~\ref{lemma:vecadd}. The rest of the proof carries through.

The interesting extension of Lemma~\ref{lemma:vecadd} occurs at $\gamma_0=1$. For $\gamma\in(-1,1)$ we have
\be
|a|^2(1-\gamma)^{-1}
\leq \bra{\bar \phi}\lp(I+\gamma \bar X\rp)^{-1} \ket{\bar \phi}
\ee
\noindent
because $1+\gamma \bar X$ is positive definite and the left-hand side drops
some positive terms. If $a\neq 0$ we can further write for $\gamma\in(0,1)$
\be
\lp(\frac{c}{1+\gamma c \bra{\bar\phi}
\lp(I+\gamma \bar X\rp)^{-1}\ket{\bar \phi}}\rp)\leq
\frac{c}{1+\gamma c |a|^2(1-\gamma)^{-1}}
= \frac{c(1-\gamma)}{1-\gamma(1- c |a|^2)}
\leq
\frac{1-\gamma}{|a|^2},
\ee
\noindent
where in the last step we assumed $c<1/|a|^2$ so that the
denominator is minimized by $\gamma\rightarrow 1$. Combined with
Eq.~(\ref{eq:fgboundbis}) we get
\be
\lp|\bra{\psi}\Big(f_\gamma(\bar X+c\ket{\bar \phi}\bra{\bar \phi}) -
f_\gamma(\bar X)\Big)\ket{\psi}\rp|
\leq
\frac{1-\gamma}{|a|^2}
\lp| \bra{\psi}\lp(I+\gamma \bar X\rp)^{-1} \ket{\bar \phi}\rp|^2
\label{eq:gamma01}
\ee
\noindent
for $\gamma\in(1-\epsilon,1)$ so long as we choose $\epsilon<1$. Therefore,
at $\gamma_0=1$ we can prove something stronger than
Lemma~\ref{lemma:vecadd}. The constructed $(n-k)$-dimensional subspace will
satisfy an upper bound of $b(\gamma-\gamma_0)^{2k+1}$ rather than just
$b(\gamma-\gamma_0)^{2k}$. The caveat is that the bound will only hold when
$|a|$ is large enough:

\begin{lemma}
Let $\HHB$, $\HH$, $n$, $\bar X$, $\ket{\psi}$ and $f_\gamma(x)$ be as
above. In particular, $\bar X$ has eigenvalues in $[-1,1)$ with a unique
eigenvector $\ket{-1}$ with eigenvalue $-1$ and $\braket{-1}{\psi}=0$.
Given a polynomial $b(1-\gamma)^{2k+1}$ with constants $b>0$ and integer
$k\geq 0$, we can find an $(n-k)$-dimensional subspace $\HHP\subset\HH$ and
a number $\Theta>0$ such that for any $\ket{\phi}\in\HHP$ and $a\geq \Theta
\braket{\phi}{\phi}$ there exists
$c>0$ and $\epsilon>0$ satisfying $\bar X+c\ket{\bar\phi}\bra{\bar\phi}<I$
and
\be
\lp|\bra{\psi}\Big(f_\gamma(\bar X+c\ket{\bar \phi}\bra{\bar \phi}) -
f_\gamma(\bar X)\Big)\ket{\psi}\rp|
\leq b(1-\gamma)^{2k+1}
\qquad\qquad\text{for $\gamma\in(1-\epsilon,1)$,}
\ee
\noindent
where $\ket{\bar\phi}= \ket{\phi}+a\ket{\-1}$.
\label{lemma:oddpolybound}
\end{lemma}

\begin{proof}
The proof starts from Eq.~(\ref{eq:gamma01}). As before,
$\braket{-1}{\psi}=0$ implies $\bra{\psi}\lp(I+\gamma \bar
X\rp)^{-1}\ket{\bar \phi}=\bra{\psi}\lp(I+\gamma X\rp)^{-1} \ket{\phi}$
where $X$ is the restriction of $\bar X$ to $\HH$. The argument from the
proof of Lemma~\ref{lemma:vecadd} gives us an $(n-k)$-dimensional subspace
$\HHP\subset\HH$ of vectors $\ket{\phi}\in\HHP$ such that the numerator of
the rational function $\bra{\psi}\lp(I+\gamma X\rp)^{-1} \ket{\phi}$ has a
zero of order $k$ at $\gamma=1$. Given $\ket{\phi}\in\HHP$ we can find, by
Lemma~\ref{lemma:poly}, a small enough $\epsilon>0$ and large enough $a>0$
so that
\be
\frac{1-\gamma}{|a|^2}
\lp| \bra{\psi}\lp(I+\gamma X\rp)^{-1} \ket{\phi}\rp|^2
\leq b(\gamma-\gamma_0)^{2k+1}
\qquad\qquad\text{for $\gamma\in(1-\epsilon,1)$.}
\ee
\noindent
To complete the proof choose $\Theta$ to be the maximum of such choices of 
$|a|^2$ over the compact set of $\ket{\phi}\in\HHP$  that additionally satisfy
$\braket{\phi}{\phi}=1$.
\end{proof}

A similar result holds for $\bar Y$ with eigenvalues in $(-1,1]$ and
$c<0$. We now prove a special version of Lemma~\ref{lemma:ppq}.

\begin{lemma}
Let $p\neq q$ be functions with support in $(-1,1)$ such that
$p\rightarrow q$ is valid, $|S(p)|=|S(q)|$ and the number of zeros of
$p\rightarrow q$ is maximal with odd order at both $\gamma=-1$ and
$\gamma=1$. Then 
\begin{itemize}
\item There exists a vector $\ket{\bar \phi}=\ket{\phi}+a\ket{-1}$ and
constant $c>0$ such that $p'\rightarrow q$ is valid,\\ where
$p'=\Prob(\bar X',\ket{\psi})$, $\bar X'= \bar X+c\ket{\bar \phi}\bra{\bar
\phi}\leq I$, $p=\Prob(\bar X,\ket{\psi})$, the dimension of $\bar X$ is
$|S(p)|+1$ including a unique eigenvector $\ket{-1}$ with eigenvalue $-1$,
$\braket{\phi}{-1}=0$ and $\ket{\phi}\neq0$.
\item There exists a vector $\ket{\bar \phi}=\ket{\phi}+a\ket{1}$ and
constant $c<0$ such that $p\rightarrow q'$ is valid,\\ where
$q'=\Prob(\bar Y',\ket{\xi})$, $\bar Y'= \bar Y+c\ket{\bar \phi}\bra{\bar
\phi}\geq -I$, $q=\Prob(\bar Y,\ket{\xi})$, the dimension of $\bar Y$ is
$|S(q)|+1$ including a unique eigenvector $\ket{1}$ with eigenvalue $1$,
$\braket{\phi}{1}=0$ and $\ket{\phi}\neq0$.
\end{itemize}
\end{lemma}

\begin{proof}
We prove the first case, the second being nearly identical. As in the proof
of Lemma~\ref{lemma:ppq} the main goal is to ensure
\be
\bra{\psi}\Big(f_\gamma(\bar X+c\ket{\bar \phi}\bra{\bar \phi}) -
f_\gamma(\bar X)\Big)\ket{\psi}
\leq \sum_z \frac{z}{1+\gamma z} (q(z)- p(z)).
\qquad\text{for all $\gamma\in[-1,1].$}
\ee
By assumption, the number of zeros of the right-hand side is $2|S(p)|-2$,
with odd order at both $\gamma=-1$ and $\gamma=1$. The problem with the
original proof of Lemma~\ref{lemma:ppq} is that it would seek vectors
$\ket{\bar \phi}$ such that the left-hand side had an even number of zeros
at $\gamma=-1$ and $\gamma=1$ and the total number of zeros was $2|S(p)|$.
Therefore, the total number of linear constraints on $\ket{\bar\phi}$ would
be $|S(p)|$ and the only vector $\ket{\bar\phi}$ that satisfies all the
constraints is $\ket{-1}$.

However, Lemma~\ref{lemma:oddpolybound} allows us to satisfy the bound
while placing only an odd number of zeros at $\gamma=1$ (though still an
even number of zeros at $\gamma=-1$), and therefore requiring only
$|S(p)|-1$ constraints (so long as the coefficient of $\ket{-1}$ is large
enough). In particular, there exists a non-zero vector $\ket{\phi}$, a
large enough $a>0$ and small enough $c>0$ so that $\ket{\bar\phi}=
\ket{\phi}+a\ket{\-1}$ satisfies the above inequality in a neighborhood of
each of the zeros of the right-hand side.

The proof is then completed by picking a potentially smaller $c>0$ to
ensure that the inequality is satisfied outside the neighborhoods and 
that $\bar X+c\ket{\bar \phi}\bra{\bar\phi}\leq I$.
\end{proof}

\begin{lemma}
Let $p\neq q$ be functions with support in $(-1,1)$ such that $p\rightarrow
q$ is valid, $|S(p)|=|S(q)|$ and the number of zeros of $p\rightarrow q$ is
maximal with odd order at both $\gamma=-1$ and $\gamma=1$. Then
\begin{itemize}
\item There exists $p':(-1,1)\rightarrow[0,\infty)$ such that
$p'\neq p$, $|S(p')|=|S(p)|$,
$p\rightarrow p'$ is expressible by matrices and
$p'\rightarrow q$ is valid.
\item There exists $q':(-1,1)\rightarrow[0,\infty)$ such that
$q'\neq q$, $|S(q')|=|S(q)|$,
$q'\rightarrow q$ is expressible by matrices and
$p\rightarrow q'$ is valid.
\end{itemize}
\end{lemma}

\begin{proof}
We shall prove the first case. Take $p'$ from the previous lemma and
increase $c$ until either $p'\rightarrow q$ gets a new zero or until
$\bar X+c\ket{\bar \phi}\bra{\bar\phi}$ gets an eigenvalue of $1$.

In the first case, either we either end up with $p'=q$, or we have
$|S(p')|=|S(q)|+1$ and $p'\rightarrow q$ has maximal zeros, in which case
we can use Lemma~\ref{lemma:ppq} to create a second perturbation to get
$p''\rightarrow q$. Increasing this second $c>0$ cannot increase the number
of zeros or decrease the support size of the support of $p''$ unless we end
up with $p''=q$. In either case we have proven that $p\rightarrow q$ is
expressible by matrices and we can choose $p'=q$ to satisfy the lemma.

Alternatively if $\bar X+c\ket{\bar \phi}\bra{\bar\phi}$ gets an eigenvalue
of $1$, then we end up with $|S(p')|=|S(p)|$ and by construction
$p\rightarrow p'$ is expressible by matrices and $p'\rightarrow q$ is
valid. To prove that $p'\neq p$ we note that because $\ket{\bar\phi}$ is
not proportional to $\ket{-1}$ the maximum $c$ must be less than $2$, so
the trace of the canonical matrix expressing $p'$ is smaller than the
trace of the canonical matrix expressing $p$.
\end{proof}

\begin{lemma}
Let $p$ and $q$ be functions with support in $(-1,1)$ such that
$p\rightarrow q$ is valid, $|S(p)|=|S(q)|$ and the number of zeros of
$p\rightarrow q$ is maximal with odd order at both $\gamma=-1$ and
$\gamma=1$. Then $p\rightarrow q$ is expressible by matrices.
\end{lemma}

\begin{proof}
Let $p=\Prob(X,\ket{\psi})$ with $X$ a $2|S(p)|-1$ dimensional matrix which
includes $|S(p)|-1$ orthogonal eigenvectors with eigenvalue $-1$. We seek
the infimum of $\int_{-1}^1 \sum_z \frac{z}{1+\gamma z} (q(z)-p'(z))
d\gamma$ over $p'=\Prob(X',\ket{\psi})$ satisfying $p'\rightarrow q$ valid
and $X\leq X'\leq I$. Because we are optimizing $X'$ over a compact set the
infimum is achievable, and the resulting $p'$ will satisfy $p\rightarrow
p'$ is expressible by matrices and $p'\rightarrow q$ is valid.

If $p'=q$ the lemma is proven. Otherwise, by the preceding theorem there
exists $p''\neq p$ such that $p'\rightarrow p''$ is expressible by matrices
and $p''\rightarrow q$ is valid. Then $\int_{-1}^1 \sum_z \frac{z}{1+\gamma
z} (p''(z)-p'(z)) d\gamma>0$ and hence $\int_{-1}^1 \sum_z
\frac{z}{1+\gamma z} (q(z)-p'(z)) d\gamma>\int_{-1}^1 \sum_z
\frac{z}{1+\gamma z} (q(z)-p''(z)) d\gamma$. But also, by transitivity,
$p\rightarrow p''$ is expressible by matrices and in fact, because the
dimension of $X$ is large enough, we can write
$p''=\Prob(X'',\ket{\psi})$ for $X\leq X''\leq I$. That is a contradiction
with the optimality of $p'$.
\end{proof}

\begin{corollary}
If $p\rightarrow q$ is valid in $[-1,1]$ then it is expressible by matrices
in $[-1,1]$.
\end{corollary}

The results used in Section~\ref{sec:compiling} can be proven as follows:
Lemma~\ref{lemma:f2mmain} follows from the above corollary,
Lemmas~\ref{lemma:inf2lamb} and~\ref{lemma:lamb21}, and the definition of
expressible by matrices. Lemma~\ref{lemma:stdformmain} follows from
Lemma~\ref{lemma:stdmattrans}.


\end{document}